\documentclass[11pt]{article}

\usepackage[square,numbers,sort&compress]{natbib}

\usepackage{tablefootnote}
\usepackage[utf8]{inputenc} 
\usepackage[T1]{fontenc}    
\usepackage{hyperref}       
\usepackage{url}            
\usepackage{booktabs}       
\usepackage{amsfonts}       
\usepackage{nicefrac}       
\usepackage{microtype}      
\usepackage{tikz}
\usetikzlibrary{decorations.pathreplacing,calc}

\usepackage{verbatim}
\usepackage{amsthm}
\usepackage{amsfonts}
\usepackage[nonamelimits]{amsmath}
\usepackage{amssymb}
\usepackage{sansmath}
\usepackage{color}
\usepackage{xcolor}
\usepackage{bm}
\usepackage{bbm}
\usepackage{fullpage}
\usepackage[bottom]{footmisc}

\theoremstyle{plain}
\newtheorem{theorem}{Theorem}[section]

\newtheorem{lemma}[theorem]{Lemma}

\newtheorem{corollary}[theorem]{Corollary}

\theoremstyle{definition}
\newtheorem{definition}[theorem]{Definition}

\usepackage{subfig}
\usepackage{graphicx}
\usepackage{listings}
\usepackage{enumitem}
\usepackage{tabularx}
\newtheoremstyle{named}{}{}{\itshape}{}{\bfseries}{.}{.5em}{\thmnote{#3}}
\theoremstyle{named}

\usepackage{tikz}
\usetikzlibrary{fit,calc}

\usepackage{tikz}
\usetikzlibrary{cd}

\DeclareSymbolFont{extraup}{U}{zavm}{m}{n}
\DeclareMathSymbol{\varheart}{\mathalpha}{extraup}{86}
\DeclareMathSymbol{\vardiamond}{\mathalpha}{extraup}{87}

\newcommand{\R}{\mathbb{R}}
\newcommand{\eps}{\varepsilon}
\renewcommand{\epsilon}{\varepsilon}
\newcommand{\calC}{\mathcal{C}}
\newcommand{\calO}{\mathcal{O}}
\newcommand{\calP}{\mathcal{P}}
\newcommand{\calX}{\mathcal{X}}
\newcommand{\calY}{\mathcal{Y}}
\newcommand{\calK}{\mathcal{K}}
\newcommand{\calT}{\mathcal{T}}
\newcommand{\calM}{\mathcal{M}}
\newcommand{\calV}{\mathcal{V}}

\newcommand{\bsx}{\boldsymbol{x}}
\newcommand{\bsv}{\boldsymbol{v}}
\newcommand{\bsy}{\boldsymbol{y}}
\newcommand{\bsp}{\boldsymbol{p}}
\newcommand{\bsq}{\boldsymbol{q}}
\newcommand{\bsc}{\boldsymbol{c}}

\newcommand{\bsX}{\boldsymbol{X}}
\newcommand{\bsY}{\boldsymbol{Y}}
\newcommand{\bsP}{\boldsymbol{P}}
\newcommand{\bsQ}{\boldsymbol{Q}}
\newcommand{\bsC}{\boldsymbol{C}}
\newcommand{\bsU}{\boldsymbol{U}}
\newcommand{\bsS}{\boldsymbol{S}}
\newcommand{\bsI}{\boldsymbol{I}}

\newcommand{\dist}{\operatorname{dist}}
\newcommand{\spc}{\operatorname{sc}}
\newcommand{\row}{\operatorname{row}}
\newcommand{\round}{\operatorname{round}}
\newcommand{\ent}{\operatorname{ent}}
\newcommand{\de}{\operatorname{d}}
\newcommand{\vol}{\operatorname{vol}}
\newcommand{\h}{\operatorname{h}}
\newcommand{\embed}{\operatorname{embedsc}}
\newcommand{\w}{\operatorname{w}}
\newcommand{\expo}{\operatorname{expo}}
\newcommand{\fraction}{\operatorname{fraction}}
\newcommand{\CC}{\operatorname{CC}}

\renewcommand{\eqref}[1]{(\ref{#1})}
\newcommand{\cost}{\ensuremath{\mathrm{cost}}}

\newcommand{\ProblemName}[1]{\textsc{#1}}
\newcommand{\kzC}{\ProblemName{$(k, z)$-Clustering}}
\newcommand{\kMedian}{\ProblemName{$k$-Median}}
\newcommand{\kMeans}{\ProblemName{$k$-Means}}

\usepackage{threeparttable}

\usepackage{todonotes}

\usepackage{xspace}

\xspaceaddexceptions{[]\{\}}

\usepackage{hyperref}
\usepackage{cleveref}

\usepackage{thm-restate}

\usepackage{algpseudocode}
\usepackage[Algorithm]{algorithm}
\algnewcommand\algorithmicforeach{\textbf{for each}}
\algdef{S}[FOR]{ForEach}[1]{\algorithmicforeach\ #1\ \algorithmicdo}

\newcommand{\eat}[1]{}

\makeatletter
\newcommand*{\rom}[1]{\expandafter\@slowromancap\romannumeral #1@}

\usepackage{multirow}

\usepackage{tablefootnote}
\usepackage{float}
\usepackage{makecell}

\allowdisplaybreaks

\title{Space Complexity of Euclidean Clustering}

\author{
Xiaoyi Zhu\thanks{{\tt zhuxy22@m.fudan.edu.cn}. School of Data Science, Fudan University, China.}
\and
Yuxiang Tian\thanks{{\tt tianyx22@m.fudan.edu.cn}. School of Data Science, Fudan University, China.}
\and
Lingxiao Huang\thanks{{\tt huanglingxiao1990@126.com}. State Key Laboratory of Novel Software Technology, Nanjing University, China.}\ \footnotemark[5]
\and
Zengfeng Huang\thanks{{\tt huangzf@fudan.edu.cn}. School of Data Science, Fudan University, China.}\ \thanks{Corresponding Author.}
}

\date{}
\date{\vspace{-1ex}}

\begin{document}

\maketitle

\begin{abstract}
The \kzC\ problem in Euclidean space $\mathbb{R}^d$ has been extensively studied.
Given the scale of data involved, compression methods for the Euclidean \kzC\ problem, such as data compression and dimension reduction, have received significant attention in the literature.
However, the space complexity of the clustering problem, specifically, the number of bits required to compress the cost function within a multiplicative error $\varepsilon$, remains unclear in existing literature.
This paper initiates the study of space complexity for Euclidean \kzC\ and offers both upper and lower bounds. 
Our space bounds are nearly tight when $k$ is constant, indicating that storing a coreset, a well-known data compression approach, serves as the optimal compression scheme. 
Furthermore, our lower bound result for \kzC\ establishes a tight space bound of $\Theta( n d )$ for terminal embedding, where $n$ represents the dataset size. 
Our technical approach leverages new geometric insights for principal angles and discrepancy methods, which may hold independent interest.
\end{abstract}


\newpage
\tableofcontents
\newpage

\section{Introduction}
\label{sec:introduction}

Clustering problems are fundamental in theoretical computer science and machine learning with various applications~\cite{arthur2007k,coates2012learning,lloyd1982least}. 
An important class of clustering is called Euclidean \kzC\ where, given a dataset $\bsP \subseteq \mathbb{R}^d$ of $n$ points and a $k\geq 1$, the goal is to find a center set $\bsC\subseteq \R^d$ of $k$ points that minimizes the cost 
$\operatorname{cost}_z(\bsP, \bsC):=\sum_{\bsp \in \bsP} \dist^z(\bsp, \bsC)$.
Here, $\dist^z(\bsp, \bsC) = \min _{\bsc \in \bsC} \dist^z(\bsp, \bsc)$ is the $z$-th power Euclidean distance of $\bsp$ to $\bsC$. 
Well-known examples of \kzC\ include \kMedian\ (when $z = 1$) and \kMeans\ (when $z = 2$).

In many real-world scenarios, the dataset $\bsP$ is large and the dimension $d$ is high, and it is desirable to compress $\bsP$ to reduce storage and computational requirements in order to solve the underlying clustering problem efficiently. 
Previous studies have proposed two approaches: data compression and dimension reduction.
On one hand, coresets have been proposed as a solution to data compression~\cite{har2004coresets} -- a coreset is a small representative subset $\bsS$ that approximately preserves the clustering cost for all possible center sets. 
Recent research has focused on developing efficient coresets~\cite{sohler2018strong,cohen2021new,cohen2022towards,cohenaddad2022improved,huang2023optimal}, showing that the coreset size remains independent of both the size $n$ of dataset and the dimension $d$.
On the other hand, dimension reduction methods have also proven to be effective for \kzC, including techniques like Johnson-Lindenstrauss (JL)~\cite{makarychev2019performance,charikar2022johnson} and terminal embedding~\cite{narayanan2019optimal,huang2020coresets}.
Specifically, terminal embedding (\Cref{def:embedding}), which projects a dataset $\bsP$ to a low-dimensional space while approximately preserving all pairwise distances between $\bsP$ and $\R^d$, is the key for removing the size dependence on $d$ for coreset \cite{huang2020coresets,cohen2021new}.

While the importance of compression for clustering has been widely acknowledged, the literature currently lacks clarity regarding the space complexity of the clustering problem itself.
Specifically, one may want to know how many bits are required to compress the cost function. 
Space complexity, a fundamental factor in theoretical computer science, serves as a measure of the complexity of the cost function.
Previous research has investigated the space complexity for various other problems, including approximate nearest neighbor~\cite{indyk2018approximate}, inner products~\cite{alon2017optimal}, Euclidean metric compression~\cite{indyk2022optimal}, and graph cuts~\cite{carlson2019optimal}.

To investigate the space complexity of the \kzC\ problem, one initial approach is to utilize a coreset $\bsS$, which yields a space requirement of $\tilde{O}(|\bsS|\cdot d)$ using standard quantization methods (see Theorem~\ref{thm:main_upper} in the paper). Here the $d$ factor arises from preserving all coordinates of each point in the coreset $\bsS$. 
One might wonder if it is possible to combine the benefits of coreset construction and dimension reduction to eliminate the dependence on the dimension $d$ in terms of space requirements.
This leads to a natural question: ``Is it possible to obtain an $|\bsS|\cdot o(d)$ bound? Additionally, is coreset the most efficient compression scheme for the \kzC\ problem?''
Perhaps surprisingly, we show that $\tilde{\Omega}(|\bsS|\cdot d)$ is necessary for interesting parameter regimes (see \Cref{thm:main_lower} in the paper). 
This means a quantized coreset is optimal, and dimensionality reduction does not help with space complexity. 
The proof of the lower bound for space complexity is our main contribution, which encounters more technical challenges.
Unlike upper bounds, existing lower bounds for coresets do not directly translate into lower bounds for space complexity since compression approaches can go beyond simply storing a subset of points as a coreset.
Overall, the study of space complexity is intricately connected to the optimality of coresets and poses technical difficulties.

\subsection{Problem Definition and Our Results}
\label{sec:result}

In this paper, we initiate the study of the space complexity for the Euclidean \kzC\ problem.
We first formally define the notion of space complexity.
Assume that $\bsP\subseteq [\Delta]^d$ for some integer $\Delta\geq 1$, i.e., every $\bsp\in \bsP$ is a grid point in $[\Delta]^d = \left\{1,2,\ldots, \Delta\right\}^d$.
This assumption is standard in the literature, e.g., for clustering~\cite{braverman2017clustering,Hu2018NearlyOD}, facility location~\cite{czumaj2013approximation}, minimum spanning tree~\cite{frahling2008sampling}, and the max-cut problem~\cite{chen2023streaming}, and necessary for analyzing the space complexity. 
\footnote{We need such an assumption to ensure that the precision of every coordinate of $\bsp\in \bsP$ is bounded. Otherwise, when $\bsP$ contains a unique point $\bsp\in \R^d$, we need to maintain all coordinates of $\bsp$ such that the information of $\cost_z(\bsP,\{\bsp\}) = 0$ is preserved. Then if the precision of $\bsp$ can be arbitrarily large, the space complexity is unlimited. }
Let $\mathcal{C}$ denote the collection of all $k$-center sets in $\R^d$, i.e. $\mathcal{C}:=\left\{\bsC \subseteq \mathbb{R}^d:|\bsC|=k\right\}$. 
An $\eps$-sketch for $\bsP$ is a data structure $\calO$ that given any center set $\bsC\in \mathcal{C}$, returns a value $\calO(\bsC)\in (1\pm \eps)\cdot \cost_z(\bsP,\bsC)$ which
recovers the value $\cost_z(\bsP, \bsC)$ up to a multiplicative error of $\eps$. 
We give the following notion.

\begin{definition}[Space complexity for Euclidean \kzC]
\label{def:space}
We are given a dataset $\bsP\subseteq [\Delta]^d$,
integers $n,k\geq 1$, constant $z\geq 1$ and an error parameter $\eps\in (0,1)$.
We define $\spc(\bsP,\Delta, k, z, d, \eps)$ to be the minimum possible number of bits of an $\eps$-sketch for $\bsP$.
Moreover, we define $\spc(n, \Delta, k, z, d, \eps):= \sup_{\bsP\subseteq [\Delta]^d: |\bsP| = n} \spc(\bsP,\Delta, k, z, d, \eps)$ to be the space complexity function, i.e., the maximum cardinality $\spc(\bsP,\Delta, k, z, d, \eps)$ over all possible datasets $\bsP\subseteq [\Delta]^d$ of size at most $n$.
\end{definition}

{
The upper bound and lower bound to the space complexity for Euclidean $\kzC$ are summarized in Table~\ref{tab:upper_lower}.
}

\begin{table}[h]
\centering
\caption{{A summary of our results. The bounds are tight when $k=O(1), d=\Omega\left(1/\varepsilon^2\right)$ and $n= \Omega\left(1/\varepsilon^2\right)$}}\label{tab:upper_lower}
\begin{tabular}{|c|c|c|c|}
\hline
Range &  $n\leq k$ & $n > k$ ($k\geq 2$ and $\Delta = \Omega\left(\frac{k^{\frac{1}{d}}\sqrt{d}}{\eps}\right)$ for lower bound)  \\\hline
    Upper Bound     & $O(nd\log\Delta)$        & $\tilde{O}\left(kd\log\Delta+k\log\log n+d\cdot \min\left\{\frac{k^{\frac{2z+2}{z+2}}}{\eps^2}, \frac{k}{\eps^{z+2}} \right\}\right)$ \\\hline
    Lower Bound    &   $\Omega(nd\log\Delta)$     & $\Omega\left(kd\log \Delta + k \log \log \frac{n}{k} + kd\min\left\{\frac{1}{\eps^2},\frac{d}{\log d}, \frac{n}{k}\right\}\right)$ \\\hline
\end{tabular}
\end{table}

\noindent\textbf{Space upper bounds.} Our first contribution is to provide upper bounds for the space complexity $\spc(n, \Delta, k, z, d, \eps)$.
We apply the idea of storing an $\eps$-coreset and have the following theorem.
Here, an $\eps$-coreset for \kzC\ is a subset $\bsS\subseteq \bsP$ together with a weight function $\w : \bsS\rightarrow \R_{\geq 0}$ such that for every $\bsC\in \calC$, $\sum_{\bsp \in \bsS} \w(\bsp) \cdot \dist^z(\bsp, \bsC) \in(1 \pm \eps) \cdot \cost_z(\bsP, \bsC)$.

\begin{theorem}[Space upper bounds] 
\label{thm:main_upper}
Suppose for any dataset $\bsP\subseteq [\Delta]^d$ of size $n$, there exists an $\eps$-coreset of $\bsP$ for \kzC\ of size at most $\Psi(n) \geq 1$.
We have the following space upper bounds:
 \begin{itemize}
     \item When $n \leq k$, $\spc(n, \Delta, k, z, d, \eps)\leq O\left(n d\log \Delta \right)$;
     \item When $n > k$, $\spc(n, \Delta, k, z, d, \eps) \leq O(kd\log\Delta + \Psi(n)(d\log 1/\eps + d\log\log \Delta + \log\log n ))$.
 \end{itemize}
\end{theorem}

The proof of this theorem can be found in Section~\ref{sec:upper_bound}.
Fully storing a coreset $\bsS$ requires $\Psi(n)\cdot d \log \Delta$ bits for points and $\Psi(n)\cdot \log n$ for its weight function $\w()$.
To further reduce the storage space, we provide a quantization scheme for the weight function $\w()$ and points in $\bsS$ (Algorithm~\ref{alg:upper}).
When ignoring the logarithmic term, we have $sc(n, \Delta, k, z, d, \eps)\leq \Tilde{O}(\Psi(n)\cdot d)$.~\footnote{In this paper, $\Tilde{O}(\cdot)$ may hide a factor of $2^{O(z)}$ and the logarithmic term of the input parameters $n, \Delta, k, d, 1/\eps$.}
Combining with the recent breakthroughs that shows that $\Psi(n) = \tilde{O}\left(\min\left\{k^{\frac{2z+2}{z+2}} \eps^{-2}, k \eps^{-z-2}\right\}\right)$ \cite{cohen2021new,cohen2022towards,cohenaddad2022improved,huang2023optimal}, we conclude that when $n > k$,
\begin{align}
\label{eq:upper_bound}
sc(n, \Delta, k, z, d, \eps) \leq \tilde{O}\left(d\cdot \min\left\{\frac{k^{\frac{2z+2}{z+2}}}{\eps^2}, \frac{k}{\eps^{z+2}} \right\}\right).
\end{align}

\noindent\textbf{Space lower bounds.}
Our main contribution is to provide the lower bounds for the space complexity $\spc(n, \Delta, k, z, d, \eps)$.

\begin{theorem}[Space lower bounds] 
\label{thm:main_lower} 
We have the following space lower bounds:
\begin{itemize}
    \item When $n \leq k$, $\spc(n, \Delta, k,z,d, \eps)\geq \Omega(n d \log \Delta)$;
    \sloppy
    \item When $n > k\geq 2$ and $\Delta = \Omega\left(\frac{k^{\frac{1}{d}}\sqrt{d}}{\eps}\right)$, $$\spc(n, \Delta, k,z,d, \eps) \geq \Omega\left(kd\log \Delta + kd\min\left\{\frac{1}{\eps^2},\frac{d}{\log d}, \frac{n}{k}\right\} + k \log \log \frac{n}{k}\right).$$
\end{itemize}
\end{theorem}

The proof of this theorem can be found in Section~\ref{sec:lower_bound}.
Compared to Theorem~\ref{thm:main_upper}, our lower bound for space complexity is tight when $n \leq k$.
For the case when $n > k$, the key term in our lower bound is $\Omega(kd\min\left\{\frac{1}{\eps^2},\frac{d}{\log d}, \frac{n}{k}\right\})$.
Comparing this with Inequality~\eqref{eq:upper_bound}, we can conclude that the optimal space complexity $\spc(n, \Delta, k,z,d, \eps)= \Theta\left( \frac{d}{\eps^2}\right)$ when $k = O(1)$, $n\geq \Omega(\frac{1}{\eps^2})$ and $d\geq \Omega(\frac{1}{\eps^2 \log 1/\eps})$.
As a corollary, we can affirm that the coreset method is indeed the optimal compression method when the size and dimension of the dataset $\bsP$ are large and the number of centers $k$ is constant.
It would be interesting to further investigate whether the coreset method remains optimal for large $k$.
Another corollary of Theorem~\ref{thm:main_upper} is a lower bound $\Omega(\frac{k}{\eps^2})$ for the coreset size $\Psi(n)$.
This bound matches the previous result in \cite{cohen2022towards}, and it has been recently improved to $\Omega(\frac{k}{\eps^{-z-2}})$ when $\eps = \Omega(k^{\frac{1}{z+2}})$ \cite{huang2023optimal}.
Since the technical approach is different, our methods for space lower bounds may also be useful to further improve the coreset lower bounds.

It is worth noting that $d$ still appears in our lower bound results, which implies that exploiting dimension reduction techniques does not necessarily lead to a reduction in storage space. 
Although this may seem counter-intuitive, it is reasonable since we still need to maintain the mapping from the original space to the embedded space (which is also the space consumed by the dimensionality reduction itself), and the storage of this mapping could also be relatively large. 
Moreover, we can utilize this fact to lower bound the space cost of these dimension reduction methods from our results; see the following applications.

\vspace{0.1in}
\noindent \textbf{Application 1: Tight space lower bound for terminal embedding.}
Our \Cref{thm:main_lower} also yields an interesting by-product: a nearly tight lower bound for the space complexity of terminal embedding, which is a well-known dimension reduction method recently introduced by \cite{elkin2017terminal,narayanan2019optimal}. It is a pre-processing step to map input data to a low-dimensional space.
The definition of it is given as follows. 
\begin{definition}[\bf{Terminal embedding}]
	\label{def:embedding}
Let $\eps\in (0,1)$ and $\bsP$ be a dataset of $n$ points. 
	A mapping $\tau: \R^d\rightarrow \R^m$ is called an $\eps$-terminal embedding of $\bsP$ 
 if for any $\bsp\in \bsP$ and $\bsq\in \R^d$,
	$
	\dist(\bsp,\bsq)\leq \dist(\tau(\bsp),\tau(\bsq))\leq (1+\eps)\cdot \dist(\bsp,\bsq).
	$
\end{definition}

As a consequence of \Cref{thm:main_lower}, the preservation of the terminal embedding function $\tau$ must incur a large space cost; summarized by the following theorem.
The result is obtained by another natural idea for sketch construction: maintaining a terminal embedding function $\tau$ for a coreset $\bsS$ and the projection $\tau(\bsS)$ of a coreset $\bsS$, in which the storage space for $\tau(\bsS)$ can be independent on the dimension $d$.

\begin{theorem}[Informal; see \Cref{thm:embedding}] 
\label{thm:embedding_intro}
Let $\eps\in (0,1)$ and assume $d = \Omega\left(\frac{\log n \log(n/\eps) }{\eps^2 }\right)$. 
An $\eps$-terminal embedding, that projects a given dataset $\bsP\subseteq \R^d$ of size $n$ to a target dimension $O(\frac{\log n}{\eps^2})$, requires space at least $\Omega(n d)$.
\end{theorem}

The bound $\Omega( n d )$ is not surprising since terminal embedding can be used to approximately recover the original dataset.
In the case when $d\geq \Omega(\frac{\log n \log(n/\eps)}{\eps^2})$, our result improves upon the previous lower bound of $\Omega(\frac{n\log n}{\eps^2})$ from \cite{alon2017optimal}. 
\footnote{Although the paper does not directly study terminal embedding, their bound for preserving inner products (Theorem 1.1 in \cite{alon2017optimal}) implies a lower bound of $\Omega(\frac{n\log n}{\eps^2})$ for terminal embedding.}
We replace their factor of $\frac{\log n}{\eps^2}$ with $d$.
Furthermore, our lower bound of $\Omega(nd)$ matches the prior upper bound of $\tilde{O}(nd)$ for terminal embedding \cite{cherapanamjeri2022terminal}, making it nearly tight.

Recently, \cite{huang2022near,huang2023coresets} proposed a coreset of size $O(m)+\tilde{O}(k^2\eps^{-2z-2})$ for this problem.

\vspace{0.1in}
\noindent \textbf{Application 2: Compression scheme for coreset construction in distributed and streaming settings.} In the era of big data, the size of datasets has grown dramatically, which presents significant challenges for analysis. 
Over the past decade, new computation models such as the distributed model and the streaming model have emerged as effective approaches for handling large-scale data. 

Extensive research has been conducted on constructing coresets in both distributed setting~\cite{balcan2013distributed} and streaming setting~\cite{har2004coresets,braverman2019streaming,cohen2023streaming}. 
These studies, similar to offline coreset construction, mainly focus on the size of the coreset without considering the specific space complexity.
The quantization scheme proposed in Algorithm~\ref{alg:upper} is flexible and can be applied to any algorithm based on the coreset method.
Therefore, by leveraging similar ideas, we can also derive algorithms with satisfactory bit complexity upper bounds in these scenarios.

In the distributed setting (see Definitions~\ref{def:distributed}), there is a set of $l$ sites $\mathcal{V}$ each holding a local data set. 
These sites communicate through an undirected connected graph $\mathcal{G}$, where an edge indicates that two sites can communicate with each other.
Our goal is to construct an $\eps$-sketch for the whole dataset on a specified site while minimizing the number of bits required for communication.
By constructing a sketch for each site using our compression method in Algorithm~\ref{alg:upper} and then transmitting the sketch to the coordinator, we obtain the communication cost of $\tilde{O}\left(ld\cdot \min\left\{\frac{k^{\frac{2z+2}{z+2}}}{\eps^2}, \frac{k}{\eps^{z+2}} \right\}\right)$ where $l$ denotes the number of sites. 
The results are summarized in Corollary~\ref{cor:distributed}.

In the streaming setting (see Definitions~\ref{def:streaming}), the input
data arrive sequentially and we require a data structure to maintain an aggregate of the points seen so far to facilitate computation of the objective function. 
Our goal is to maintain the data structure using as few bits as possible.  
Using our compression scheme in Algorithm~\ref{alg:upper}, we obtain the bit complexity for \kzC\ problem in the streaming setting summarized in Corollary~\ref{cor:stream}. \color{black}

\subsection{Technical Overview}
\label{sec:overview}
We now describe the high-level technical ideas behind our main contribution Theorem \Cref{thm:main_lower}. 
In general, our approach involves using a clever counting argument to establish lower bounds on space. 
We do this by creating a large family of datasets $\calP$ where, for any pair $\bsP$ and $\bsQ$ from this family, there exists a center set $\bsC$ that separates their cost function by a significant margin, denoted as $\cost_{z}\left(\bsP, \bsC\right) \notin \left(1\pm O(\varepsilon)\right) \cost_{z}\left(\bsQ, \bsC\right)$. 
This difference in cost implies that $\bsP$ and $\bsQ$ can not share the same sketch, which leads to a lower bound on space of $\log\left(\left|\calP\right|\right)$ (Lemma~\ref{lmm:setsize_to_bit}).
Hence, we focus on how to construct such a family $\calP$.

We discuss the most technical bound, which is $\Omega\left(kd\min\left\{\frac{1}{\eps^2},\frac{d}{\log d}, \frac{n}{k}\right\}\right)$, when $n > k\geq 2$ and $\Delta = \Omega\left(\frac{k^{\frac{1}{d}}\sqrt{d}}{\eps}\right)$. 
The proofs for other bounds are pretty standard. 
For brevity, we will explain the technical idea for the case of $z=k=2$ (\ProblemName{2-Means}). 
The extension to general $z$ and $k$ is straightforward, by analyzing Taylor expansions for $(1+x)^z$ (Section~\ref{sec:generalize_power_high}) and make $\Omega(k)$ copies of datasets in $\calP$ (Section~\ref{sec:generalize_k}).
Our construction of $\calP$ relies on a fundamental geometric concept known as \emph{principal angles} (Definition~\ref{def:principal_angles}). 
The Cosine of these angles, when given the orthonormal bases $\bsP=\left\{\bsp_i:i\in [n]\right\}$ and $\bsQ=\left\{\bsq_i:i\in [n]\right\}$ of two distinct subspaces in $\mathbb{R}^d$, uniquely correspond to the singular values of $\bsP^\top \bsQ$ (Lemma~\ref{lmm:property_pa}). 
This correspondence essentially measures how orthogonal the two subspaces are to each other.
With principal angles in mind, we outline the two main components of our proof. 
Assuming that $d>n$, the first component (Lemma~\ref{lmm:large_cost}) demonstrates that if the largest $O(n)$ principal angles between two orthonormal bases $\bsP$ and $\bsQ$ are sufficiently large, there exists a center set $\bsC = \left\{\bsc,-\bsc\right\}\in \calC$ with $\|\bsc\|_2=1$ such that $\operatorname{cost}_2(\bsP, \{\bsc,-\bsc\}) -\operatorname{cost}_2(\bsQ, \{\bsc,-\bsc\}) \geq \Omega\left(\sqrt{n}\right)$. 
This induced error of $\Omega\left(\sqrt{n}\right)$ from $C$ achieves the desired scale of $\eps\cdot \cost_z(\bsP,\bsC) = O(\varepsilon n)$ when $n=O\left(\frac{1}{\varepsilon^2}\right)$.
The second component (Lemma~\ref{lmm:large_size}) states that when $n = O\left(\frac{d}{\log d}\right)$, there exists a large family $\calP$ of orthonormal bases (for different $n$-dimensional subspaces) with size $\exp\left(nd\right)$ such that most principal angles of any two different orthonormal bases in the family are sufficiently large. 
The space lower bound $\Omega\left(d\min\left\{\frac{1}{\eps^2},\frac{d}{\log d}, n\right\}\right)$ directly follows from these two lemmas.

Next, we delve into the technical insights behind Lemmas \ref{lmm:large_cost} and \ref{lmm:large_size}.

\vspace{0.1in}
\noindent \textbf{\Cref{lmm:large_cost}: Reduction from principal angles to cost difference.}
Recall that we aim to show the existence of a center set $\bsC= \left\{\bsc,-\bsc\right\}$ that incurs a large cost difference between two orthonormal bases $\bsP=\left\{\bsp_i:i\in [n]\right\}$ and $\bsQ=\left\{\bsq_i:i\in [n]\right\}$.
By the formulation of $\bsC$, we note that $\cost_2(\bsP,\bsC) - \cost_2(\bsQ,\bsC) = 2 \left( \sum_{i=1}^n \left | \langle  \bsq_{i},\bsc\rangle\right| - \left | \langle  \bsp_{i},\bsc\rangle\right|\right)$.
Hence, we focus on showing the existence of a unit vector $\bsc\in \R^d$ such that 
\begin{align}
\label{eq:desired_diff}
\sum_{i=1}^n \left | \langle  \bsq_{i},\bsc\rangle\right| - \left | \langle  \bsp_{i},\bsc\rangle\right|\geq \Omega(\sqrt{n}).
\end{align}
Intuitively, our goal is to increase the magnitude of the first term $\sum_{i=1}^n \left | \langle  \bsq_{i},\bsc\rangle\right|$ while decreasing the magnitude of the second term $\sum_{i=1}^n \left | \langle  \bsp_{i},\bsc\rangle\right|$.
One initial approach is to choose $\bsc = \frac{1}{\sqrt{n}} \bsQ \boldsymbol{\zeta} = \frac{1}{\sqrt{n}} \sum_{i\in [n]} \boldsymbol{\zeta}_i \bsq_i$, where $\boldsymbol{\zeta}\in \left\{-1, +1\right\}^n$. 
By this selection, center $\bsc$ lies on the subspace spanned by $\bsQ$ and maximizes the first term $\sum_{i=1}^n \left | \langle  \bsq_{i},\bsc\rangle\right|$ to be $\sqrt{n}$.
{
Moreover, the second term becomes $\sum_{i=1}^n \left | \langle  \bsp_{i},\bsc\rangle\right| = \frac{1}{\sqrt{n}}\|\bsP^\top \bsQ \boldsymbol{\zeta}\|_1 \leq \sqrt{n} \|\bsP^\top \bsQ \boldsymbol{\zeta}\|_\infty$ and we want to minimize it.
This objective is very similar to the goal of coloring.
Informally speaking, the goal of coloring is to find a vector $\boldsymbol{\zeta}\in \left\{-1, +1\right\}^n$ for a given matrix $\bsU$ that minimizes $\|\bsU \boldsymbol{\zeta}\|_\infty$ (See Definition~\ref{def:partial_coloring}).
}
Ideally, if we can find a coloring $\boldsymbol{\zeta}\in \left\{-1, +1\right\}^n$ such that $\|\bsP^\top \bsQ \boldsymbol{\zeta}\|_\infty \leq 0.5$, we can achieve the desired cost difference in Inequality~\eqref{eq:desired_diff}.
However, the existence of such $\boldsymbol{\zeta}$ appears to be non-trivial.
For instance, if we randomly select a coloring $\boldsymbol{\zeta}$ from $\left\{-1, +1\right\}^n$, the expected value of $\|\bsP^\top \bsQ \boldsymbol{\zeta}\|_\infty$ can be as large as $O(\log n)$ \cite{spencer1985six}.
{
On the other hand, directly applying proofs from discrepancy literature (e.g., \cite{spencer1985six,dadush2022new}) does not achieve the desired property $\|\bsP^\top \bsQ \boldsymbol{\zeta}\|_\infty \leq 0.5$. 
This is because existing techniques work for arbitrary matrices $U$ instead of the specific matrix $P^\top Q$ that may have additional geometric properties, and hence, only yield unsatisfactory results.
}

{
To bypass this technical difficulty, we enhance the previous idea by allowing $\boldsymbol{\zeta} \in \left\{-1,0,+1\right\}^n$ to be a partial coloring with $\|\boldsymbol{\zeta}\|_1 \geq 0.75n$, which means $\boldsymbol{\zeta}$ can now have at most 25\% entries that are zero.
With this modification, we have $\sum_{i=1}^n \left | \langle  \bsq_{i},\bsc\rangle\right| = 0.75 \sqrt{n}$. 
Thus, it still suffices to bound $\|\bsP^\top \bsQ \boldsymbol{\zeta}\|_\infty \leq 0.5$ such that Inequality~\eqref{eq:desired_diff} holds.
Such a stricter bound calls for new ideas.  
}

{
Our core objective is to find the conditions on $\bsP$ and $\bsQ$ that allow us to identify such a partial coloring.
Let's consider two simple examples to illustrate the idea. 
When $\bsP$ and $\bsQ$ are identical, we would have $\bsP^\top \bsQ$ is the identity matrix. 
For any coloring vector $\boldsymbol{\zeta}$ with $\|\boldsymbol{\zeta}\|_0 > 0$, we must have at least one entry of $|\bsp^\top \bsQ \zeta|$ is 1 and thus $\|\bsP^\top \bsQ \boldsymbol{\zeta}\|_\infty = 1$.
When $\bsP$ and $\bsQ$ are orthogonal, we would have $\bsP^\top \bsQ$ is the zero matrix.
For any coloring vector $\boldsymbol{\zeta}$, we must have $\bsp^\top \bsQ \zeta = \boldsymbol{0} $ and thus $\|\bsP^\top \bsQ \boldsymbol{\zeta}\|_\infty = 0$.
This suggests that the greater the difference between P and Q, the easier it is to find a partial coloring that meets our requirements.
We will show that such differences can be characterized using principal angles.
}

{
After closely examining the value of the partial coloring, we find that this value is closely related to the norm of each row $\|(\bsP^\top \bsQ)_i\|_2$.
Using random coloring as an example, applying the Chernoff bound, we would find that the magnitude of each corresponding value for a row is bounded by the norm of that row, i.e. $\|(\bsP^\top \bsQ \boldsymbol{\zeta})_i\|_2 \leq  a \|(\bsP^\top \bsQ)_i\|_2$ for some constant $a$. 
Therefore, to achieve a smaller partial coloring, we require the row norms of $\bsP^\top \bsQ$ to be relatively small.
Since $\bsP$ and $\bsQ$ are two orthonormal bases, the row norms of $\bsP^\top \bsQ$ correspond to the length of the projection of $\bsp_i$ to the subspace of $\bsQ$.
For example, when $\bsp_i$ lies in the subspace of $\bsQ$, we have $\|(\bsP^\top \bsQ)_i\|_2 = 1$. 
On the other side, when $\bsp_i$ is orthogonal to the subspace of $\bsQ$, we have $\|(\bsP^\top \bsQ)_i\|_2 = 0$. 
Our aim is to find $\bsP$ and $\bsQ$ such that the length of the projection of each data point $\bsp_i$ to the subspace of $\bsQ$ is relatively small. 
}

{
Note that in the simplified two-dimensional space cases where $\bsP$ and $\bsQ$ are reduced to a single data point, a small projection length from $\bsP$ to $\bsQ$ is equivalent to having a large angle between $\bsP$ and $\bsQ$.
Based on this idea, we find a similar pattern for high-dimensional subspaces with the help of the notion ``principal angles''. 
The formal definition of principal angles can be found in Definition~\ref{def:principal_angles}.
Intuitively, principal angles are a set of minimized angles between the two subspaces.
Small principal angles indicate that the two subspaces are nearly parallel in many directions, and the length of projection to these directions would be high. 
For example, when $\bsP$ and $\bsQ$ are identical, the principal angles between them are all 0.
$\bsP^\top \bsQ$ equals the identity matrix and the row norms of it are all 1.
On the other hand, large principal angles imply that
the two subspaces span many directions that are nearly orthogonal to each other.
For example, when $\bsP$ and $\bsQ$ are orthogonal, the principal angles between them are all $\frac{\pi}{2}$.
$\bsP^\top \bsQ$ equals the zero matrix and the row norms of it are all 0.
}

{
Therefore, large principal angles imply that the majority of the row norms $\|(\bsP^\top \bsQ)_i\|_2$ are small (Lemma~\ref{lmm:small_sum_of_square}).
These small row sums enable us to find a partial coloring $\boldsymbol{\zeta}$ that further reduces the bound for $\|\bsP^\top \bsQ \boldsymbol{\zeta}\|_\infty$ to 0.5 (Lemma \ref{lmm:discrepancy}), employing similar approaches as in \cite{spencer1985six}.
In summary, we have completed the proof of Lemma~\ref{lmm:large_cost}.
}

\vspace{0.1in}
\noindent \textbf{\Cref{lmm:large_size}: Construction of $\calP$.}
Our construction is inspired by a geometric observation made in Absil et al. \cite{absil2006largest}, which states that the largest principal angle between the orthonormal bases $\bsP$ and $\bsQ$ of two $n$-dimensional subspaces, independently drawn from the uniform distribution on the Grassmann manifold of $n$-planes in $\mathbb{R}^d$, is at least $\Omega(1)$ with high probability.
We extend this result and prove that even the largest $O(n)$ principal angles between $\bsP$ and $\bsQ$ are at least $\Omega(1)$ (Lemma~\ref{lmm:principal_angles}). 
This extension relies on a more careful integral calculation for the density function of principal angles.
Moreover, this extension leads to an enhanced geometric observation: on average, these two orthonormal bases $\bsP$ and $\bsQ$ are distinct with respect to principal angles, which could be of independent research interest.
Then using standard probabilistic arguments, we can randomly select a family $\mathcal{P}$ of $\exp(\Omega(nd))$ orthonormal bases, ensuring that the largest $O(n)$ principal angles between any pair $\bsP$ and $\bsQ$ from $\mathcal{P}$ are consistently large.

\subsection{Other Related Work}
\label{sec:related}

\noindent\textbf{Coreset construction for clustering.} There are a series of works towards closing the upper and lower bounds of coreset size for \kzC\ in high dimensional Euclidean spaces \cite{feldman2011unified,braverman2016new,cohen2021new,cohen2022towards,cohenaddad2022improved,huang2023optimal}.
The current best upper bound is $\Tilde{O}(\min\{\frac{k^{\frac{2z+2}{z+2}}}{\eps^2},\frac{k}{\eps^{z+2}}\})$ by \cite{cohen2021new,cohen2022towards,cohenaddad2022improved,huang2023optimal}.
{Specifically, Cohen-Addad et al.~\cite{cohen2022towards} got an upper bound of $\tilde{O}\left(k \varepsilon^{-2} \cdot \min \left(\varepsilon^{-z}, k\right)\right)$ and Huang et al.~\cite{huang2023optimal} got an upper bound of $\tilde{O}\left(k^{\frac{2 z+2}{z+2}} \varepsilon^{-2}\right)$.}
On the other hand, Huang and Vishnoi~\cite{huang2020coresets} proved a size lower bound $\Omega(k\min\{2^{z/20},d\})$ and Cohen-addad et al.~\cite{cohen2022towards} showed bound $\Omega(k \eps^{-2})$.
Very recently, Huang and Li~\cite{huang2023optimal} gave a size lower bound of $\Omega(k \eps^{-z-2})$ for $\eps=\Omega(k^{-1/(z+2)})$, which matches the size upper bound and is nearly tight.
There have also been studies for the coreset size when the dimension is small, see e.g. \cite{har2004coresets,huang2023coresets}.
In addition to offline settings, coresets have also been studied in the stream setting~\cite{har2004coresets,braverman2019streaming,cohen2023streaming}, distributed setting~\cite{balcan2013distributed} and dynamic setting~\cite{henzinger2020fully}. {It is worth noting that the existing literature all assumes that we can store vectors with infinite precision and thus focuses primarily on the size of the coreset. This simplification makes it impossible for us to determine the exact space complexity when using these algorithms\cite{cohen2021new,cohen2022towards,cohenaddad2022improved,huang2023optimal}. Our paper addresses this issue by designing a quantization scheme for weights and points. Meanwhile, these papers only obtained lower bounds on the number of points used by the coreset method, whereas we focus on the space complexity that any algorithm might require and provide the lower bound. }

\vspace{0.1in}
\noindent\textbf{Dimension reduction.}
Dimension reduction is an important technique for data compression, including techniques like Johnson-Lindenstrauss (JL)~\cite{makarychev2019performance,charikar2022johnson} and terminal embedding~\cite{narayanan2019optimal,huang2020coresets}. 
The target dimension of any embedding satisfying the JL lemma is shown to be $\Theta(\eps^{-2}\log n)$ \cite{johnson1986extensions,alon2017optimal,larsen2017optimality}, where $n$ is the size of the data set.
The space complexity of JL is shown to be $O(\log d+\log (1 / \delta)(\log \log (1 / \delta)+\log (1 / \varepsilon)))$ random bits \cite{kane2011almost}, where $\eps$ and $\delta$ are error and fail probability respectively.
In the context of clustering, Makarychev et al.~\cite{makarychev2019performance} give a nearly optimal target dimension $O(\log(k/\eps)/\eps^2)$ for \kzC\ by applying JL. 
Their reduction ensures that the cost of the optimal clustering is preserved within a factor of $(1 + \eps)$ instead of preserving the clustering cost for all center sets. 
For terminal embedding, Narayanan and Nelson~\cite{narayanan2019optimal} provided an optimal terminal embedding with target dimension $O(\eps^{-2}\log n)$.
For the space complexity, the best-known construction of terminal embedding costs $\Tilde{O}(nd)$ bits~\cite{cherapanamjeri2022terminal}.

\vspace{0.1in}
\noindent\textbf{Space complexity.} 
{
Space complexity is receiving increasing attention in the era of big data. 
Various problems have been studied by previous research.
For example, Carlson et al.\cite{carlson2019optimal} shows that approximately storing the sizes of all cuts in an undirected graph on $n$ vertices up to a $(1\pm \varepsilon)$ error requires $\Omega\left(\frac{n\log n}{\varepsilon^2}\right)$ bits. 
Recently, Dexter et al.\cite{dexterspace} consider the problem of approximating logistic loss. 
They prove that the lower bound of space complexity is $\Omega\left(\frac{d}{\varepsilon^2}\right)$ when the complexity of the problem is constant and existing coreset constructions are optimal up to logarithmic factors in this regime. 
Their technique is based on the reduction to ReLu regression and the INDEX problem in communication complexity, which is completely different from ours.
}

\section{Preliminaries}\label{sec:preliminary}

{
Before we start our proof, we will first fix some notations. In the following chapters, we will use lowercase letters to denote scalars, such as $x$; lowercase boldface to represent vectors, such as $\boldsymbol{p}$; and uppercase boldface to denote matrices.
$\bsI_q$ is denoted as a $q\times q$ identity matrix. 
For convenience, we slightly abuse the notation by also using uppercase boldface to denote datasets, since a dataset with $n$ points in a $d$-dimensional space can be represented as a $d \times n$ matrix. Calligraphic capital letters will be used to denote sets {other than datasets}, such as $\mathcal{P}$, and upright font will be used to represent functions, such as $\operatorname{dist}(\cdot, \cdot)$. Table \ref{tab:notation} summarizes some {frequently used} notations in this paper. 
}
\begin{table}[ht]
\centering
\caption{Notations used in this paper.}\label{tab:notation}
\resizebox{\linewidth}{!}{\begin{tabular}{|c|c|c|c|}
\hline
Notation & Definition & Notation & Definition \\\hline
    $k$     & the number of cluster centers         & $z$  & the power parameter for the distance function           \\\hline
      $d$   &  dimension           &  $n$       & the size of the dataset           \\\hline
    $\Delta$     &  the parameter for the discretization of the space   &   
    $\varepsilon$ & the error parameter for the estimation of the cost function
         \\\hline
    $\sigma_i, i\in [n]$      &  the $i$-th singular value of the matrix    &  $\theta_i, i\in [n]$      &  the $i$-th principal angle
         \\\hline  $\boldsymbol{\zeta}$      &  coloring vector &
$\bsp_i,\bsq_i$ & data points in the dataset 
         
        \\\hline $\boldsymbol{P},\boldsymbol{Q}$      &  datasets with $n$ points        
     &   $\boldsymbol{C}$      & center set with $k$ points  
         \\\hline
         $\boldsymbol{S}$      &  coreset  &

    $\bsI_{n} \in \mathbb{R}^{n\times n}$ &  the identity matrix 
         \\\hline
 $\boldsymbol{U}$      & the inner product matrix $\boldsymbol{U}=\boldsymbol{P}^{T}\boldsymbol{Q}$  &

 $\mathcal{P}$      &  a large family of datasets    
          
         \\\hline
          $\mathcal{C}$      & the collection of all center sets with $k$ points  &

    $\operatorname{sc}\left(\right)$      &  the space complexity function of \kzC
         \\\hline

    $\operatorname{cost}\left(\right) $     & the cost function for clustering &

    $\operatorname{ent}\left(\right)$      &  the entropy function  
         \\\hline

    $\w\left(\right) $ & the weight function for coreset  &
     $\mathcal{O}\left(\right)$ & the $\varepsilon$-sketch for dataset 
    \\\hline

    $\expo\left(\right) $ & the encoding function for the exponent part  &
     $\fraction\left(\right)$ & the encoding function for the significand part 
    \\\hline

    $\Psi\left( \right)$      &  the size function of the coreset & $\operatorname{dens}\left(\right)$  & the density function
    \\\hline
     $\tau\left(\right)$     & the mapping function for the terminal embedding & 
    $\det\left(\right)$ & the determinant of a matrix 
    \\\hline

    $\embed\left(\right)$ & the space complexity function of terminal embedding
    &  $\CC\left(\right)$ &  \makecell[c]{the communication complexity function of \\distributed \kzC} \\\hline
    
\end{tabular}}
\end{table}

{
Next, we give a brief prelude to the tools used in our proof. 
In the proof of space upper bound in Section~\ref{sec:upper_bound}, we need the following lemma to bound the distance between two points.
\begin{lemma}[Relaxed triangle inequality (Lemma 10 of \cite{cohen2022towards})]
    \label{lmm:triangle}
    Let $\bsp_1, \bsp_2, \bsp_3$ be arbitrary points in a metric space with distance function $\dist()$, and let $z$ be a positive integer. Then for any $\eps>0,$
    \begin{align*}
         &\dist^z(\bsp_1, \bsp_2) \leq (1+\eps)^{z-1} \dist^z(\bsp_1, \bsp_3)+\left(\frac{1+\eps}{\eps}\right)^{z-1} \dist^z(\bsp_2, \bsp_3), \\
         &\left|\dist^z(\bsp_1, \bsp_2)-\dist^z(\bsp_1, \bsp_3)\right| \leq \eps \cdot \dist^z(\bsp_1, \bsp_3)+\left(\frac{z+\eps}{\eps}\right)^{z-1} \dist^z(\bsp_2, \bsp_3).
    \end{align*}
\end{lemma} }

{
The proof of space lower bound in Section~\ref{sec:lower_bound} relies on two key concepts.
The first one is partial coloring, which is commonly found in the discrepancy literature. 
}

{
\begin{definition}[Partial Coloring]\label{def:partial_coloring}
    Let $\bsU$ be a matrix in $\mathbb{R}^{n\times n}$. The goal of a partial coloring is to find a vector ${\boldsymbol{\zeta}}\in \{-1,0,1\}^n$ such that 
    \begin{enumerate}
        \item The number of zero entries $|i:{\boldsymbol{\zeta}}_i=0|\leq \frac{1}{4}n$;
        \item  The discrepancy, i.e. maximum norm $\left\|  \bsU{\boldsymbol{\zeta}}\right\|_\infty$ is as small as possible.
    \end{enumerate}
\end{definition}
}

{The second concept is called the principal angles, which characterize the relative positions of two subspaces. }
{
\begin{definition}[Principal angles]
\label{def:principal_angles}
Suppose $n\leq d$.
Given two $n$-dimensional subspaces $\mathcal{X}$ and $\mathcal{Y}$ of $\mathbb{R}^{d}$, there exists then a sequence of angles called the principal angles (or canonical angles)  $0\leq \theta_{1}(\mathcal{X}, \mathcal{Y}),\cdots, \theta_{n}(\mathcal{X}, \mathcal{Y})  \leq \frac{\pi}{2}$. The first one is defined as
\begin{align*}
\theta_{1}\left(\mathcal{X}, \mathcal{Y}\right) := 
\min & \left\{  \arccos \left(|\bsx^\top \bsy| \right) \mid \bsx \in \mathcal{X}, \bsy \in \mathcal{Y},\left\|\bsx\right\|_2=1,\left\|\bsy\right\|_2=1  \right\} = \angle\left(\bsx_1, \bsy_1\right),
\end{align*}

where the vectors $\bsx_1$ and $\bsy_1$ are the corresponding principal vectors. The other principal angles and vectors are then defined recursively via
\begin{align*}
\theta_{i}(\mathcal{X}, \mathcal{Y}) := 
\min &\left\{ \arccos\left(\left.|\bsx^\top \bsy| \right) \right\rvert\, \bsx \in \mathcal{X}, \bsy \in \mathcal{Y},\left\|\bsx\right\|_2=1,\left\|\bsy\right\|_2=1, \bsx \perp \bsx_j, \bsy \perp \bsy_j, \right.\\
&\left. \forall j \in\{1, \ldots, i-1\}\right\}.
\end{align*}

\end{definition}
Slightly abusing the notation, we use $\theta_i$ for brevity when the corresponding two subspaces are clear from the context. 
The notion of principal angles between subspaces was first introduced by Jordan~\cite{jordan1875essai} and has many important applications in statistics and numerical analysis \cite{dahlquist1968comparison,varah1970computing}. 
Intuitively, we can see that the principal vectors in each subspace form an orthonormal basis and the principal angles  $(\theta _{1}, \cdots, \theta _{k})$ are a set of minimized angles between the two subspaces. 
Small principal angles indicate that the two subspaces are nearly parallel in many directions, while large principal angles imply that 
the two subspaces are more distinct and span many directions that are nearly orthogonal to each other. 
For example, when $\mathcal{X}\perp \mathcal{Y}$, all principal angles $\theta_i  = \frac{\pi}{2}$. 
}

{
As another example, we consider two distinct planes in $\mathbb{R}^3$ (i.e., two-dimensional subspaces) intersect along a line shown in Figure~\ref{fig:principal_angle}. By the definition, we will choose $\bsx_1=\bsy_1$ on the intersection line and thus $\theta_1 = 0$. We then have $\bsx_2$ and $\bsy_2$ as the orthogonal directions to the intersection line on each of the two planes respectively. The angle between them is the second principal angle $\theta_2 = \theta$.
}

\begin{figure}[ht]
\centering\includegraphics[width=0.4\linewidth]{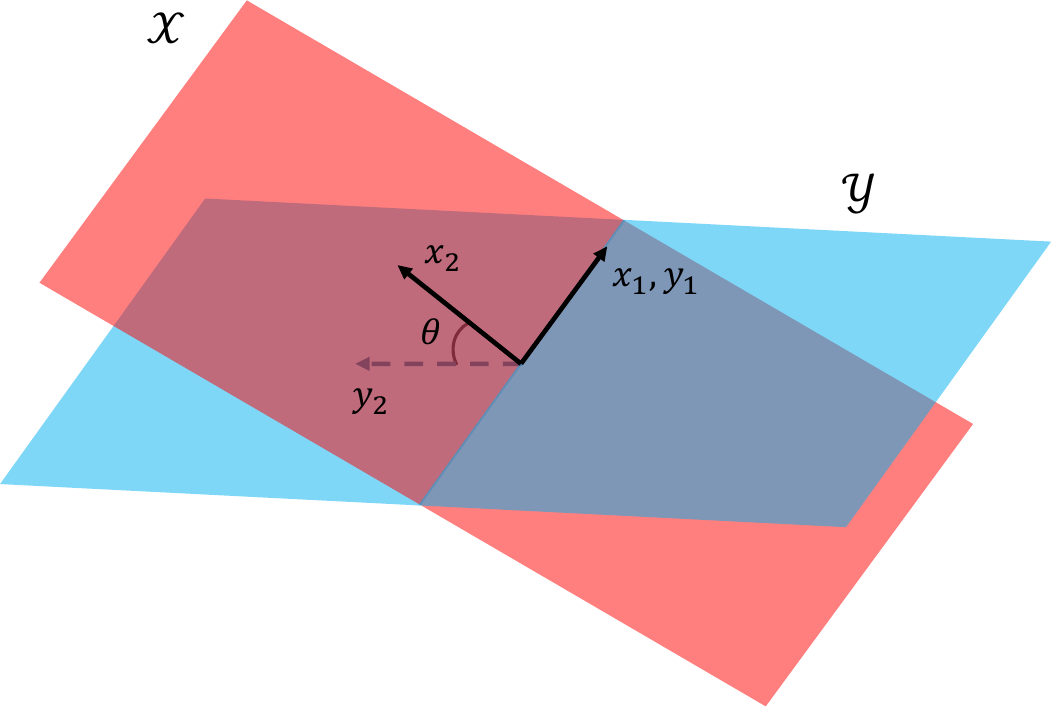}
  \caption{Example of principal angles of two distinct planes in $\mathbb{R}^3$ sharing a line .}\label{fig:principal_angle}
\end{figure}

The following lemma shows a relation between principal angles and singular value decomposition.

\begin{lemma}[Property of principal angles (Theorem 1 in~\cite{bjorck1973numerical})] \label{lmm:property_pa} 
Given two $n$-dimensional subspaces $\mathcal{X}$ and $\mathcal{Y}$ of $\mathbb{R}^{d}$, let the columns of matrices $\bsX \in \mathbb{R}^{d\times n}$ and $\bsY \in \mathbb{R}^{d\times n}$ form orthonormal bases for the subspaces $\calX$ and $\calY$ respectively. 
Denote $1\geq \sigma_1 \geq \cdots \geq \sigma_n$ to be the singular values of the inner product matrix $\bsX^{\top}\bsY$. 
We have 
$\sigma_{i} = \cos \left(\theta_{i}\right), \forall i\in [n].$
\end{lemma}

Note that Lemma~\ref{lmm:property_pa} holds for any orthonormal basis $\bsX$ and $\bsY$ of corresponding subspaces and the values of principal angles are independent of the choices of them. 
For any orthogonal matrix $\boldsymbol{A},\boldsymbol{B} \in \mathbb{R}^{n\times n}$, it is easy to find that $\sigma\left(\boldsymbol{A}^{\top}\bsX^{\top}\bsY \boldsymbol{B}\right) = \sigma\left(\bsX^{\top}\bsY\right)$. 

Using Figure~\ref{fig:principal_angle} as an example, we would have $\bsX = [\bsx_1,\bsx_2]$ and $\bsY = [\bsy_1,\bsy_2]$. By calculation, we would have 
\begin{align*}
\bsX^\top \bsY &= \begin{bmatrix} \bsx_1^T
 \\
\bsx_2^T
\end{bmatrix} [\bsy_1,\bsy_2]  = \begin{bmatrix}
\bsx_1^T \bsy_1  &\bsx_1^T \bsy_2 \\
\bsx_2^T \bsy_1 & \bsx_2^T \bsy_2
\end{bmatrix} = \begin{bmatrix}
\cos(0) &0 \\
0 & \cos(\theta)
\end{bmatrix}.
\end{align*}

\section{Proof of Theorem~\ref{thm:main_upper}: Space Upper Bounds}
\label{sec:upper_bound}

The proof for the first part when $n\leq k$ is to simply store all data points.
Since $\bsP\subseteq [\Delta]^d$, the storage space for each coordinate is at most $\log \Delta$, which results in the space upper bound $O(nd \log \Delta)$.
Next, we focus on the second part when $n > k$.
The main idea is to construct a sketch to store a coreset using space as small as possible.

Let $\bsP\subseteq [\Delta]^d$ be a dataset of size $n> k$.
Let $(\bsS,\w)$ be an $\frac{\eps}{5}$-coreset of $\bsP$ for \kzC.
Let $\bsC^\star$ be an $O(1)$-approximation of optimal center set, that is, a center set satisfying $\cost_z(\bsP, \bsC^\star) \leq O(1)\cdot \min_{\bsC\in \mathcal{C}} \cost_z(\bsP, \bsC).$
We argue that by rounding each $\bsc^*\in \bsC^\star$ to the nearest point in $\bsP$, i.e. $\bsp_{\bsc^*}:=\operatorname{argmin}_{\bsp\in \bsP}\dist(\bsc^*,\bsp)$, it remains the property of $O(1)$-approximation.
To this end, by Lemma~\ref{lmm:triangle}, for any $\bsp\in \bsP,$ $|\dist^z(\bsp,\bsc^*)-\dist^z(\bsp,\bsp_{\bsc^*})| \leq \dist^z(\bsp,\bsc^*) + (1+z)^{z-1}\dist^z(\bsp_{\bsc^*},\bsc^*)  \leq (1+(1+z)^{z-1})\dist^z(\bsp,\bsc^*).$
As $z$ is constant, the claim is proved.
Without of loss of generality, we assume $\bsC^\star \subseteq \bsP$ and $\cost_z(\bsP, \bsC^\star) \leq 2 \min_{\bsC\in \mathcal{C}} \cost_z(\bsP, \bsC).$

The compression scheme is summarized in Algorithm~\ref{alg:upper}.
Intuitively, we need to compress the weight $\w(\bsp)$ and each coordinate $i\in [d]$ of points $\bsp\in \bsP$. { Here we use a base-2 floating-point format. For the exponent, we use an encoding function $\expo(\cdot,\cdot),$ where the first argument is a data point and the second argument is either $\w$ or an integer $i$, which stands for weight or its $i$'th coordinate. For the significand, we use an encoding function $\fraction(\cdot,\cdot,\cdot)$, where the first two arguments are similar to $\expo()$'s and the last argument is the precision parameter.}
For $\w(\bsp)$, we either safely ignore too small weight, i.e., $\w(\bsp)\leq \frac{\eps}{4|\bsS|}$, or {remain its exponent by $\expo(\bsp, \w)$ and the first $\lceil \log 4/\eps \rceil$ significant digits by $\fraction(\bsp,\w,\eps)$}.
{For each $\bsp$, we denote $\bsc_{\bsp}$ to be the closest center of $\bsp$ in $\bsC$, and $\bsc^*_{\bsp}$ to be the closest center of $\bsp$ in $\bsC^*$.} Then compress each coordinate of $\bsp-\bsc^*_{\bsp}$ by a similar idea as for $\w(\bsp)$. {We use the notation $\bsc^*_{l}$ to denote the $l$'th point in $\bsC^\star$, and use $\bsp[i]$ to denote $i$'th coordinate of $\bsp$}.

\begin{algorithm}[H]
    \caption{A compression scheme based on coreset}
    \label{alg:upper}
    \begin{algorithmic}[0]
    \State \textbf{Input}: Error parameter $\eps\in (0,1)$, an $\frac{\eps}{5}$-coreset $\bsS\subseteq \bsP$ of size $|\bsS|\leq \Psi(n)$ together with a weight function $\w:\bsS\rightarrow \R_{\ge 0}$, an $2$-approximate center set $\bsC^\star\subseteq \bsP$ of $\bsP$ for \kzC 
   \State \textbf{Output}: A sketch $\calO$ of $\bsP$ for \kzC 
   \State Partition $\bsS$ into\\ $\bsS_l := \left\{ \bsp\in \bsS |\bsc^*_l = \operatorname{argmin}_{\bsc^*\in \bsC^\star} \dist(\bsp,\bsc^*)\right\}, l \in [k]$;
        \For{$\bsc^*_l\in \bsC^\star$}
            \For{$\bsp\in \bsS_{l}$}
                \If{$\w(\bsp)\leq \frac{\eps}{4|\bsS|}$} 
                     \State $(\fraction(\bsp,\w,\eps),\expo(\bsp, \w))\leftarrow (0,0)$;
                \Else {}
                    $\fraction(\bsp,\w,\eps)\leftarrow \frac{\w(\bsp)}{2^{\lfloor \log \w(\bsp) \rfloor}}$, rounding to $\lceil \log 4/\eps \rceil$ decimal places; $\expo(\bsp, \w)\leftarrow \lfloor \log \w(\bsp) \rfloor$;
                \EndIf
                \ForEach{coordinate $i\in [d]$}
                    \If{$\bsp[i]-\bsc^*_{l}[i] = 0$}
                        \State $(\fraction(\bsp,i,\eps),\expo(\bsp,i)) \leftarrow (0,0)$;
                    \Else
                        {} $\fraction(\bsp,i,\eps)\leftarrow \frac{\bsp[i]-\bsc^*_{l}[i]}{2^{\lfloor \log (\bsp[i]-\bsc^*_{l}[i]) \rfloor}}$, rounding to $\lceil \log 4 z/\eps \rceil$ decimal places; \State $\expo(\bsp,i)\leftarrow\lfloor \log (\bsp[i]-\bsc^*_{l}[i]) \rfloor$;
                    \EndIf
                \EndFor
                \State $\calO_l \leftarrow \cup_{\bsp\in \bsS_l} (\{\fraction(\bsp,i,\eps),\expo(\bsp, i)\}_{i=1}^d)\cup \cup_{\bsp\in \bsS_l}(\fraction(\bsp,\w,\eps),\expo(\bsp, \w))$;
            \EndFor
        \EndFor
        \\ \Return $\calO \leftarrow \cup_{l\in [k]}(\bsc^*_l,\calO_l)$
    \end{algorithmic}
\end{algorithm}

Since $\bsS$ is a coreset, we will make use of the following lemma.

\begin{lemma}[Sum of weights]
\label{lmm:sum}
Given dataset $\bsP\subseteq [\Delta]^d$ of size $n$ and suppose $\eps\in (0,0.5)$. An $\eps$-coreset of $\bsP$ for \kzC\ satisfies that $\sum_{\bsp\in \bsS} \w(\bsp) \in (1\pm 4\eps) n$.
\end{lemma}

\begin{proof}
    
    Consider a center set $\bsC=\{\bsc,\ldots,\bsc\}$ such that for all $\bsp_1, \bsp_2 \in \bsP$, $$\dist^z(\bsp_1,\bsC) \in \left(\left(1\pm 0.2\eps\right)\dist^z(\bsp_2,\bsC)\right),$$ which could be obtained by choosing $\bsC$ far away from $\bsP$. Then we have $\frac{\sum_{\bsp\in \bsP} \dist^z(\bsp,\bsC)}{\sum_{\bsp \in \bsS} \w(\bsp)\dist^z(\bsp,\bsC)}$ is bounded by
    \begin{align*}
         & \left(\frac{\left(1- 0.2\eps\right)n}{\left(1+ 0.2\eps\right)\sum_{\bsp\in \bsS} \w(\bsp)}, \frac{\left(1+ 0.2\eps\right)n}{\left(1- 0.2\eps\right)\sum_{\bsp\in \bsS} \w(\bsp)}\right) \in  \left(1\pm \eps\right) \frac{n}{\sum_{\bsp\in \bsS} \w(\bsp)}.
    \end{align*}
    By the coreset definition, we have $ \frac{\cost_z(\bsP,\bsC)}{\sum_{\bsp \in \bsS} \w(\bsp)\dist^z(\bsp,\bsC)} = \frac{\sum_{\bsp\in \bsP} \dist^z(\bsp,\bsC)}{\sum_{\bsp \in \bsS} \w(\bsp)\dist^z(\bsp,\bsC)}\in 1\pm \eps$, then 
    \begin{align*}
        \sum_{\bsp\in \bsS} \w(\bsp) & \in \left(\frac{1-\eps}{1+\eps}, \frac{1+\eps}{1-\eps}\right) n \in (1\pm 4\eps) n.
    \end{align*}
\end{proof}\color{black}

Now we are ready to prove the second part of Theorem~\ref{thm:main_upper}.

\begin{proof}[Proof of Theorem~\ref{thm:main_upper} (second part)]
    \noindent \textbf{Correctness analysis.} 
    We first show Algorithm~\ref{alg:upper} indeed outputs an $\eps$-sketch of $\bsP$ for \kzC\ function. 
    We use $\calO$ to obtain $\widehat{\w}(\bsp) = \fraction(\bsp,\w,\eps)\cdot 2^{\expo(\bsp, \w)}$ and $ \widehat{\bsp}=\bsp_0 + \bsc^*_{\bsp}$ with $ \bsp_0[i] = \fraction(\bsp,i,\eps)\cdot 2^{\expo(\bsp, i)}$ for $i\in [d]$.
    Given a center set $\bsC\in\mathcal{C},$ we approximate \kzC\ function by the following value: 
    \begin{equation*}
        \sum_{\bsp\in \bsS} \widehat{\w}(\bsp) \cdot \dist^z(\widehat{\bsp},\bsC).
    \end{equation*}
    We claim that for each $\bsp\in \bsS$, $\widehat{\w}(\bsp) \in (1\pm \frac{\eps}{4}) \w(\bsp)$ when $\w(\bsp)> \frac{\eps}{4|\bsS|}.$
    This is because $ \frac{\w(\bsp)}{2} \leq 2^{\expo(\bsp, \w)}\leq \w(\bsp)$ and $|\fraction(\bsp,\w,\eps)-\frac{\w(\bsp)}{2^{\expo(\bsp, \w)}}|\leq \frac{\eps}{4}$, which implies that $\fraction(\bsp,\w,\eps) \in (1\pm \frac{\eps}{4})\frac{\w(\bsp)}{2^{\expo(\bsp, \w)}}$ and $\widehat{\w}(\bsp) \in (1\pm \frac{\eps}{4})\w(\bsp)$. 
    When $\w(\bsp)\leq \frac{\eps}{4|\bsS|},$ we have $\w(\bsp)\dist^z(\bsp,\bsC)\leq \frac{\eps}{4|\bsS|} \dist^z(\bsp,\bsC) \leq \frac{\eps}{4|\bsS|}\sum_{\bsp\in \bsP}\dist^z(\bsp,\bsC)$, which means this quantity is too small to affect the \kzC\ function and we could set all such $\w(\bsp)$ to zero.

    Next, we analyze $\widehat{\bsp}$.
    By Lemma~\ref{lmm:triangle}, for any $\bsc \in \bsC,$  $|\dist^z(\bsp,\bsc)-\dist^z(\widehat{\bsp},\bsc)|$ is upper bounded by $\frac{\eps}{4} \dist^z(\bsp,\bsc) + (1+\frac{4z}{\eps})^{z-1}\dist^z(\bsp,\widehat{\bsp}).$
    By our construction, $\dist(\bsp,\widehat{\bsp}) = \dist(\bsp-\bsc^*_{\bsp},\bsp_0) = \sqrt{\sum_{i=1}^d (\bsp[i]-\bsc^*_{\bsp}[i]-\bsp_0[i])^2},$ and as the same argument for weight, $\bsp[i]-\bsc^*_{\bsp}[i]-\bsp_0[i] \leq 
    \frac{\eps}{4 z} (\bsp[i]-\bsc^*_{\bsp}[i]),$ thus $\dist(\bsp,\widehat{\bsp}) \leq \frac{\eps}{4z} \dist(\bsp,\bsc^*_{\bsp}).$

    Putting the above results together,
\begin{align*}
    \dist^z(\widehat{\bsp},\bsC) 
      & \in
        \dist^z(\bsp,\bsc_{\bsp}) \pm\left(\frac{\eps}{4}\dist^z(\bsp,\bsc_{\bsp})
        + \left(1+\frac{4z}{\eps}\right)^{z-1}\dist^z(\bsp,\widehat{\bsp})\right)\\
        & \in 
        \left(1\pm \frac{\eps}{4}\right)\dist^z(\bsp,\bsc_{\bsp}) 
        \pm \frac{\eps}{4z}\left(1+\frac{\eps}{4z}\right)^{z-1}\dist^z(\bsp,\bsc^*_{\bsp}) \\
        & \in 
        \left(1\pm \frac{\eps}{4}\right) \dist^z(\bsp,\bsc_{\bsp}) \pm \frac{\eps}{8} \dist^z(\bsp,\bsc^*_{\bsp})
        \\
        & \in 
        \left(1 \pm \frac{\eps}{2} \right)\dist^z(\bsp,\bsc_{\bsp}),
    \end{align*}
    and thus we have that
    \begin{align*}
        \sum_{\bsp\in \bsS} \widehat{\w}(\bsp)\cdot 
    \dist^z(\widehat{\bsp},\bsC) 
        & \in 
        \left(1\pm \frac{\eps}{4}\right)\left(1 \pm \frac{\eps}{2} \right)\left(1\pm \frac{\eps}{5}\right) 
        \cdot\sum_{\bsp\in \bsP}\dist^z(\bsp,\bsc_{\bsp})\\
        &\in 
        (1\pm\eps)\sum_{\bsp\in \bsP}\dist^z(\bsp,\bsc_{\bsp}),
    \end{align*}
    where the third line follows from $\ln (1+\frac{\eps}{4z}) \leq \frac{\eps}{4z}$, the fourth line follows from $\bsC^\star$ is a $2$-approximation of an optimal center set, and the penultimate line follows from the construction of coreset.
    Therefore we construct an $\eps$-sketch for $\bsP$.

    \vspace{0.1in}
    \noindent \textbf{Space complexity analysis.} 
    We analyze its space complexity from now on.
    The storage for $k$ grid points $\bsC^\star$ is $O(kd\log \Delta).$
    We store each weight by a set $(\fraction(\bsp,\w,\eps),\expo(\bsp, \w))$, where the first number is up to $O(\lceil \log 4/\eps \rceil)$ decimal places, and representing the integer number $\expo(\bsp, \w) = \lfloor \log \w(\bsp) \rfloor$ requires $O( \log\max\{\log\frac{4|\bsS|}{\eps},\log n \})$ bits by Lemma~\ref{lmm:sum}.
    Similarly, the storage for each $\bsp$ is $O(d\log 4z/\eps + d\log\log \Delta)$ bits.
    Combining them and notice that $|\bsS|\leq n,$ we obtain the final bound
    \begin{align*}
         &sc(P,\Delta, k, z, d, \eps)
       \\& \leq 
        O\left( kd\log\Delta + |\bsS|\left( \log\frac{4}{\eps} +  \log\max\left\{\log\frac{4|\bsS|}{\eps},\log n \right\} + d\log \frac{4z}{\eps} + d\log\log \Delta\right) \right) \\
        & = 
        O\left(kd\log\Delta +\Psi(n)(d\log 1/\eps  d\log\log \Delta + \log\log n)\right),
    \end{align*}
    where we ignore the dependence on $z.$
\end{proof}
\section{Proof of Theorem~\ref{thm:main_lower}: Space Lower Bounds}
\label{sec:lower_bound}

In this section, we prove the space lower bounds. 
The high-level idea is to construct a large family of datasets such that any two of them can not use the same sketch; summarized by the following lemma.
{
The intuition behind this lemma is that: if two datasets $\bsP$ and $\bsQ$ yield similar results under any clustering center set, our sketch does not need to allocate additional space to distinguish between them.
On the other hand, if they produce significantly different results under some clustering center set, our function must retain this information; otherwise, it would produce incorrect results on one of the datasets.
}

\begin{lemma}[A family of datasets leads to space lower bounds] 
\label{lmm:setsize_to_bit} 
Suppose there exists a family $\mathcal{P}$ of datasets of size $n\geq 1$ such that for any two datasets $\bsP, \bsQ\in \mathcal{P}$, there exists a center set $\bsC\in \calC$ with $\operatorname{cost}_z(\bsP, \bsC) \notin \left(1\pm 3\varepsilon\right)\operatorname{cost}_z(\bsQ, \bsC)$.
Then we have $\spc(n, \Delta, k,z,d, \varepsilon) \geq \Omega(\log \left|\mathcal{P}\right|)$. 
\end{lemma}

\begin{proof}
We prove this by contradiction. 
Assume that $\spc(n, \Delta, k,z,d, \varepsilon) = o(\log\left|\mathcal{P}\right|)$, we must be able to find two datasets $\bsP$ and $\bsQ$ such that they correspond to the same $\eps$-sketch $\calO$. 
Since $\calO$ is an $\eps$-sketch for both $\bsP$ and $\bsQ$, we have for every center set $\bsC\in \calC$,
\begin{align*}
& \calO(\bsC) \in (1\pm \eps)\cdot \operatorname{cost}_z(\bsP, \bsC) , ~\calO(\bsC) \in (1\pm \eps)\cdot \operatorname{cost}_z(\bsQ, \bsC),
\\& \operatorname{cost}_z(\bsP, \bsC) \leq \frac{1}{1-\eps}\calO(\bsC) \leq \frac{1+\eps}{1-\eps} \operatorname{cost}_z(\bsQ, \bsC) \leq (1+3\eps)\operatorname{cost}_z(\bsQ, \bsC),
\\& \operatorname{cost}_z(\bsP, \bsC) \geq \frac{1}{1+\eps}\calO(\bsC) \geq \frac{1-\eps}{1+\eps} \operatorname{cost}_z(\bsQ, \bsC) \geq (1-3\eps)\operatorname{cost}_z(\bsQ, \bsC).
\end{align*}
This contradicts with our assumption that $\operatorname{cost}_z(\bsP, \bsC) \notin \left(1\pm 3\varepsilon\right)\operatorname{cost}_z(\bsQ, \bsC)$.
\end{proof}

\subsection{Proof of Theorem~\ref{thm:main_lower}}
\label{sec:lower_bound_high}

We first prove the lower bound $\Omega\left(n d \log \Delta\right)$ when $n\leq k$. 
{In other words, we must store the entire dataset in this case.  }

\begin{proof}[Proof of Theorem~\ref{thm:main_lower} (first part)]
We construct a family $\mathcal{P}$ as follows: for each dataset $\bsP\in \calP$, we 
choose $n$ different grid points in $[\Delta]^d$. 
Note that for each single grid point, there are $\Delta^d$ choices. Therefore, the size $|\calP|$ is
$\binom{\Delta^d}{n}$, which implies that $\log |\calP| = \Omega(nd \log \Delta)$.

For any two datasets $\bsP, \bsQ \in \mathcal{P}$, there must be a single grid point $\bsq$ such that $\bsq\in \bsQ\setminus \bsP$. 
Let 
$\bsC = \bsP\cup \{\underbrace{\bsc,\cdots,\bsc}_{k-n\text{ points}}\}$
, where $\bsc\in \R^d$ is an arbitrary point with $\bsc\neq \bsq$. 
We have
\[\operatorname{cost}_{z}\left(\bsQ,\bsC\right) \geq \dist\left(\bsq,\bsC\right)^z >0,~\operatorname{cost}_{z}\left(\bsP,\bsC\right) = 0 \notin \left(1\pm 3\eps\right)\operatorname{cost}_{z}\left(\bsQ,\bsC\right). \]
Using Lemma~\ref{lmm:setsize_to_bit}, we obtain the space lower bound $\Omega\left(n d \log \Delta\right)$.
\end{proof}

We then consider the second part of Theorem~\ref{thm:main_lower} when $n>k$. 
We first prove the lower bound $\Omega\left(kd\min\left\{\frac{1}{\varepsilon^2},\frac{d}{\log d}, \frac{n}{k}\right\}\right)$.
Recall that we have $n > 2$ and $\Delta = \Omega(\frac{k^{\frac{1}{d}}\sqrt{d}}{\eps})$.
Recall that we have the notion and properties of principal angles in Definitio~\ref{def:principal_angles} and Lemma~\ref{lmm:property_pa}.

The idea is still to construct a large family $\calP$ of datasets to apply Lemma~\ref{lmm:setsize_to_bit}.
For ease of analysis, we first do not require the construction of datasets $\bsP\subseteq [\Delta]^d$ and actually ensure that every $\bsP$ consists of orthonormal bases of some subspaces. 
At the end of the proof, we will show how to round and scale these datasets $\bsP$ into $[\Delta]^d$. 


Let $\theta_{i}$ represent the $i$-th least principal angles of the two subspaces spanned by $\bsP$ and $\bsQ$. 
The following lemma shows that large principal angles between $\bsP$ and $\bsQ$ imply a large cost difference on some center set $\left\{\bsc,-\bsc\right\}$.

\begin{lemma}[Principal angles to cost difference] \label{lmm:large_cost} Let $\bsP$, $\bsQ$ be datasets of $n$ orthonormal bases ($ 100\leq n \leq \frac{d}{2}$) satisfying that 
$\theta_{\frac{1}{32}10^{-6}\cdot n}  \geq \arccos\left(\frac{10^{-3}}{4\sqrt{2}}\right).$
There exists a unit vector $\bsc\in \R^d$ such that
$\operatorname{cost}_2(\bsP, \{\bsc,-\bsc\}) -\operatorname{cost}_2(\bsQ, \{\bsc,-\bsc\}) \geq \frac{1}{2}\sqrt{n}.
$
\end{lemma}

\begin{proof}
Proof of Lemma~\ref{lmm:large_cost} can be found in section~\ref{sec:large_cost}. 
\end{proof}

Applying $n = O\left(\frac{1}{\varepsilon^2}\right)$ to the above lemma leads to our desired cost difference $ \Omega\left(\varepsilon n\right)= \Omega\left(\eps \cost_2(\bsP, \left\{\bsc,-\bsc\right\})\right)$.
We then show it is possible to construct a large family $\mathcal{P}$ such that the principal angles between any two datasets in $\calP$ are sufficiently large.

\begin{lemma}[Construction of a large family of datasets] \label{lmm:large_size} When $n = O\left(\frac{d}{\log d}\right)$, there is a family $\mathcal{P}$ of size $$\exp\left(\frac{1}{256}10^{-6} \log\left(\frac{1}{1- \frac{1}{32}10^{-6}}\right) \cdot nd\right)$$
such that for any two dataset $\bsP,\bsQ\in \calP$, we have their principal angles satisfying $\theta_{\frac{1}{32}10^{-6}\cdot n} \geq \arccos\left(\frac{10^{-3}}{4\sqrt{2}}\right).$
\end{lemma}
\begin{proof}
Proof of Lemma~\ref{lmm:large_size} can be found in section~\ref{sec:large_size}.
\end{proof}

Combining Lemma~\ref{lmm:large_cost} and \ref{lmm:large_size}, we are ready to prove the second part of Theorem~\ref{thm:main_lower}. 

\begin{proof}[Proof of Theorem~\ref{thm:main_lower} (second part)]
\sloppy
The lower bound of $\Omega(kd \log \Delta)$ is trivial since $sc(n, \Delta, k,z,d, \varepsilon)$ is non-decreasing with $n$.
Then by the first part of Theorem~\ref{thm:main_lower}, we have $sc(n, \Delta, k,z,d, \varepsilon)\geq sc(k, \Delta, k,z,d, \varepsilon)\geq \Omega(kd \log \Delta)$.

\sloppy
Next, we prove the lower bound of $\Omega\left(kd\min\left\{\frac{1}{\varepsilon^2},\frac{d}{\log d},\frac{n}{k}\right\}\right)$. 
For ease of analysis, we prove the case of $k = z = 2$. 
The extensions to general $k$ and $z$ can be found in Section~\ref{sec:generalize_power_high} and \ref{sec:generalize_k}. 
{We first ensure that our parameters meet the requirements of our previous lemmas.}
Denote 
\[\tilde{n}  = \min\left\{\Theta\left(\frac{1}{484\varepsilon^2}\right),\Theta\left(\frac{d}{\log d}\right),n\right\} \geq 100,\]
where the first term $\Theta\left(\frac{1}{484\varepsilon^2}\right)$ is for achieving a large cost difference by Lemma~\ref{lmm:large_cost} and the second term $\Theta\left(\frac{d}{\log d}\right)$ is to satisfy the condition of Lemma~\ref{lmm:large_size}. 
Since $sc(n, \Delta, 2,2,d, \varepsilon)$ is non-decreasing with $n$, it suffices to prove a lower bound for $sc(\tilde{n}, \Delta, 2,2,d, \varepsilon)$. 

\sloppy
{Our proof proceeds as follows:
Lemma \ref{lmm:large_size} shows that we can find a family $\mathcal{P}$ of size $\exp\left(\frac{1}{256}10^{-6} \log\left(\frac{1}{1- \frac{1}{32}10^{-6}}\right) \cdot \tilde{n}d\right)$ such that for any two dataset in this set $\bsP$ and $\bsQ$, we have their principal angles 
$\theta_{\frac{1}{32}10^{-6}\cdot n}\geq \arccos\left(\frac{10^{-3}}{4\sqrt{2}}\right).$
Using this condition, Lemma \ref{lmm:large_cost} allows us to find a unit-norm vector $\bsc$ such that 
$\operatorname{cost}_2(\bsP, \{\bsc,-\bsc\}) -\operatorname{cost}_2(\bsQ, \{\bsc,-\bsc\})\geq \frac{1}{2} \sqrt{\tilde{n}}.$
By our choice of $\tilde{n}$, we have $\frac{1}{2} \sqrt{\tilde{n}} \geq  11\varepsilon \tilde{n}\geq 5\eps\cdot \cost_2(\bsP, \{\bsc,-\bsc\})$. This already satisfies the requirements of Lemma~\ref{lmm:setsize_to_bit}. However, it is important to note that the family we have obtained so far is constructed in a continuous space. Therefore, we need to discretize this family and demonstrate that this process does not significantly affect the properties we need.
}

We then round and scale every dataset $\bsP\in \calP$ to $[\Delta]^{d}$, where $\Delta = \lceil\frac{10\sqrt{d}}{\varepsilon}\rceil$.
The extra term $k^{\frac{1}{d}}$ will show up when we extend the result to general $k\geq 2$ (Section~\ref{sec:generalize_k}). 
Without loss of generality, we may assume that $\Delta$ is an odd integer. Otherwise, we just let $\Delta = \lceil\frac{10\sqrt{d}}{\varepsilon}\rceil+1$.

Denote $\boldsymbol{1}=(1,\cdots,1).$ For a dataset $\bsP = \left(\bsp_{1},\cdots,\bsp_{\tilde{n}}\right)\in \calP$, we will construct $\tilde{\mathcal{P}}$ to be our final family as follows: For each of dataset $\bsP\in\mathcal{P}$, we shift the origin to $\lceil\frac{\Delta}{2}\rceil\cdot\mathbf{1}$, scale it by a factor of $\frac{\Delta}{2}$ and finally perform an upward rounding on each dimension to put every point on the grid: 
$\tilde{\bsP} = \left(\lceil\frac{\Delta}{2}\bsp_{1}\rceil,\cdots,\lceil\frac{\Delta}{2}\bsp_{\tilde{n}}\rceil\rceil\right) + \lceil\frac{\Delta}{2}\rceil\cdot\mathbf{1} = \left(\tilde{\bsp}_{1},\cdots,\tilde{\bsp}_{\tilde{n}}\right).$ We will then show that this set fulfills the requirement of Lemma~\ref{lmm:setsize_to_bit}. For ease of explanation, we also define $\hat{\bsP}$ to be the dataset without rounding: $\hat{\bsP} = \left(\frac{\Delta}{2}\bsp_{1},\cdots,\frac{\Delta}{2}\bsp_{\tilde{n}}\right) + \lceil\frac{\Delta}{2}\rceil\cdot\mathbf{1}= \left(\hat{\bsp}_{1},\cdots,\hat{\bsp}_{\tilde{n}}\right).$ Moreover, let $\bar{\bsc} = \frac{\Delta}{2}\bsc+ \lceil\frac{\Delta}{2}\rceil\cdot\mathbf{1}$. We must have that for the scaling dataset, 
\begin{align*}
&\operatorname{cost}_2\left(\hat{\bsP}, \left\{\bar{\bsc},-\bar{\bsc}\right\}\right) = \frac{\Delta^2}{4} \operatorname{cost}_2\left(\bsP, \{\bsc,-\bsc\}\right) \leq \frac{\Delta^2 \tilde{n}}{2},
\\& \operatorname{cost}_2\left(\hat{\bsP}, \left\{\bar{\bsc},-\bar{\bsc}\right\}\right) -\operatorname{cost}_2\left(\hat{\bsQ}, \left\{\bar{\bsc},-\bar{\bsc}\right\}\right)=\frac{\Delta^2}{4}\left(\operatorname{cost}_2(\bsP, \{\bsc,-\bsc\}) -\operatorname{cost}_2(\bsQ, \{\bsc,-\bsc\})\right) \geq  \frac{11 \Delta^2\varepsilon \tilde{n}}{4}.
\end{align*}
On the other hand, for the rounding dataset, we have 
\begin{align*}
\left|\left\|\hat{\bsp}_{i}-\bar{\bsc}\right\|_2^2 - \left\|\tilde{\bsp}_{i}-\bar{\bsc}\right\|_2^2 \right|& \leq 2\left\|\hat{\bsp}_{i}-\tilde{\bsp}_{i}\right\|_2\left\|\hat{\bsp}_{i}-\bar{\bsc}\right\|_2 + \left\|\hat{\bsp}_{i}-\tilde{\bsp}_{i}\right\|_2^2 \leq 2\Delta\sqrt{d} + d\leq \frac{\Delta^2\varepsilon}{4} .  
\end{align*}
The case for $-\bar{\bsc}$ and other datasets is similar. Therefore, we have that for any dataset
\begin{align*}
& \left|\operatorname{cost}_2(\hat{\bsP}, \{\bar{\bsc},-\bar{\bsc}\}) - \operatorname{cost}_2(\tilde{\bsP}, \{\bar{\bsc},-\bar{\bsc}\})\right| \leq \frac{\Delta^2\varepsilon \tilde{n}}{4},
\\& \operatorname{cost}_2\left(\tilde{\bsP}, \left\{\bar{\bsc},-\bar{\bsc}\right\}\right) \leq \operatorname{cost}_2\left(\hat{\bsP}, \left\{\bar{\bsc},-\bar{\bsc}\right\}\right) + \frac{\Delta^2\varepsilon \tilde{n}}{4} \leq \frac{\Delta^2 \tilde{n}}{2} + \frac{\Delta^2\varepsilon \tilde{n}}{4} \leq \frac{3\Delta^2 \tilde{n}}{4}.
\end{align*}
We now have a rounded family $\tilde{\mathcal{\bsP}}$ such that all points are on the grid and we can find $\bar{\bsc}$ that 
\begin{align*}
\operatorname{cost}_2\left(\tilde{\bsP}, \left\{\bar{\bsc},-\bar{\bsc}\right\}\right) -\operatorname{cost}_2\left(\tilde{\bsQ}, \left\{\bar{\bsc},-\bar{\bsc}\right\}\right) & \geq \operatorname{cost}_2\left(\hat{\bsP}, \left\{\bar{\bsc},-\bar{\bsc}\right\}\right) -\operatorname{cost}_2\left(\hat{\bsQ}, \left\{\bar{\bsc},-\bar{\bsc}\right\}\right) - \frac{\Delta^2\varepsilon \tilde{n}}{2}\\& \geq \frac{9\Delta^2\varepsilon \tilde{n}}{4}.
\end{align*}
Since the cost function value is upper bounded by $\frac{3\Delta^2 \tilde{n}}{4}$, we must have that 
$\operatorname{cost}_2(\tilde{\bsP}, \{\bar{\bsc},-\bar{\bsc}\}) \notin \left(1\pm 3\varepsilon\right) \operatorname{cost}_2(\tilde{\bsQ}, \{\bar{\bsc},-\bar{\bsc}\}).$
Moreover, we can find origin being $\lceil\frac{\Delta}{2}\rceil\cdot\mathbf{1}$ such that all the center points and data points have distance to it less than $\frac{\Delta}{2}+\sqrt{d} \leq \Delta$.  
By Lemma~\ref{lmm:setsize_to_bit}, we have
\[ 
\spc(n, \Delta, 2,2,d, \varepsilon) \geq \Omega\left(\log \left|\tilde{\mathcal{P}}\right|\right)=\Omega\left(\log \left|\mathcal{P}\right|\right)\geq  \Omega\left(d\min\left\{\frac{1}{\varepsilon^2},\frac{d}{\log d},n\right\}\right).
\]

The extension to any constant $z\geq 1$ can be found in Section~\ref{sec:generalize_power_high}, which relies on the analysis of the Taylor expansion of function $(1+x)^z$. 
The extension to general $k\geq 2$ can be found in Section~\ref{sec:generalize_k}, whose main idea is to let every dataset consist of $\Theta(k)$ datasets from $\calP$ and set the positions of their center points in $[\Delta]^d$ ``remote'' from each other. 

Finally, we prove the lower bound $\Omega\left(k\log\log\frac{n}{k}\right)$. 
We again construct a large family $\calP$ of datasets.
For preparation, we find arbitrary $\frac{k}{2}$ points, denoted as $\bsp_1,\cdots, \bsp_{\frac{k}{2}}$, such that the distance between every two points is at least 10. 
This is available since $\Delta^{d} = \Omega(k)$.
Every dataset $\bsP\in \mathcal{P}$ is constructed as follow: denote $\boldsymbol{e}_1 =(1,0,\cdots,0)$ for each $i\in \left[\frac{k}{2}\right]$, we select $m_i\in [\log \frac{n}{k}]$ and put $2^{m_i}$ points at $\bsp_{i}+\boldsymbol{e}_1$ and $\frac{2n}{k} -2^{m_i}$ points at $\bsp_{i}$. 
Therefore, the total number of possible datasets is $|\mathcal{P}| = \prod_{i=1}^{k}m_i  = O\left(\left(\log \frac{n}{k}\right)^{\frac{k}{2}}\right)$. We then consider for any two different datasets $\bsP, \bsQ \in \mathcal{P}$, there must exist $l$ such that $\bsP$ and $\bsQ$ have different assignments for $\bsp_{l}$ and $\bsp_{l}+\boldsymbol{e}_1$. Without loss of generality,assume that $\bsP$ put $2^{i}$ at $\bsp_{l}+\boldsymbol{e}_1$ while $\bsQ$ put $2^{j}$ for $i<j$. Choosing center set 
$\bsC = \left\{\bsp_1,\bsp_1+\boldsymbol{e}_1,\cdots,\bsp_l,\bsp_l+2\boldsymbol{e}_1,\cdots,\bsp_{\frac{k}{2}},\bsp_{\frac{k}{2}}+\boldsymbol{e}_1\right\},$
we must have that 
$\operatorname{cost}_z(\bsP, \bsC) = 2^i \leq \frac{1}{2} 2^j \notin  \left(1 \pm \frac{1}{2}\right) \operatorname{cost}_z(\bsQ, \bsC),$
which satisfies our requirement of $\mathcal{P}$. 
Lemma \ref{lmm:setsize_to_bit} provides us with a lower bound of $\Omega\left(\log\left|\mathcal{P}\right|\right) \geq \Omega\left(k\log \log \frac{n}{k}\right).$

\end{proof}

\subsection{Proof of Lemma~\ref{lmm:large_cost}: Principal Angles to Cost Difference}
\label{sec:large_cost}

{In this section, we primarily prove Lemma~\ref{lmm:large_cost}, restated as follows
\begin{lemma}[Restatement of Lemma~\ref{lmm:large_cost}] Let $\bsP$, $\bsQ$ be datasets of $n$ orthonormal bases ($ 100\leq n \leq \frac{d}{2}$) satisfying that 
$\theta_{\frac{1}{32}10^{-6}\cdot n}  \geq \arccos\left(\frac{10^{-3}}{4\sqrt{2}}\right).$
There exists a unit vector $\bsc\in \R^d$ such that
$\operatorname{cost}_2(\bsP, \{\bsc,-\bsc\}) -\operatorname{cost}_2(\bsQ, \{\bsc,-\bsc\}) \geq \frac{1}{2}\sqrt{n}.
$
\end{lemma}
Our proof strategy proceeds as follows: Let $\bsP=\left\{\bsp_{i} \in \R^d \right\}_{i \in [n]}$ and $ \bsQ=~\left\{\bsq_{i} \in \R^d \right\}_{i \in [n]}$ be two orthonormal bases and let their inner product matrix be
\[\bsU :=\bsP^{\top} \bsQ = \left[\bsU_{1}, \cdots,\bsU_{n} \right]^\top = \begin{bmatrix}
\bsp_{1}^{\top} \bsq_{1}  & \cdots   & \bsp_{1}^\top\bsq_{n}\\
\vdots  & \ddots  & \vdots \\
\bsp_{n}^{\top} \bsq_{1}  &\cdots  &\bsp_{n}^\top\bsq_{n}
\end{bmatrix} = \left(\bsU_{ij}\right)_{i,j \in [n]}.\]
Compute the cost function for any unit vector $\bsc\in \R^d$, we have 
\begin{align*}
\operatorname{cost}_2(\bsP, \{\bsc,-\bsc\}) = \sum_{i=1}^n \left(\left\|\bsp_{i}\right\|_2^2+\left\|\bsc\right\|_2^2 - 2\left | \langle \bsp_{i},\bsc\rangle\right|\right) =2n - 2\sum_{i=1}^n \left | \langle  \bsp_{i},\bsc\rangle\right| \leq 2n. 
\end{align*}
The difference between the cost of the two datasets $\bsP$ and $\bsQ$ is 
\[ \operatorname{cost}_2(\bsP, \{\bsc,-\bsc\}) -\operatorname{cost}_2(\bsQ, \{\bsc,-\bsc\}) = 2 \left( \sum_{i=1}^n \left | \langle  \bsq_{i},\bsc\rangle\right| - \left | \langle  \bsp_{i},\bsc\rangle\right|\right).\]
Our objective is to maximize this value as much as possible. 
To maximize the first term, We choose $\bsc$ to be a vector in the subspace spanned by $\bsQ$ that $\bsc = \bsQ \boldsymbol{\zeta} = \sum_{i=1}^{n} \boldsymbol{\zeta}_i \bsq_i$, thus our difference becomes:
\begin{align*}
\operatorname{cost}_2(\bsP, \{{\bsc},-{\bsc}\}) -\operatorname{cost}_2(\bsQ, \{{\bsc},-{\bsc}\}) &= 2\left(  \sum_{i=1}^n \left | \boldsymbol{\zeta}_i \right| - \sum_{i=1}^n \left | \langle  \bsp_{i},\sum_{j=1}^{n}  \boldsymbol{\zeta}_j \bsq_{j}\rangle\right| \right)\\&= 2 \left( \sum_{i=1}^n \left |\boldsymbol{\zeta}_i \right| - \sum_{i=1}^n \left | \sum_{j=1}^{n} \boldsymbol{\zeta}_j \bsU_{ij} \right| \right) .   
\end{align*}
The first term is maximized when most of the $|\boldsymbol{\zeta}_i |$ are $\frac{1}{\sqrt{n}}$. 
Additionally, we observe that minimizing the second term is very similar to the objective of partial coloring, which is formally defined in Definition~\ref{def:partial_coloring}.
Partial coloring focuses on the infinity norm, whereas we are concerned with the 1-norm. 
However, as long as this infinity norm is sufficiently small, our 1-norm can be controlled by it up to a factor of $n$.
Therefore, we consider employing techniques from this area to achieve this. 
We note that if $\bsU$ is an arbitrary matrix, it is challenging to minimize the second term effectively. 
However, our $\bsU$ is the inner product matrix of two subspaces with significantly different orientations (i.e., having many large principal angles), and its row norm is sufficiently small for us to bound the second term. 
}



%
To this end, we first show that large principal angles imply that most rows of the inner product matrix $\bsU$ have a small $\ell_2$-norm. 

\begin{lemma}[Principal angles to row norms] \label{lmm:small_sum_of_square} Let $\bsP$, $\bsQ$ be datasets of $n$ orthonormal bases following the condition that 
$\theta_{\frac{1}{32}10^{-6}\cdot n}  \geq \arccos\left(\frac{10^{-3}}{4\sqrt{2}}\right).$
There exists a set $\mathcal{K}$ of size larger than $\left(1-\frac{10^{-4}}{16}\right)n$ and having property that the inner product matrix satisfies
$\|\bsU_i\|_2^2:= \sum_{j=1}^{n} \left(\bsU_{ij}\right)^2 = \left\{\begin{array}{l}
\leq 10^{-2}, i \in  \calK \\
\leq 1, i \notin \calK
\end{array}\right. .$
\end{lemma}
\begin{proof}
Since both of the datasets are composed of only orthonormal bases, we have 
$$\forall i \in [n],\sum_{j=1}^{n} \left(\bsU_{ij}\right)^2 = \sum_{j=1}^{n} \left(\bsp_{i}^{\top}\bsq_{j}\right)^2 \leq \left\|\bsp_{i} \right\|_2^2  = 1 . $$
Therefore, we focus on the existence of set $\calK$. 
By the property of principal angles shown in Lemma~\ref{lmm:property_pa}, we have that 
$\sigma_{i} = \cos \theta_i, i=1,\cdots, n.$
Moreover, we have the relation between singular values and Frobenius norm of the matrix being 
$\sum_{i=1}^{n}\sigma_{i}^2 = \left\|\bsU\right\|_F^2 = \sum_{i=1}^n \sum_{j=1}^n \left(\bsU_{ij}\right)^2 .$
On the other hand, the condition shows that 
\begin{align*}
&\theta_{n} \geq \cdots \geq \theta_{\frac{1}{32}10^{-6}\cdot n}  \geq \arccos\left(\frac{10^{-3}}{4\sqrt{2}}\right), \sigma_{n} \leq \cdots \leq \sigma_{\frac{1}{32}10^{-6}\cdot n} \leq \cos \left(\arccos\left(\frac{10^{-3}}{4\sqrt{2}}\right)\right) = \frac{10^{-3}}{4\sqrt{2}}.
\end{align*}
With the general upper bound that $\sigma_{i} = \cos\theta_{i} \leq 1 $, we have 
\[\sum_{i=1}^{n}\sigma_{i}^2 \leq \frac{1}{32}10^{-6}n + \left(1-\frac{1}{32}10^{-6}\right) n \cdot\left(\frac{10^{-3}}{4\sqrt{2}}\right)^2  \leq \frac{1}{16}10^{-6}n.\]
We would then have that the number of rows with sum of square larger than $10^{-2}$ is less than $\frac{1}{16}10^{-6} n / 10^{-2} = \frac{1}{16}10^{-4} n$, which completes the proof.  
\end{proof}

Next, we show that a partial coloring can be found for the rows of $\bsU$ using similar techniques as that of \cite{spencer1985six}.

\begin{lemma}[Row norms to partial coloring] \label{lmm:discrepancy} Considering that we have a matrix $\bsU \in \mathbb{R}^{n\times n} $ (with $n\geq 100$) such that we can find a set $\calK$ of size larger than $\left(1-\frac{10^{-4}}{16}\right)n$ and having the property that 
$\sum_{j=1}^{n} \bsU_{ij}^2 = \left\{\begin{array}{l}
\leq 10^{-2}, i \in  \calK \\
\leq 1, i \notin \calK
\end{array}\right. .$
we can find a partial coloring $\{\boldsymbol{\zeta}_{i}\}_{[n]} \in \{-1,0,1\}_{[n]}$ such that $\left|i:\boldsymbol{\zeta}_i = 0\right| \leq \frac{1}{4}n$ and $\left| \sum_{j=1}^{n}\boldsymbol{\zeta}_j \bsU_{ij}\right|\leq \frac{1}{2},\forall i \in [n]$.
\end{lemma}

\begin{proof}
Define $\row_i(\boldsymbol{\zeta}_1,\cdots, \boldsymbol{\zeta}_n):= \sum_{j=1}^{n}\boldsymbol{\zeta}_j \bsU_{ij}$ for all $i\in [n]$.
Define the rounding map $$\round(\boldsymbol{\zeta}_1,\cdots, \boldsymbol{\zeta}_n) =(b_1,\cdots, b_n),$$ where $b_i$ is the nearest integer to $\row_i(\boldsymbol{\zeta}_1,\cdots, \boldsymbol{\zeta}_n)$. 
That is, 
$b_i=0$ if and only if $\left|\row_i\right| \leq \frac{1}{2}, b_i=1$ if and only if $ \frac{1}{2}< \left|\row_i\right| \leq \frac{3}{2}, b_i=-1$ if and only if $ -\frac{3}{2}\leq \left|\row_i\right| < -\frac{1}{2},$ etc.
We then define a subset $\mathcal{B}\subset \mathbb{Z}^n $ of the range that 
$
\mathcal{B} = \{(b_1,\cdots, b_n) \in \mathbb{Z}^n: |\{i:|b_i|\geq s\}| < \alpha_s n, \text{for all positive integer } s\},$ where 
\[\alpha_s = \left[2\left(1-\frac{10^{-4}}{16}\right) \exp\left(-\frac{(2s-1)^2}{8 \cdot 10^{-2}}\right) + 2\cdot \frac{10^{-4}}{16}\exp\left(-\frac{(2s-1)^2}{8}\right) \right] 2^{s+1}.\]
We firstly prove that $\left|\round^{-1}(\mathcal{B})\right|\geq \frac{1}{2} 2^{n}.$ 
To see this, let $\boldsymbol{\zeta}_1, \cdots, \boldsymbol{\zeta}_n \in\{-1,+1\}$ be independent and uniform and let $\row_1, \cdots, \row_n, b_1, \cdots, b_n$ be the values generated. 
We shall note that the standard deviation of $\row_i$ is the $l_2$-norm of $i$-th row. 
Thus, classic Chernoff bound provides 
\[
\operatorname{Pr}\left[\left|b_i\right| \geq s\right]=\operatorname{Pr}\left[\left|\row_i\right| \geq \frac{2s-1}{2}\right]<\left\{\begin{matrix}
2\exp\left(-\frac{(2s-1)^2}{8 \cdot 10^{-2}}\right), i \in \calK \\
2\exp\left(-\frac{(2s-1)^2}{8}\right), i \notin \calK
\end{matrix}\right. .\]
As expectation is linear, we would have 
\begin{gather*}
\mathbb{E}\left[\left|\left\{i:\left|b_i\right| \geq s\right\}\right|\right] <  n \left[2\left(1-\frac{10^{-4}}{16}\right) \exp\left(-\frac{(2s-1)^2}{8 \cdot 10^{-2}}\right) + 2\cdot \frac{10^{-4}}{16}\exp\left(-\frac{(2s-1)^2}{8}\right) \right] ,\\
\operatorname{Pr}\left[\left|\left\{i:\left|b_i\right| \geq s\right\}\right|\geq \alpha_s n \right] \leq \frac{1}{2^{s+1}},\operatorname{Pr}\left[\left(b_1, \cdots, b_n\right) \notin \mathcal{B}\right] \leq \sum_{s=1}^{\infty} \frac{1}{2^{s+1}}=\frac{1}{2} .
\end{gather*}
That is, at least half of all $\left(\boldsymbol{\zeta}_1, \cdots, \boldsymbol{\zeta}_n\right) \in\{-1,+1\}^n$ are in $\round^{-1}(\mathcal{B})$, yielding $\left|\round^{-1}(\mathcal{B})\right|\geq \frac{1}{2} 2^{n}$.

We then consider the size of $\mathcal{B}$ by crude counting arguments. We have 
\[
|\mathcal{B}| \leq \prod_{s=1}^{\infty}\left[\left[\sum_{l=0}^{\alpha_s n}{
\binom{n}{l}}
\right] 2^{\alpha_s n}\right].
\]
This is because  $\left\{i:\left|b_i\right|=s\right\}$ can be chosen in at most $\sum_{i=0}^{\alpha_s n}{
\binom{n}{i}}$ ways and, having been selected, can be split into $\left\{i: b_i=s\right\}$ and $\left\{i: b_i=-s\right\}$ in at most $2^{\alpha_s n}$ ways. 
We bound this value with
\[
\sum_{l=0}^{\alpha n}
{\binom{n}{l}}
 \leq 2^{n\cdot\ent(\alpha)}, \ent(\alpha)=-\alpha \log _2 \alpha-(1-\alpha) \log _2(1-\alpha), \]
where $\ent(\alpha) $ is the entropy function. Therefore, $|\mathcal{B}| \leq 2^{h n},$ where $h=\sum_{s=1}^{\infty}[\ent(\alpha_s)+\alpha_s].$
In our case, we shall have that 
\begin{align*}
&h=\sum_{s=1}^{\infty}\left[\ent\left(\alpha_s\right)+\alpha_s\right] \leq \sum_{s=1}^{\infty}\left[\exp(1) \cdot \alpha_s^{\frac{1}{\ln 4}}+\alpha_s\right] \leq (\exp(1)+1)\sum_{s=1}^{\infty} \alpha_s^{\frac{1}{\ln 4}},\\
&\frac{\alpha_s}{\alpha_{s+1}} \geq \frac{1}{2} \exp\left(\frac{(2s+1)^2-(2s-1)^2}{8}\right) = \frac{\exp(s)}{2}  \geq \frac{\exp(1)}{2},\\
&\alpha_1 = \left[2\left(1-\frac{10^{-4}}{16}\right) \exp\left(-\frac{1}{8 \cdot 10^{-2}}\right) + 2\cdot \frac{10^{-4}}{16}\exp\left(-\frac{1}{8}\right) \right] 4 \leq 10^{-4}.
\end{align*}
Combining together, we have 
\[h \leq \left(\exp(1)+1\right)\sum_{s=1}^{\infty} \alpha_s^{\frac{1}{\ln 4}} \leq (\exp(1)+1) \frac{\alpha_1^{\frac{1}{\ln 4}}}{1-\left(\frac{2}{\exp(1)}\right)^{\frac{1}{\ln 4}}} \leq (\exp(1)+1)\frac{\left(10^{-4}\right)^{\frac{1}{\ln 4}}}{1-\left(\frac{2}{\exp(1)}\right)^{\frac{1}{\ln 4}}} \leq 0.03.\]

Applying the pigeonhole principle, there exists specific $(\tilde{b}_1, \cdots, \tilde{b}_n) \in \mathcal{B}$ such that
\begin{align*}
&\mathcal{A}=\{(\boldsymbol{\zeta}_1, \cdots, \boldsymbol{\zeta}_n) \in\{-1,+1\}^n: \round(\boldsymbol{\zeta}_1, \cdots, \boldsymbol{\zeta}_n)=(\tilde{b}_1, \cdots, \tilde{b}_n)\}, \\&|\mathcal{A}| \geq \frac{\round^{-1}(\mathcal{B})}{|\mathcal{B}|}| \geq 2^{n(1-h)-1}.
\end{align*}
 We use the following result.
\begin{lemma}[\cite{kleitman1966combinatorial}] If $\mathcal{A} \subset\{-1,+1\}^n$ and $|\mathcal{A}| \geq 2^{n\cdot \ent(1 / 2-p)}$ with $p>0$, then $\operatorname{diam}(\mathcal{A}) \geq(1-2 p) n$.
\end{lemma}
In our case, we have for $n \geq 100$ and hence 
$2^{n(1-h)-1} \geq 2^{n\left(1-0.03-0.01\right)} \geq 2^{n\cdot \ent\left(\frac{1}{2} - \frac{1}{8}\right)}.$
Thus, $ \operatorname{diam}(\mathcal{A}) \geq  \left(1-\frac{1}{4}\right)n$. That is, there exist two vectors in $\mathcal{A}$ which differ in at least $\left(1-\frac{1}{4}\right)n$ coordinates. Let $\boldsymbol{\zeta}^{\prime}=\left(\boldsymbol{\zeta}_1^{\prime}, \cdots, \boldsymbol{\zeta}_n^{\prime}\right), \boldsymbol{\zeta}^{\prime \prime}=\left(\boldsymbol{\zeta}_1^{\prime \prime}, \cdots, \boldsymbol{\zeta}_n^{\prime \prime}\right) \in \mathcal{A}$ with $
\left\|\boldsymbol{\zeta}^{\prime} - \boldsymbol{\zeta}^{\prime \prime}\right\|_1 =\operatorname{diam}(\mathcal{A})$. Set
$
\boldsymbol{\zeta}=\left(\boldsymbol{\zeta}_1, \cdots, \boldsymbol{\zeta}_n\right)=\frac{\boldsymbol{\zeta}^{\prime}-\boldsymbol{\zeta}^{\prime \prime}}{2}.$ 
All $\boldsymbol{\zeta}_j \in\{-1,0,+1\}$ and $\boldsymbol{\zeta}_j=0$ if and only if $\boldsymbol{\zeta}_j^{\prime}=\boldsymbol{\zeta}_j^{\prime \prime}$. Therefore, 
\[
\left|\left\{i: \boldsymbol{\zeta}_i=0\right\}\right|=n-\left\|\boldsymbol{\zeta}^{\prime} - \boldsymbol{\zeta}^{\prime \prime}\right\|_1=n-\operatorname{diam}(\mathcal{A}) \leq \frac{1}{4} n .
\]
Moreover, for all $i$, since $\round$ is identical on $\boldsymbol{\zeta}^{\prime}$ and $\boldsymbol{\zeta}^{\prime \prime}, \row_i\left(\boldsymbol{\zeta}^{\prime}\right)$ and $\row_i\left(\boldsymbol{\zeta}^{\prime \prime}\right)$ lie on a common interval of length $1$, we have
\[ \left|
\row_i(\boldsymbol{\zeta}_1,\cdots,\boldsymbol{\zeta}_n)\right|=\left|\frac{\row_i(\boldsymbol{\zeta}_1^{\prime},\cdots,\boldsymbol{\zeta}_n^{\prime}) - \row_i(\boldsymbol{\zeta}_1^{\prime\prime},\cdots,\boldsymbol{\zeta}_n^{\prime\prime})}{2} \right|\leq \frac{1}{2}.
\]

\end{proof} 

By Lemmas~\ref{lmm:small_sum_of_square} and \ref{lmm:discrepancy}, we are ready to prove Lemma~\ref{lmm:large_cost}. 

\begin{proof}[Proof of Lemma~\ref{lmm:large_cost}]
By Lemmas \ref{lmm:small_sum_of_square} and \ref{lmm:discrepancy}, there exists a partial coloring $\{\boldsymbol{\zeta}_{j}\}_{[n]} \in \{-1,0,1\}_{[n]}$ such that $\left|i:\boldsymbol{\zeta}_i = 0\right| \leq \frac{1}{4}n$ and $\left| \sum_{j=1}^{n}\boldsymbol{\zeta}_j \bsU_{ij}\right|\leq \frac{1}{2},\forall i\in[n]$. 
Let $\tilde{\bsc} = \sum_{i=1}^{n}  \frac{\boldsymbol{\zeta}_i}{\sqrt{n}} \bsq_{i}$. 
The cost difference w.r.t. $\{\tilde{\bsc},-\tilde{\bsc}\}$ is 
\begin{align*}
  &\operatorname{cost}_2(\bsP, \{\tilde{\bsc},-\tilde{\bsc}\}) -\operatorname{cost}_2(\bsQ, \{\tilde{\bsc},-\tilde{\bsc}\})\\& =2 \sum_{i=1}^n \left | \langle  \bsq_{i},\tilde{\bsc}\rangle\right| - \sum_{i=1}^n \left | \langle  \bsp_{i},\tilde{\bsc}\rangle\right| = 2  \sum_{i=1}^n \left | \frac{\boldsymbol{\zeta}_i}{\sqrt{n}} \right| - \sum_{i=1}^n \left | \langle  \bsp_{i},\sum_{j=1}^{n}  \frac{\boldsymbol{\zeta}_j}{\sqrt{n}} \bsq_{j}\rangle\right|\\&= 2 \sum_{i=1}^n \left | \frac{\boldsymbol{\zeta}_i}{\sqrt{n}} \right| - \sum_{i=1}^n \left | \sum_{j=1}^{n} \frac{\boldsymbol{\zeta}_j}{\sqrt{n}} \bsU_{ij} \right| .   
\end{align*}
Due to our choice of $\{\boldsymbol{\zeta}_i\}_{[n]}$, we would have that 
\begin{align*}
&\sum_{i=1}^n \left | \frac{\boldsymbol{\zeta}_j}{\sqrt{n}} \right| = \left(1-\frac{\left|i:\boldsymbol{\zeta}_i = 0\right|}{n}\right) \sqrt{n} \geq \frac{3}{4} \sqrt{n},\\
&\sum_{i=1}^n \left | \sum_{j=1}^{n} \frac{\boldsymbol{\zeta}_j}{\sqrt{n}} \bsU_{ij} \right| = \frac{1}{\sqrt{n}}\sum_{i=1}^n \left | \sum_{j=1}^{n} \boldsymbol{\zeta}_j \bsU_{ij} \right| \leq  \frac{1}{\sqrt{n}}\sum_{i=1}^n \frac{1}{2} = \frac{1}{2} \sqrt{n}.
\end{align*}
Combining together, we would have that 
$\operatorname{cost}_2(\bsP, \{\tilde{\bsc},-\tilde{\bsc}\}) -\operatorname{cost}_2(\bsQ, \{\tilde{\bsc},-\tilde{\bsc}\}) \geq 2(\frac{3}{4} \sqrt{n} - \frac{1}{2} \sqrt{n}) = \frac{1}{2} \sqrt{n}.$
To complete proof, note that 
$\left\| \tilde{\bsc} \right\|_2 = \sqrt{\sum_{i=1}^{n} \frac{\boldsymbol{\zeta}_i^2}{n}} \leq  1.$
Since $d > 2n$, we can always find a vector $\hat{\bsc}$ such that $\hat{\bsc} \perp \bsP,\bsQ$ and $\left\| \hat{\bsc} \right\|_2^2 = 1- \left\| \tilde{\bsc} \right\|_2^2 \geq 0$. Let $\bsc = \tilde{\bsc} + \hat{\bsc}$, we should have that $\bsc$ is of unit norm and 
\[\operatorname{cost}_2(\bsP, \{\bsc,-\bsc\}) -\operatorname{cost}_2(\bsQ, \{\bsc,-\bsc\})= \operatorname{cost}_2(\bsP, \{\tilde{\bsc},-\tilde{\bsc}\}) -\operatorname{cost}_2(\bsQ, \{\tilde{\bsc},-\tilde{\bsc}\}) \geq  \frac{1}{2} \sqrt{n}.\]
\end{proof}

\subsection{Proof of Lemma~\ref{lmm:large_size}: Construction of A Large Family $\calP$}
\label{sec:large_size}

{
In this chapter, our goal is to prove the existence of a sufficiently large family that meets our conditions. 
Our approach uses probabilistic methods to show that the probability of obtaining a sufficiently large family through random sampling of subspaces is greater than zero.
}
We first present the following lemma showing that the principal angles are likely to be large between random subspaces. 

\begin{lemma}[Large principal angles between random subspaces] \label{lmm:principal_angles}
Let $\mathcal{X},\mathcal{Y}$ be two sub-spaces chosen from the uniform distribution on the Grassmann manifold of $n$-planes in $\mathbb{R}^d$ $\left(n = O\left(\frac{d}{\log d}\right)\right)$ endowed with its canonical metric. 
We have 
\[
\operatorname{Pr}\left(\theta_{\frac{1}{32}10^{-6}\cdot n} <  \arccos\left(\frac{10^{-3}}{4\sqrt{2}}\right)\right)< \exp\left(-\frac{1}{128}10^{-6} \log\left(\frac{1}{1- \frac{1}{32}10^{-6}}\right)\cdot nd \right).\]
\end{lemma}

With Lemma~\ref{lmm:principal_angles}, we are ready to show the existence of a large enough family $\mathcal{P}$.

\begin{proof}[Proof of Lemma~\ref{lmm:large_size}]
We will generate $\exp\left(\frac{1}{256}10^{-6} \log\left(\frac{1}{1- \frac{1}{32}10^{-6}}\right)\cdot nd\right)$ subspaces, each of which is chosen from uniform distribution on the Grassmann manifold of $n$-planes in $\mathbb{R}^d$.
We then let our dataset be arbitrary orthonormal bases of the subspace. 
Notice that by Lemma~\ref{lmm:property_pa}, the choice of the orthonormal bases will not affect the value of the principal angles. With the result of lemma \ref{lmm:principal_angles}, we have that for any two datasets $\bsP$ and $\bsQ$, 
\begin{align*}
\operatorname{Pr}\left(\theta_{\frac{1}{32}10^{-6}\cdot n} <  \arccos\left(\frac{10^{-3}}{4\sqrt{2}}\right)\right) <  \exp\left(-\frac{1}{128}10^{-6}\log\left(\frac{1}{1- \frac{1}{32}10^{-6}}\right)\cdot nd \right).
\end{align*}
We then consider that 
\begin{align*}
&\operatorname{Pr}\left(\forall \bsP\neq \bsQ \in \mathcal{P},\theta_{\frac{1}{32}10^{-6}\cdot n} \geq  \arccos\left(\frac{10^{-3}}{4\sqrt{2}}\right)\right)  \\&=1-\operatorname{Pr}\left(\exists \bsP\neq \bsQ \in \mathcal{P},\theta_{\frac{1}{32}10^{-6}\cdot n} <  \arccos\left(\frac{10^{-3}}{4\sqrt{2}}\right)\right)
\\&\geq 1-\sum_{i\neq j}\operatorname{Pr}\left(\theta_{\frac{1}{32}10^{-6}\cdot n} <  \arccos\left(\frac{10^{-3}}{4\sqrt{2}}\right)\right) \\&> 1-\sum_{i\neq j}\exp\left(-\frac{1}{128}10^{-6}\log\left(\frac{1}{1- \frac{1}{32}10^{-6}}\right)\cdot nd \right)
\\&> 1- \left(\exp\left(\frac{1}{256}10^{-6} \log\left(\frac{1}{1- \frac{1}{32}10^{-6}}\right)\cdot nd\right)\right)^2 
\cdot\exp\left(-\frac{1}{128}10^{-6}\log\left(\frac{1}{1- \frac{1}{32}10^{-6}}\right)\cdot nd \right)\\&= 0.
\end{align*}
Therefore, there is a positive probability for us to find enough datasets that fulfill our requirement, which shows the existence of the family. 
\end{proof}

It remains to prove Lemma~\ref{lmm:principal_angles}. 
The first idea is to apply Random Matrix Theory, e.g. \cite{Meckes_2019}, which shows that with high probability the Hilbert-Schmidt norm of a random matrix is small. 
However, the theory only considers the squared matrices and their results may not be strong enough to have a bound of $\exp\left(-O(nd)\right)$. 
Instead, we take advantage of the techniques in \cite{absil2006largest}, which calculates the distribution of the largest principal angle $\theta_1$. 
Using algebraic calculations, we extend their ideas to bound $\theta_{O(n)}$. 

{Our core idea is to bound $\operatorname{Pr}(\theta_{O(n)} < \theta)$ for some specific value $\theta$. For preparation, we first have the joint distribution of the square of the cosine of the principal angles given in \cite{absil2006largest}.}

{
\begin{lemma}[Section 2 in \cite{absil2006largest}]\label{lmm:distribution}
Let $\mathcal{X},\mathcal{Y}$ be two sub-spaces chosen from the uniform distribution on the Grassmann manifold of $n$-planes in $\mathbb{R}^d$, with $d\geq 2n$. Let $\theta_1\leq \cdots\leq \theta_n $ be the principal angles between $\mathcal{X}$ and $\mathcal{Y}$, and have $\mu_1 \geq \cdots\geq \mu_n $ such that $\mu_i = \cos^2\theta_i$. The joint probability density function of the $\mu$’s is thus given by
\[\operatorname{dens}\left(\mu_1, \cdots, \mu_n\right)=c_{n, n, d-n} \prod_{i<j}\left|\mu_i-\mu_j\right| \prod_{i=1}^p \mu_i^{-\frac{1}{2}}\left(1-\mu_i\right)^{\frac{1}{2}\left(d-2n-1\right)} .\]
where $c_{n, n, d-n}=\frac{\pi^{\frac{n^2}{2}}}{\Gamma_n\left(\frac{n}{2}\right)} \cdot \frac{\Gamma_n\left(\frac{d}{2}\right)}{\Gamma_n\left(\frac{n}{2}\right) \Gamma_n\left(\frac{d-n}{2}\right)}$.
\end{lemma}}

{Here we use the notion of multivariate gamma function, which can be expressed as a product of ordinary gamma functions. 
It should be noted that $a$ and some other scalars in subsequent definitions are complex numbers.}

{
\begin{definition}[Theorem 2.1.12 in \cite{muirhead2009aspects}]\label{lmm:2112} Given real numbers $a,m$ satisfying $a>\frac{1}{2}(m-1)$, the multivariate gamma function, denoted by $\Gamma_m(a)$, is defined to be
\[
\Gamma_m(a)=\pi^{m(m-1) / 4} \prod_{i=1}^m \Gamma\left[a-\frac{1}{2}(i-1)\right].
\]
\end{definition}}

{By letting $m=1$ the multivariate gamma function reduces to the ordinary one, and similarly we have $\Gamma_2(a)=\pi^{1 / 2} \Gamma(a) \Gamma(a-1 / 2)$ and $\Gamma_3(a)=\pi^{3 / 2} \Gamma(a) \Gamma(a-1 / 2) \Gamma(a-1)
$.} 

{With the joint probability density function, we bound the probability by integrating over the set $\left\{\theta_{1}\leq \cdots \leq \theta_{O(n)} < \theta\right\}$. To facilitate the integration process, we will ultimately convert it into an integration over matrices. Therefore, we also need the Gaussian hypergeometric function ${ }_2 F_1$ of matrix argument. The original definition of ${ }_2 F_1$ is rather complex and readers can refer to Definition~\ref{def:original_gaussian_hyper} and \cite{muirhead2009aspects} for details. In our computation, we will only focus on the integral representation of it defined in Lemma~\ref{lmm:742}. Recall that $\bsI_q$ is denoted as a $q\times q$ identity matrix.}

{
\begin{lemma}[Theorem 7.4.2. in  \cite{muirhead2009aspects}] \label{lmm:742}
Given real numbers $e,f,g$ and a real symmetric $q\times q$ matrix $\bsX$ satisfying 
$$\bsX<\bsI_q, e>\frac{1}{2}(q-1),g-e>\frac{1}{2}(q-1), $$
the Gaussian hypergeometric function of matrix argument ${ }_2 F_1(e, f ; g ; \bsX)$ function has the integral representation
\begin{align*}
&{ }_2 F_1(e, f ; g ; \bsX)\\&=\frac{\Gamma_q(g)}{\Gamma_q(e) \Gamma_q(g-e)} \int_{0<\bsY<\bsI_q} \operatorname{det}(\bsI_q-\bsX \bsY)^{-f}(\operatorname{det} \bsY)^{e-(q+1) / 2} 
\cdot \operatorname{det}(\bsI_q-\bsY)^{g-e-(q+1) / 2}(\de \bsY),
\end{align*}
\end{lemma}}

{
The function value of the identity matrix is, }

{
\begin{lemma}[Equation (3.2) in \cite{richards2020reflection}]\label{lmm:32} Given real numbers $e,f,g,q$, The value of ${ }_2 F_1$ function for identity matrix $\bsI_q $ is 
\[{ }_2 F_1\left(e, f ; g ; \bsI_q\right) = \frac{\Gamma_{q}(g)\Gamma_{q}(g-e-f)}{\Gamma_{q}(g-e)\Gamma_{q}(g-f)}. \]
\end{lemma}}

{We also need the following properties of Gaussian hypergeometric function of matrix argument}

{
\begin{lemma}[Refinement of Theorem 7.4.2. and Theorem 7.4.3. in  \cite{muirhead2009aspects}]\label{lmm:742+743} Given real numbers $e,f,g$ and a real symmetric $q\times q$ matrix $\bsX$ satisfying 
$$\bsX<\bsI_q, e>\frac{1}{2}(q-1),g-e>\frac{1}{2}(q-1), $$
the ${ }_2 F_1$ function satisfies that 
\begin{align*}
& \int_{0<\bsY<\bsI_q} \operatorname{det}(\bsI_q-\bsX \bsY)^{-f}(\operatorname{det} \bsY)^{e-(q+1) / 2} \operatorname{det}(\bsI_q-\bsY)^{g-e-(q+1) / 2} \left(\de \bsY\right) \\
&=\frac{\Gamma_q(e) \Gamma_q(g-e)}{\Gamma_q(g)} \operatorname{det}(\bsI_q-\bsX)^{-f}{ }_2 F_1\left(g-e, f ; g ;-\bsX(\bsI_q-\bsX)^{-1}\right)
\end{align*}
\end{lemma}}

{
\begin{lemma}\label{lmm:less_than_identity}Given real numbers $e,f,g,q$ and $0<h<1$, the ${ }_2 F_1$ function satisfies that 
\[{ }_2 F_1\left(e, f ; g ; (1-h)\bsI_q\right) \leq { }_2 F_1\left(e, f ; g ; \bsI_q\right). \]
\end{lemma}}
\begin{proof}
{
For the complete proof, please refer to Appendix~\ref{appendix}. The idea is to use the original definition of the Gaussian hypergeometric function of matrix argument. The value of the function can be viewed as positive polynomials in the eigenvalues of the matrix. Since the eigenvalue of $(1-h)\bsI_q$ is smaller than that of $\bsI_q$, we get our desired bound. } 
\end{proof}

With the help of the above notations and properties, we are now ready for the proof of Lemma~\ref{lmm:principal_angles}.

\begin{proof}[Proof of Lemma~\ref{lmm:principal_angles}] For simplicity of expression, we denote $\frac{1}{32}10^{-6}\cdot n = an$ and $\arccos\left(\frac{10^{-3}}{4\sqrt{2}}\right)= \theta$.
Let $1\geq \mu_1\geq \cdots\geq \mu_n \geq 0$ such that $\mu_i = \cos^{2}\theta_i$ and $\mu = \cos^{2}\theta$. By Lemma~\ref{lmm:distribution}, we have 

\[\operatorname{dens}\left(\mu_1, \cdots, \mu_n\right)=c_{n, n, d-n} \prod_{i<j}\left|\mu_i-\mu_j\right| \prod_{i=1}^p \mu_i^{-\frac{1}{2}}\left(1-\mu_i\right)^{\frac{1}{2}\left(d-2n-1\right)} .\]

Observe that the determinant of the Jacobian of the change of variables between $\mu_i$s and $\theta_i$s are $\prod_{i=1}^{n} 2\sin\theta_i \cos\theta_i$, to obtain
\[\operatorname{dens}\left(\theta_1, \cdots, \theta_n\right)=2 c_{n, n, d-n} \prod_{i<j}\left|\cos ^2 \theta_i-\cos ^2 \theta_j\right| \prod_{i=1}^n \cos ^{0} \theta_i \sin ^{d-2n} \theta_i.\]

{
The rest of the proof is to calculate the probability we are interested in. We define $\mathcal{T} = \left\{0\leq\theta_{1}\leq \cdots \leq \theta_{a n} < \theta\right\}$ be the set of angles we want and $\mathcal{M} = \left\{1\geq\mu_{1}\geq \cdots \geq \mu_{a n} > \mu\right\}$ be the corresponding set for the cosine. We have 
}

\begin{align*}
\operatorname{Pr}\left(\theta_{a n}<\theta \right) 
 & = \int_{\calT} \operatorname{dens}\left(\theta_1, \cdots, \theta_n\right) \left(\de\theta_1\cdots\theta_n\right) \\
&= \int_{\calT} 2 c_{n, n, d-n} \prod_{i<j}\left|\cos ^2 \theta_i-\cos ^2 \theta_j\right| \prod_{i=1}^n \sin ^{d-2n} \theta_i \left(\de\theta_1\cdots\theta_n \right)\\
&\leq \int_{\calT} 2 c_{n, n, d-n} \prod_{i<j\leq a n}\left|\cos ^2 \theta_i-\cos ^2 \theta_j\right| \prod_{i=1}^{a n} \sin ^{d-2n} \theta_i \left(\de\theta_1\cdots\theta_n\right) \\
&\leq \left(\frac{\pi}{2}\right)^{(1-a)n}\int_{\calT} 2 c_{n, n, d-n} 
\prod_{i<j\leq a n}\left|\cos ^2 \theta_i-\cos ^2 \theta_j\right| \prod_{i=1}^{a n} \sin ^{d-2n} \theta_i \left(\de\theta_1\cdots\theta_{a n}\right)\\
&= \left(\frac{\pi}{2}\right)^{(1-a)n}\int_{\calM} c_{n, n, d-n}\prod_{i<j\leq a n}\left|\mu_i-\mu_j\right| \prod_{i=1}^{a n} \mu_{i}^{-\frac{1}{2}}(1-\mu_i)^{\frac{d-2n-1}{2}} \left(\de\mu_1\cdots\mu_{a n}\right),
\end{align*}
The change of variables $\mu_i = (1 - \mu)t_i + \mu $ gives
\begin{align*}
 \operatorname{Pr}\left(\theta_{a n}<\theta \right)  &\leq \left(\frac{\pi}{2}\right)^{(1-a)n}  c_{n, n, d-n}(1-\mu)^{\frac{a n(a n-1)}{2}}\mu^{-\frac{a n}{2}} (1-\mu)^{\frac{a n(d-2n-1)}{2}}(1-\mu)^{a n}\\
 &\quad \int_{1\geq t_1\geq \cdots\geq t_{a n}\geq 0}\prod_{i<j\leq a n}\left|t_i-t_j\right| \times \prod_{i=1}^{a n} \left(1+\frac{1-\mu}{\mu}t_i\right)^{-\frac{1}{2}}(1-t_i)^{\frac{d-2n-1}{2}} \left(\de t_1\cdots t_{a n}\right)\\
 &=  \left(\frac{\pi}{2}\right)^{(1-a)n}  c_{n, n, d-n}(1-\mu)^{\frac{a n(a n-1)}{2}}\mu^{-\frac{a n}{2}} (1-\mu)^{\frac{a n(d-2n+1)}{2}}\\
 & \quad\frac{2^{a n}}{\vol\left(\boldsymbol{O}_{a n}\right)} \int_{0<\bsY<\bsI_{a n}} \left(\operatorname{det}\left(\bsI_{a n}+\frac{1-\mu}{\mu}\bsI_{a n}\bsY\right)\right)^{-\frac{1}{2}}\times \left(\operatorname{det}(\bsI_{a n}-\bsY)\right)^{\frac{d-2n-1}{2}}\left(\de \bsY\right),
\end{align*}
where $\vol\left(\boldsymbol{O}_q\right)=\int_{\boldsymbol{O}_q} \boldsymbol{A}^{\top} \left(\mathrm{d} \boldsymbol{A}\right)=\frac{2^{q}\pi^{q^2/2}}{\Gamma_q\left(\frac{q}{2}\right)}$ is the volume of the orthogonal group $\boldsymbol{O}_q$ [Page 71 in \cite{muirhead2009aspects}] and the last equation comes from the fact that $\left(\de \bsY\right)=\prod\left|t_i-t_j\right| \left(\de \boldsymbol{T}\left(\boldsymbol{A}^{\top} \left(\de \boldsymbol{A}\right)\right)\right)$, where $\bsY=\boldsymbol{A} \boldsymbol{T} \boldsymbol{A}^{\top}$ is an eigen-decomposition with eigenvalues sorted in non-increasing order, and $\boldsymbol{A}$ cancels out everywhere in the integrated. 
The inequality signs in the integral means PSD ordering and factor $2^{a n}$ appears because the eigen-decomposition is defined up to the choice of the direction of the eigenvectors. This procedure is similar to that in Section 3 of \cite{absil2006largest}.
 
On the other hand, we have the property of Lemma~\ref{lmm:742+743} to have 

\begin{align*}
 \int_{0<\bsY<\bsI_q} & \operatorname{det}(\bsI_q-\bsX \bsY)^{-f}(\operatorname{det} \bsY)^{e-(q+1) / 2} \operatorname{det}(\bsI_q-\bsY)^{g-e-(q+1) / 2} \left( \de \bsY\right) \\
& \quad=\frac{\Gamma_q(e) \Gamma_q(g-e)}{\Gamma_q(g)} \operatorname{det}(\bsI_q-\bsX)^{-f}{ }_2 F_1\left(g-e, f ; g ;-\bsX(\bsI_q-\bsX)^{-1}\right),
\end{align*}
We make the following appointment 
\[\left\{\begin{array}{l}
 q = a n\\
 -f = -\frac{1}{2} \\
 e - (q + 1)/2 = 0\\
g-e- (q + 1)/2 = \frac{d-2n-1}{2}\\
\bsX = -\frac{1-\mu}{\mu} \bsI_{a n}
\end{array}\right. \Rightarrow \left\{\begin{array}{l}
 q = a n\\
 f = \frac{1}{2} \\
 e =  \frac{a n + 1}{2} \\
g = \frac{d-2(1-a)n+1}{2} \\
\bsX = -\frac{1-\mu}{\mu} \bsI_{a n}
\end{array}\right. ,\]
and we get 
\begin{align*}
\operatorname{Pr}\left(\theta_{a n}<\theta \right) \leq & \left(\frac{\pi}{2}\right)^{(1-a)n} \cdot \frac{\pi^{\frac{n^2}{2}}}{\Gamma_n\left(\frac{n}{2}\right)} \cdot \frac{\Gamma_n\left(\frac{d}{2}\right)}{\Gamma_n\left(\frac{n}{2}\right) \Gamma_n\left(\frac{d-n}{2}\right)} \cdot \frac{\Gamma_{a n}\left(\frac{a n}{2}\right)}{ \pi^{(a n)^2/2}}\cdot \frac{\Gamma_{a n}\left(\frac{a n+1}{2}\right) \Gamma_{a n}\left(\frac{d-(2-a)n}{2}\right)}{\Gamma_{a n}\left(\frac{d-2(1-a)n+1}{2}\right)}  \\
& \quad \mu^{-\frac{a n}{2}}(1-\mu)^{\frac{a n(d-(2-a)n)}{2}}
\operatorname{det}\left(\frac{1}{\mu}\bsI_{a n}\right)^{-\frac{1}{2}} \\&\quad { }_2 F_1\left(\frac{d-(2-a)n}{2}, \frac{1}{2} ; \frac{d-2(1-a)n+1}{2} ;(1-\mu)\bsI_{a n}\right).
\end{align*}

We now consider the value of Gaussian hypergeometric function of matrix argument. By Lemma~\ref{lmm:less_than_identity} and the value of the identity matrix given by Lemma~\ref{lmm:32}, we have 

%
%
\begin{align*}
&{ }_2 F_1\left(\frac{d-(2-a)n}{2}, \frac{1}{2} ; \frac{d-2(1-a)n+1}{2} ;(1-\mu)\bsI_{a n}\right)\\&\leq { }_2 F_1\left(\frac{d-(2-a)n}{2}, \frac{1}{2} ; \frac{d-2(1-a)n+1}{2} ;\bsI_{a n}\right) = \frac{\Gamma_{a n}\left(\frac{d-2(1-a)n+1}{2}\right)\Gamma_{a n}\left(\frac{a n}{2}\right)}{\Gamma_{a n}\left(\frac{a n+1}{2}\right)\Gamma_{a n}\left(\frac{d-2(1-a)n}{2}\right)}.
\end{align*}
Bringing it back, we have that 
\begin{align*}
\operatorname{Pr}\left(\theta_{a n}<\theta \right) &\leq \left(\frac{\pi}{2}\right)^{(1-a)n} \cdot \frac{\pi^{\frac{n^2}{2}}}{\Gamma_n\left(\frac{n}{2}\right)} \cdot \frac{\Gamma_n\left(\frac{d}{2}\right)}{\Gamma_n\left(\frac{n}{2}\right) \Gamma_n\left(\frac{d-n}{2}\right)} \cdot \frac{\Gamma_{a n}\left(\frac{a n}{2}\right)}{ \pi^{(a n)^2/2}} \\&\quad \cdot \frac{\Gamma_{a n}\left(\frac{a n}{2}\right) \Gamma_{a n}\left(\frac{d-(2-a)n}{2}\right)}{\Gamma_{a n}\left(\frac{d-2(1-a)n}{2}\right)} (1-\mu)^{\frac{a n(d-(2-a)n)}{2}}
\\&\leq  \left(\frac{\pi}{2}\right)^{(1-a)n} \cdot \pi^{\frac{(1-a^2)n^2}{2}} \cdot \frac{\Gamma_{n}\left(\frac{d}{2}\right)}{\Gamma_{n}\left(\frac{d-n}{2}\right)} \left(\Gamma_{a n}\left(\frac{a n}{2}\right)\right)^2 \cdot (1-\mu)^{\frac{a n(d-(2-a)n)}{2}}.
\end{align*}
By Definition~\ref{lmm:2112} and identities $\Gamma\left(\frac{1}{2}\right)= \sqrt{\pi}$ and  $\Gamma\left(m+1\right)= m \Gamma\left(m\right)$, we would have that 

\begin{align*}
\frac{\Gamma_{n}\left(\frac{d}{2}\right)}{\Gamma_{n}\left(\frac{d-n}{2}\right)} &= \frac{\prod_{i=1}^n \Gamma\left(\frac{d-n}{2}+\frac{i}{2}\right)}{\prod_{i=1}^n \Gamma\left(\frac{d-2n}{2}+\frac{i}{2}\right)} = \prod_{i=1}^n \left(\frac{d-n-1}{2}+\frac{i}{2}\right)\cdots \left(\frac{d-2n}{2}+\frac{i}{2}\right)\\&\leq \left(\frac{d-1}{2}\right)^{n^2} = \exp\left(O(n^2 \log d)\right).
\end{align*}
\begin{align*}
\left(\Gamma_{a n}\left(\frac{a n}{2}\right)\right)^2 &= \pi^{\frac{a n(a n-1)}{2}} \left(\prod_{i=1}^{a n} \Gamma\left(\frac{i}{2}\right)\right)^2\leq \pi^{\frac{a n(a n-1)}{2}} \left( \Gamma\left(\frac{a n}{2}\right)\right)^{2a n}\\&\leq \pi^{\frac{a n(a n-1)}{2}} \left(\sqrt{\pi} \right)^{2a n}\left(\frac{an - 1}{2}\right)^{a^2 n^2} = \exp\left(O(n^2 \log n)\right).
\end{align*}
Combining together, we would have that 

\begin{align*}
\operatorname{Pr}\left(\theta_{a n}<\theta \right) &\leq \exp\left(O(n^2 \log n)\right)\cdot \exp\left(O(n^2 \log d)\right)  \cdot (1-\mu)^{\frac{a nd}{2}} 
\\& \leq \exp\left(-\frac{a}{2} \log\left(\frac{1}{1-\mu}\right)\cdot nd + O(n^2 \log n)+O(n^2 \log d)\right).
\end{align*}
In our case, $a =\frac{1}{32}10^{-6} $ and $\mu = \cos^{2} \theta = \frac{1}{32}10^{-6}$ this should be 
\begin{align*}
&\operatorname{Pr}\left(\theta_{\frac{1}{32}10^{-6}\cdot n} <  \arccos\left(\frac{10^{-3}}{4\sqrt{2}}\right)\right) \\&\leq \exp\left(-\frac{1}{64}10^{-6} \log\left(\frac{1}{1- \frac{1}{32}10^{-6}}\right)\cdot nd + O(n^2 \log n)+O(n^2 \log d)\right) \\ &<  \exp\left(-\frac{1}{128}10^{-6}\log\left(\frac{1}{1- \frac{1}{32}10^{-6}}\right)\cdot nd \right).
\end{align*}
where the last equation holds for $d = \Omega(n \log d)$ with sufficiently large constant.
\end{proof}

\subsection{Extension to General $z\geq 1$}
\label{sec:generalize_power_high}

In this section, we generalize the lower bound to arbitrary powers $z\geq 1$. We again focus on the cases where $k=2$ as we can generalize the result to higher $k$ in Section~\ref{sec:generalize_k}. 
For ease of analysis, we again do not require the construction of datasets $\bsP\subseteq [\Delta]^d$ and actually ensure that every $\bsP$ consists of orthonormal bases of some subspaces.
At the end of the proof, we will show how to round and scale these datasets $\bsP$ into $[\Delta]^d$. 
The cost function we have now without scaling is 

\begin{align*}
\operatorname{cost}_z(\bsP, \{\bsc,-\bsc\})&=\sum_{i=1}^n \min \left\{\dist(\bsp_{i}, \bsc)^z,\dist(\bsp_{i}, -\bsc)^z \right\}= \sum_{i=1}^n \left(\left\|\bsp_{i}\right\|_2^2+\left\|\bsc\right\|_2^2 - 2\left | \langle \bsp_{i},\bsc\rangle\right|\right)^{\frac{z}{2}}\\&=\sum_{i=1}^n \left(2- 2\left | \langle \bsp_{i},\bsc\rangle\right|\right)^{\frac{z}{2}}. 
\end{align*}
The value is upper bounded by $2^{\frac{z}{2}} n$ and the difference between the cost of two datasets is 
\[\left| \operatorname{cost}_z(\bsP, \{\bsc,-\bsc\}) -\operatorname{cost}_z(\bsQ, \{\bsc,-\bsc\})\right| = \left| \sum_{i=1}^n \left(2- 2\left | \langle \bsp_{i},\bsc\rangle\right|\right)^{\frac{z}{2}} - \sum_{i=1}^n \left(2- 2\left | \langle \bsq_{i},c\rangle\right|\right)^{\frac{z}{2}}\right|.\]
Similarly, to prove a lower bound on $\spc(n, \Delta, 2,z,d, \varepsilon)$, we just need to find a large enough set $\mathcal{P}$ such that the cost function of any two elements is separated. The key tool is Lemma~\ref{lmm:inequation_for_arbitrary_power}. 

\begin{lemma}[Inequality of Taylor expansion]
\label{lmm:inequation_for_arbitrary_power} 
For any $0 < z \leq 2$ and any $x \in  \left[0, \frac{1}{2}\right]$, we have
\[1-{\frac{z}{2}}x - z\left(1-{\frac{z}{2}}\right)x^2 \leq (1-x)^{\frac{z}{2}} \leq  1-{\frac{z}{2}}x  . \]
For any $ z\geq 2$ and any $x \in  \left[0, \frac{1}{2}\right]$, we have
\[1-{\frac{z}{2}}x \leq (1-x)^{\frac{z}{2}} \leq   1-{\frac{z}{2}}x + {\frac{z}{2}}\left({\frac{z}{2}}-1\right)x^2. \]
\end{lemma}

\begin{proof}
The detailed proof is in Appendix~\ref{appendix}. The inequation is 
essentially obtained by ignoring the higher-order terms in the Taylor expansion. The main idea is to calculate the first and second order of the function, and then use monotonicity to derive our final bound. 
\end{proof}

{
In our previous conclusions, we have already identified a sufficiently large family such that for any two datasets, when $z = 2$, the value of the cost function shows a significant difference.
\[\left| \operatorname{cost}_2(\bsP, \{\bsc,-\bsc\}) -\operatorname{cost}_2(\bsQ, \{\bsc,-\bsc\})\right| = \left| \sum_{i=1}^n \left | \langle \bsp_{i},\bsc\rangle\right| - \sum_{i=1}^n \left | \langle \bsq_{i},c\rangle\right|\right| \geq \frac{1}{2}\sqrt{n}.\]
All we need to do is use our Lemma~\ref{lmm:inequation_for_arbitrary_power} to expand the value of $\operatorname{cost}_z$ so that the result we obtain for $z=2$ an be generalized to any $z$. The scaling process would then be standard. 
}

\begin{lemma}\label{lmm:arbitraty_power}Assume $\bsP$ and $\bsQ$ being two set of $n$ orthonormal bases in $\mathbb{R}^{d}$ (with $2n < d$). If we can find $\hat{c}$ to satisfy 
\[\left| \sum_{i=1}^{n} \left|\langle \bsp_{i},\hat{\bsc}\rangle\right| - \sum_{i=1}^{n}\left| \langle \bsq_{i},\hat{\bsc}\rangle\right|  \right| > \frac{1}{2} \sqrt{n}, \]
where $\left\|\hat{\bsc}\right\| = 1$. We would then be able to find unit-norm vector $\bsc$ such that 
\[\left|\sum_{i=1}^{n}\left(2-2\left|\langle \bsp_{i},\bsc\rangle\right|\right)^{\frac{z}{2}} - \sum_{i=1}^{n}\left(2-2\left|\langle \bsq_{i},\bsc\rangle\right|\right)^{\frac{z}{2}} \right| \geq \left\{\begin{matrix}
 \frac{2^{\frac{z}{2}} {z}}{8}\sqrt{n} - \frac{2^{\frac{z}{2}} {z}\left(1-{\frac{z}{2}}\right)}{4} ,0<z\leq 2\\
\frac{2^{\frac{z}{2}} {z}}{8}\sqrt{n} - \frac{2^{\frac{z}{2}} {z}\left({\frac{z}{2}}-1\right)}{8} ,z\geq 2
\end{matrix}\right.
\]
\end{lemma}

\begin{proof}
Without loss of generality, we can assume that 
\[ \sum_{i=1}^{n} \left|\langle \bsp_{i},\hat{\bsc}\rangle\right| - \sum_{i=1}^{n}\left| \langle \bsq_{i},\hat{\bsc}\rangle\right|   > \frac{1}{2} \sqrt{n}. \]
We can choose $\tilde{\bsc} = \frac{1}{2} \hat{\bsc}$ such that 
\[\left|\langle \bsp_{i},\tilde{\bsc}\rangle\right| \leq \left\| \bsp_{i}\right\|\left\| \tilde{\bsc}\right\| \leq \frac{1}{2}, \left|\langle \bsq_{i},\tilde{\bsc}\rangle\right| \leq \left\| \bsq_{i}\right\|\left\| \tilde{\bsc}\right\| \leq \frac{1}{2},\forall i\in[n],\]
which fulfills the requirement of Lemma \ref{lmm:inequation_for_arbitrary_power}. For $0<z\leq 2$, we have that
\begin{align*}
&\sum_{i=1}^{n}\left(1-\left|\langle \bsq_{i},\tilde{\bsc}\rangle\right|\right)^{\frac{z}{2}}- \sum_{i=1}^{n}\left(1-\left|\langle \bsp_{i},\tilde{\bsc}\rangle\right|\right)^{\frac{z}{2}} 
 \\&\geq \sum_{i=1}^{n} \left(1 - {\frac{z}{2}}\left|\langle \bsq_{i},\tilde{\bsc}\rangle\right| - z\left(1-{\frac{z}{2}}\right)\langle \bsq_{i},\tilde{\bsc}\rangle^2\right) - \sum_{i=1}^{n}\left(1-{\frac{z}{2}}\left|\langle \bsp_{i},\tilde{\bsc}\rangle\right|\right) \\
&= {\frac{z}{2}}\left(\sum_{i=1}^{n} \left|\langle \bsp_{i},\tilde{\bsc}\rangle\right| - \sum_{i=1}^{n}\left| \langle \bsq_{i},\tilde{\bsc}\rangle\right|\right) - z\left(1-{\frac{z}{2}}\right) \sum_{i=1}^{n} \langle \bsq_{i},\tilde{\bsc}\rangle^2.
\end{align*}
Based on our choice and the fact that $\bsq_{i}$ is a set of $n$ orthonormal bases, we have 
\begin{align*}
& \sum_{i=1}^{n} \left|\langle \bsp_{i},\tilde{\bsc}\rangle\right| - \sum_{i=1}^{n}\left| \langle \bsq_{i},\tilde{\bsc}\rangle\right| = \frac{1}{2} \left(\sum_{i=1}^{n} \left|\langle \bsp_{i},\hat{\bsc}\rangle\right| - \sum_{i=1}^{n}\left| \langle \bsq_{i},\hat{\bsc}\rangle\right| \right)  >  \frac{1}{4}\sqrt{n}, \\
& \sum_{i=1}^{n} \langle \bsq_{i},\tilde{\bsc}\rangle^2 \leq \left\|\tilde{\bsc}\right\|_2^2 = \frac{1}{4}.
\end{align*}
We finally achieve 
\[\sum_{i=1}^{n}\left(1-\left|\langle \bsq_{i},\tilde{\bsc}\rangle\right|\right)^{\frac{z}{2}} - \sum_{i=1}^{n}\left(1-\left|\langle \bsp_{i},\tilde{\bsc}\rangle\right|\right)^{\frac{z}{2}} \geq \frac{z}{8}\sqrt{n} - \frac{z\left(1-{\frac{z}{2}}\right)}{4}.\]
Since $d > 2n$, we can always find a vector $\bar{\bsc}$ such that $\bar{\bsc} \perp \bsp_{i},\bsq_{i}, i \in [n]$ and $\left\| \bar{\bsc} \right\|_2^2 = 1- \left\| \tilde{\bsc} \right\|_2^2 \geq 0$. 
Let $\bsc = \tilde{\bsc} + \bar{\bsc}$, we should have that $\bsc$ is of unit norm and 
\begin{align*}
\sum_{i=1}^{n}\left(1-\left|\langle \bsq_{i},\bsc\rangle\right|\right)^{\frac{z}{2}} - \sum_{i=1}^{n}\left(1-\left|\langle \bsp_{i},\bsc\rangle\right|\right)^{\frac{z}{2}} &=  \sum_{i=1}^{n}\left(1-\left|\langle \bsq_{i},\tilde{\bsc}\rangle\right|\right)^{\frac{z}{2}} - \sum_{i=1}^{n}\left(1-\left|\langle \bsp_{i},\tilde{\bsc}\rangle\right|\right)^{\frac{z}{2}} \\& \geq \frac{z}{4}\varepsilon n - \frac{z\left(1-{\frac{z}{2}}\right)}{4}. \end{align*}
Therefore, 
\begin{align*}
\left|\sum_{i=1}^{n}\left(2-2\left|\langle \bsp_{i},\bsc\rangle\right|\right)^{\frac{z}{2}} - \sum_{i=1}^{n}\left(2-2\left|\langle \bsq_{i},\bsc\rangle\right|\right)^{\frac{z}{2}} \right| \geq \frac{2^{\frac{z}{2}} z}{8}\sqrt{n} - \frac{2^{\frac{z}{2}} z\left(1-{\frac{z}{2}}\right)}{4}.
\end{align*}
On the other hand, for $z\geq 2$, we have 
\begin{align*}
&\sum_{i=1}^{n}\left(1-\left|\langle \bsq_{i},\tilde{\bsc}\rangle\right|\right)^{\frac{z}{2}} - \sum_{i=1}^{n}\left(1-\left|\langle \bsp_{i},\tilde{\bsc}\rangle\right|\right)^{\frac{z}{2}} 
\\& \geq \sum_{i=1}^{n} \left(1 - {\frac{z}{2}}\left|\langle \bsq_{i},\tilde{\bsc}\rangle\right|\right) - \sum_{i=1}^{n}\left(1-{\frac{z}{2}}\left|\langle \bsp_{i},\tilde{\bsc}\rangle\right| +  {\frac{z}{2}}\left({\frac{z}{2}}-1\right)\langle \bsp_{i},\tilde{\bsc}\rangle^2\right) \\
&= {\frac{z}{2}}\left(\sum_{i=1}^{n} \left|\langle 
\bsp_{i},\tilde{\bsc}\rangle\right| - \sum_{i=1}^{n}\left| \langle \bsq_{i},\tilde{\bsc}\rangle\right|\right) - {\frac{z}{2}}\left({\frac{z}{2}}-1\right) \sum_{i=1}^{n} \langle \bsp_{i},\tilde{\bsc}\rangle^2 \\&\geq \frac{z}{8}\sqrt{n} - \frac{z\left({\frac{z}{2}}-1\right)}{8}.
\end{align*}
Since $d > 2n$, we can always find a vector $\bar{\bsc}$ such that $\bar{\bsc} \perp \bsp_{i},\bsq_{i}, i \in[n]$ and $\left\| \bar{\bsc} \right\|_2^2 = 1- \left\| \tilde{\bsc} \right\|_2^2 \geq 0$. 
Let $\bsc = \tilde{\bsc} + \bar{\bsc}$, we should have that $\bsc$ is of unit norm and 
\begin{align*}
\sum_{i=1}^{n}\left(1-\left|\langle \bsq_{i},\bsc\rangle\right|\right)^{\frac{z}{2}} - \sum_{i=1}^{n}\left(1-\left|\langle \bsp_{i},\bsc\rangle\right|\right)^{\frac{z}{2}} &=  \sum_{i=1}^{n}\left(1-\left|\langle \bsq_{i},\tilde{\bsc}\rangle\right|\right)^{\frac{z}{2}} - \sum_{i=1}^{n}\left(1-\left|\langle \bsp_{i},\tilde{\bsc}\rangle\right|\right)^{\frac{z}{2}}\\& \geq \frac{z}{8}\sqrt{n} - \frac{z\left({\frac{z}{2}}-1\right)}{8}.
\end{align*}
Finally, we have
\begin{align*}
\left|\sum_{i=1}^{n}\left(2-2\left|\langle \bsp_{i},\bsc\rangle\right|\right)^{\frac{z}{2}} - \sum_{i=1}^{n}\left(2-2\left|\langle \bsq_{i},\bsc\rangle\right|\right)^{\frac{z}{2}} \right| \geq \frac{2^{\frac{z}{2}} z}{8}\sqrt{n} - \frac{2^{\frac{z}{2}} z\left({\frac{z}{2}}-1\right)}{8}.
\end{align*}
\end{proof}
With the results of Lemma~\ref{lmm:arbitraty_power}, we just need to perform scaling to finally generalize to arbitrary $z$. 

\begin{lemma}\label{lmm:setsize_to_bit_z_high} When $\Delta = \frac{3072\cdot 2^{\frac{z}{2}}\sqrt{d}}{z^2\cdot \varepsilon} = \Theta\left(\frac{\sqrt{d}}{ \varepsilon}\right)$, $d$ larger than a large enough constant and constant $z$, we have that $ \spc(n,\Delta,2,z,d,\varepsilon)
\geq \Omega\left(d\min\left\{\frac{1}{\varepsilon^2},\frac{d}{\log d},n\right\}\right)$. 
\end{lemma}

\begin{proof}
Let $\varepsilon = \frac{z}{96}\tilde{\varepsilon}$. Denote 
\[\tilde{n} = \min\left\{\Theta\left(\frac{1}{\tilde{\varepsilon}^2}\right),\Theta\left(\frac{d}{\log d}\right),n\right\} = \Omega\left(\min\left\{\frac{1}{\varepsilon^2},\frac{d}{\log d},n\right\}\right) \geq 100.\]
With the proof of second part of Theorem~\ref{thm:main_lower} and Lemma \ref{lmm:arbitraty_power}, for constant $z$, we are able to find a set $\mathcal{P}$ with all points being orthonormal bases and size being 
\[\exp\left(\Theta\left(\tilde{n}d\right)\right)=\exp\left(\Theta\left(d\min\left\{\frac{1}{\tilde{\varepsilon}^2},\frac{d}{\log d},n\right\}\right)\right) = \exp\left(\Theta\left(d\min\left\{\frac{1}{\varepsilon^2},\frac{d}{\log d},n\right\}\right)\right),\]
such that for any two dataset in this set $\bsP$ and $\bsQ$, we would be able to find a unit-norm vector $\bsc$ satisfying 
\[\left|\sum_{i=1}^{n}\left(2-2\left|\langle \bsp_{i},\bsc\rangle\right|\right)^z - \sum_{i=1}^{n}\left(2-2\left|\langle \bsq_{i},\bsc\rangle\right|\right)^z \right| \geq \left\{\begin{matrix}
 \frac{2^{\frac{z}{2}} {z}}{8}\sqrt{\tilde{n}} - \frac{2^{\frac{z}{2}} {z}\left(1-{\frac{z}{2}}\right)}{4} ,0<z\leq 2\\
\frac{2^{\frac{z}{2}} {z}}{8}\sqrt{\tilde{n}} - \frac{2^{\frac{z}{2}} {z}\left({\frac{z}{2}}-1\right)}{8} ,z\geq 2
\end{matrix}\right.. \]
We now show how to round and scale every dataset $\bsP\in \calP$ to $[\Delta]^{d}$, where $\Delta = \Theta\left(\frac{\sqrt{d}}{ \varepsilon}\right)$.
Without loss of generality, we may assume that $\Delta$ is an odd integer.
For a dataset $\bsP = \left(\bsp_{1},\cdots,\bsp_{\tilde{n}}\right)\in \calP$, we will construct $\tilde{\mathcal{P}}$ to be our final family as follows: 
For each of dataset $\bsP\in\mathcal{P}$, we shift the origin to $\left(\lceil\frac{\Delta}{2}\rceil,\cdots, \lceil\frac{\Delta}{2}\rceil\right)$, scale it by a factor of $\frac{\Delta}{2}$ and finally perform an upward rounding on each dimension to put every point on the grid: 
\[\tilde{\bsP} = \left(\lceil\frac{\Delta}{2}\bsp_{1}\rceil + \lceil\frac{\Delta}{2}\rceil,\cdots,\lceil\frac{\Delta}{2}\bsp_{\tilde{n}}\rceil+ \lceil\frac{\Delta}{2}\rceil\right)= \left(\tilde{\bsp}_{1},\cdots,\tilde{\bsp}_{\tilde{n}}\right).\]
We will then show that this set fulfills the requirement of Lemma~\ref{lmm:setsize_to_bit}. For ease of explanation, we also define $\hat{P}$ to be the dataset without rounding. 
\begin{align*} \hat{\bsP} = \left(\frac{\Delta}{2}\bsp_{1} + \lceil\frac{\Delta}{2}\rceil,\cdots,\frac{\Delta}{2}\bsp_{\tilde{n}}+ \lceil\frac{\Delta}{2}\rceil\right) = \left(\hat{\bsp}_{1},\cdots,\hat{\bsp}_{\tilde{n}}\right).
\end{align*}
Moreover, let $\bar{\bsc} = \frac{\Delta}{2}\bsc+ \left(\lceil\frac{\Delta}{2}\rceil,\cdots,\lceil\frac{\Delta}{2}\rceil\right)$. 
We must have that for the scaling dataset, 
\begin{align*}
&\operatorname{cost}_z\left(\hat{\bsP}, \left\{\bar{\bsc},-\bar{\bsc}\right\}\right) = \frac{\Delta^z}{2^{z}} \operatorname{cost}_z\left(\bsP, \{\bsc,-\bsc\}\right) \leq \frac{\Delta^z \tilde{n}}{2^{\frac{z}{2}}},
\\& \operatorname{cost}_z\left(\hat{\bsP}, \left\{\bar{\bsc},-\bar{\bsc}\right\}\right) -\operatorname{cost}_z\left(\hat{\bsQ}, \left\{\bar{\bsc},-\bar{\bsc}\right\}\right)\\&=\frac{\Delta^z}{2^{z}}\left(\operatorname{cost}_z(\bsP, \{\bsc,-\bsc\}) -\operatorname{cost}_z(\bsQ, \{\bsc,-\bsc\})\right) \geq  \left\{\begin{matrix}
 \frac{ {z}\Delta^z}{8\cdot 2^{\frac{z}{2}}}\sqrt{\tilde{n}} - \frac{ {z}\left(1-{\frac{z}{2}}\right)\Delta^z}{4\cdot 2^{\frac{z}{2}}} ,0<z\leq 2\\
\frac{ {z}\Delta^z}{8\cdot 2^{\frac{z}{2}}}\sqrt{\tilde{n}}- \frac{ {z}\left({\frac{z}{2}}-1\right)\Delta^z}{8\cdot 2^{\frac{z}{2}}} ,z\geq 2
\end{matrix}\right..
\end{align*}
On the other hand, for the rounding dataset, by our choice of $\Delta$, we have 
\begin{align*}
\left|\left\|\hat{\bsp}_{i}-\bar{\bsc}\right\|_2^z - \left\|\tilde{\bsp}_{i}-\bar{\bsc}\right\|_2^z \right|& \leq \frac{\Delta^z}{2^z}2^z \frac{2\sqrt{d}}{\Delta}\leq \frac{\tilde{\varepsilon} \Delta^z}{64\cdot 2^{\frac{z}{2}}} .  
\end{align*}
The case for $-\bar{\bsc}$ and other datasets is similar. Therefore, we have that for any dataset
\begin{align*}
& \left|\operatorname{cost}_2(\hat{\bsP}, \{\bar{\bsc},-\bar{\bsc}\}) - \operatorname{cost}_2(\tilde{\bsP}, \{\bar{\bsc},-\bar{\bsc}\})\right| \leq \frac{\tilde{\varepsilon} \Delta^z \tilde{n}}{64\cdot 2^{\frac{z}{2}}},
\\& \operatorname{cost}_2\left(\tilde{\bsP}, \left\{\bar{\bsc},-\bar{\bsc}\right\}\right) \leq \operatorname{cost}_2\left(\hat{\bsP}, \left\{\bar{\bsc},-\bar{\bsc}\right\}\right) + \frac{\tilde{\varepsilon} \Delta^z n}{64\cdot 2^{\frac{z}{2}}} \leq \frac{\Delta^z \tilde{n}}{2^{\frac{z}{2}}} + \frac{\tilde{\varepsilon} \Delta^z \tilde{n}}{64\cdot 2^{\frac{z}{2}}} \leq \frac{\Delta^z \cdot 2\tilde{n}}{2^{\frac{z}{2}}}.
\end{align*}
We can then have a rounded family $\tilde{\mathcal{P}}$ such that all points are on the grid and we would be able to find $\bar{\bsc}$, 
\begin{align*}
\operatorname{cost}_2\left(\tilde{\bsP}, \left\{\bar{\bsc},-\bar{\bsc}\right\}\right) -\operatorname{cost}_2\left(\tilde{\bsQ}, \left\{\bar{\bsc},-\bar{\bsc}\right\}\right) & \geq \operatorname{cost}_2\left(\hat{\bsP}, \left\{\bar{\bsc},-\bar{\bsc}\right\}\right) -\operatorname{cost}_2\left(\hat{\bsQ}, \left\{\bar{\bsc},-\bar{\bsc}\right\}\right) - \frac{\tilde{\varepsilon} \Delta^z \tilde{n}}{32\cdot 2^{\frac{z}{2}}} \\&\geq \left\{\begin{matrix}
 \frac{ {z}\Delta^z}{8\cdot 2^{\frac{z}{2}}}\sqrt{\tilde{n}} - \frac{ {z}\left(1-{\frac{z}{2}}\right)\Delta^z}{4\cdot 2^{\frac{z}{2}}}- \frac{\tilde{\varepsilon} \Delta^z \tilde{n}}{32\cdot 2^{\frac{z}{2}}} ,0<z\leq 2\\
\frac{ {z}\Delta^z}{8\cdot 2^{\frac{z}{2}}}\sqrt{\tilde{n}} - \frac{ {z}\left({\frac{z}{2}}-1\right)\Delta^z}{8\cdot 2^{\frac{z}{2}}}- \frac{\tilde{\varepsilon} \Delta^z \tilde{n}}{32\cdot 2^{\frac{z}{2}}} ,z\geq 2
\end{matrix}\right.
\end{align*}
Note that for constant $z$ and our choice of $\tilde{n}$ that $\sqrt{\tilde{n}} \geq \tilde{\varepsilon} \tilde{n}$ the right side is larger than $\frac{ {z}\Delta^z}{16\cdot 2^{\frac{z}{2}}}\tilde{\varepsilon} \tilde{n}$. 
On the other side, the cost function is upper bounded by $\frac{\Delta^z \cdot 2\tilde{n}}{2^{\frac{z}{2}}}$.
Therefore, as long as we make $ \varepsilon$-approximation, we must have 
\[\operatorname{cost}_{z}\left(\bsP,\{\bsc,-\bsc\}\right) \notin (1\pm 3\eps)\operatorname{cost}_{z}\left(\bsQ,\{\bsc,-\bsc\}.\right)\] 
Moreover, we can find an origin being $\left(\lceil\frac{\Delta}{2}\rceil,\cdots, \lceil\frac{\Delta}{2}\rceil\right)$ such that all the center points and data points have distance to it less than $\frac{\Delta}{2}+\sqrt{d} \leq \Delta$.  
By Lemma~\ref{lmm:setsize_to_bit}, we have
\[\spc(n,\Delta,2,z,d,\varepsilon)\geq \Omega(\log \left|\mathcal{P}\right|)\geq \Omega\left(d\min\left\{\frac{1}{\varepsilon^2},\frac{d}{\log d},n\right\}\right).\]
\end{proof}


\subsection{Extension to General $k\geq 2$}
\label{sec:generalize_k}

Without loss of generality, we may assume that $k$ is even. 
(If $k$ is odd, we can use $k-1$ centers and put the rest center in a place not to affect the value of the function but still have a similar asymptotic lower bound on the size.) 
{In our previous proof, we have identified a sufficiently large family for $k = 2$. 
When $k > 2$, our approach is to generate $\frac{k}{2}$ copies of our previous family and ensure that the distance between each copy is sufficiently large so that they do not interfere with each other. 
This way, as long as the two datasets we ultimately select have significant differences in a sufficient number of copies, their overall difference will also be substantial. 
Let $\tilde{\Delta} = \Omega\left(\frac{\sqrt{d}}{\varepsilon}\right)$ be the large enough discretization parameter for $k=2$. 
We have $\Delta \geq 4\lceil k^{\frac{1}{d}}\rceil \tilde{\Delta} = \Omega\left(\frac{k^{\frac{1}{d}}\sqrt{d}}{\varepsilon}\right)$. }

By proof of second part of Theorem~\ref{thm:main_lower} and Lemma~\ref{lmm:setsize_to_bit_z_high}, we can find a set $\mathcal{P}$ with size being $\exp\left(\spc\left(\frac{2n}{k}, \tilde{\Delta}, 2,z,d, \varepsilon\right)\right)$ with property that 
\begin{align*}
& \operatorname{cost}_z(\bsP, \{\bsc,-\bsc\}),\operatorname{cost}_z(\bsQ, \{\bsc,-\bsc\}) \leq  \frac{\tilde{\Delta}^z \cdot 2\cdot\frac{2n}{k}}{2^{\frac{z}{2}}},\\
& \operatorname{cost}_z(\bsP,\{\bsc,-\bsc\}) -\operatorname{cost}_z(\bsQ, \{\bsc,-\bsc\}) \geq  \frac{3\varepsilon\tilde{\Delta}^z \cdot 2\cdot \frac{2n}{k}}{2^{\frac{z}{2}}}.
\end{align*}
Moreover, we can find an origin such that all the center points and data points have the distance to it less than $\frac{\tilde{\Delta}}{2}+\sqrt{d} \leq \tilde{\Delta} =\Theta\left(\frac{\sqrt{d}}{\varepsilon}\right)$.  
The full instance is then made of $\frac{k}{2}$ distinct copies of the $k = 2$ instance, denoted as $\mathcal{P}^{\frac{k}{2}}$. 

{
We first prove that our entire space can accommodate those copies.}
Note that each instance can be wrapped by a hypercube with side length $4\tilde{\Delta}$ by putting its origin at the center of the hypercube. 
In our model, The whole space is a hypercube with side length $\Delta$. 
We can then put at most 
\[\left(\frac{\Delta}{4\tilde{\Delta}}\right)^d \geq \left(\lceil k^{\frac{1}{d}}\rceil\right)^d \geq k.\]
in the large hypercube, which fulfills our requirements. 

Moreover, the distance between the origins of any two different instances is at least $4\tilde{\Delta}$, which means that the points in the two instances will have a distance of at least $2\tilde{\Delta}$ and will not interfere with the assignment of clustering. 

We then force that any two datasets $\bsP, \bsQ \in \mathcal{P}^{\frac{k}{2}}$ is different on at least $\frac{k}{4}$ copies. 
It may be easier for us to think of each dataset as a "vector" of dimension $\frac{k}{2}$, where each entry $i$ denotes the choice of the dataset on the $i$-th copy, and thus there are $\exp\left(\spc\left(\frac{2n}{k}, \tilde{\Delta}, 2,z,d, \varepsilon\right)\right)$ choices. 
Our additional requirement is equal to the condition that $\left\| \bsP - \bsQ \right\|_0 \geq \frac{k}{4}. $
We define $\bsP_i, \bsQ_i$ to be te dataset chosen in the $i$-th copy and we place two centers $\{\bsc_{i},-\bsc_{i}\}$.
The total cost of dataset $\bsP$ is thus 
\[\operatorname{cost}_z(\bsP,\{\bsc_{i},-\bsc_{i}\}_{i\in \left[\frac{k}{2}\right]}) = \sum_{i=1}^{\frac{k}{2}} \operatorname{cost}_z(\bsP_i, \{\bsc_{i},-\bsc_{i}\})\leq \frac{1}{2}k \cdot \frac{\tilde{\Delta}^z \cdot 2\cdot\frac{2n}{k}}{2^{\frac{z}{2}}}.\]
On the other hand, when it comes to the cost difference of two datasets, $\bsP$ and $\bsQ$ are different on at least $\frac{k}{4}$ copies. 
Moreover, by the property of our chosen dataset, we can find $\{\bsc_{i},-\bsc_i\}$ such that the cost on $i$-th copy is at least $3\varepsilon \cdot \frac{\tilde{\Delta}^z \cdot 2\cdot\frac{2n}{k}}{2^{\frac{z}{2}}}$.
We thus have the total cost difference be 
\begin{align*}
\operatorname{cost}_z(\bsP, \{\bsc_{i},-\bsc_{i}\}_{i\in \left[\frac{k}{2}\right]}) -\operatorname{cost}_z(\bsQ, \{\bsc_{i},-\bsc_{i}\}_{i\in \left[\frac{k}{2}\right]}) & = \sum_{i=1}^{\frac{k}{2}} \left[\operatorname{cost}_z(\bsP_i, \{\bsc_i,-\bsc_i\}) - \operatorname{cost}_z(\bsQ, \{\bsc_i,-\bsc_i\}\right] \\&\geq \left\| \bsP - \bsQ \right\|_0 \cdot 3\varepsilon \cdot \frac{\tilde{\Delta}^z \cdot 2\cdot\frac{2n}{k}}{2^{\frac{z}{2}}} \geq \frac{3}{4}\varepsilon k\cdot \frac{\tilde{\Delta}^z \cdot 2\cdot\frac{2n}{k}}{2^{\frac{z}{2}}}.
\end{align*}
This fulfills the requirement of family $\calP$ in  Lemma~\ref{lmm:setsize_to_bit}.

We then consider to compute the number of the dataset. 
Note that without the additional requirement, the number of the dataset is 
\[\left(\exp\left(\spc\left(\frac{2n}{k}, \tilde{\Delta}, 2,z,d, \varepsilon\right)\right)\right)^{\frac{k}{2}} = \exp\left(\frac{k}{2}\spc\left(\frac{2n}{k}, \tilde{\Delta}, 2,z,d, \varepsilon\right)\right).\]
On the other hand, we denote the neighbors of a dataset as those who have differences on less than $\frac{k}{4}$ copies.
For a dataset, the number of neighbors is less than 
\[\sum_{i=1}^{\frac{k}{4}} \begin{pmatrix}
 \frac{k}{2} \\
i
\end{pmatrix} \left(\exp\left(\spc\left(\frac{2n}{k}, \tilde{\Delta}, 2,z,d, \varepsilon\right)\right)\right)^{i}=\sum_{i=1}^{\frac{k}{4}} \begin{pmatrix}
 \frac{k}{2} \\
i
\end{pmatrix} \exp\left(i\cdot \spc\left(\frac{2n}{k}, \tilde{\Delta}, 2,z,d, \varepsilon\right)\right),\]
which means that we first choose $i$ copies to be the different ones (the rest $\frac{k}{2} - i$ copies are then fixed to be the same as the original dataset), and then choose the value on these positions arbitrarily. We can upper bound this value with the Stirling inequation that 
\begin{align*}
\exp(1)\left(\frac{n^n}{\exp(n)}\right) &\leq \sqrt{2 \pi n}\left(\frac{n^n}{\exp(n)}\right) \exp\left(\frac{1}{12n+1}\right) \leq n ! \\&\leq \sqrt{2 \pi n}\left(\frac{n^n}{\exp(n)}\right) \exp\left(\frac{1}{12n}\right) \leq \exp(1) n\left(\frac{n^n}{\exp(n)}\right),
\end{align*}
which gives 
\begin{align*}
& \sum_{i=1}^{\frac{k}{4}} \begin{pmatrix}
 \frac{k}{2} \\
i
\end{pmatrix} \exp\left(i\cdot  \spc\left(\frac{2n}{k}, \tilde{\Delta}, 2,z,d, \varepsilon\right)\right)  \leq \frac{k}{4} \begin{pmatrix}
 \frac{k}{2} \\
\frac{k}{4}
\end{pmatrix} \exp\left(\frac{k}{4}\spc\left(\frac{2n}{k}, \tilde{\Delta}, 2,z,d, \varepsilon\right)\right) \\&=\frac{k}{4}  \frac{\left(\frac{k}{2}\right)!}{\left(\left(\frac{k}{4}\right)!\right)^2}\exp\left(\frac{k}{4}\spc\left(\frac{2n}{k}, \tilde{\Delta}, 2,z,d, \varepsilon\right)\right) 
\leq \frac{k}{4} \frac{e \frac{k}{2}\left(\frac{k}{2e}\right)^{\frac{k}{2}}}{e^2 \left(\frac{k}{4e}\right)^{\frac{k}{2}}}\exp\left(\frac{k}{4}\spc\left(\frac{2n}{k}, \tilde{\Delta}, 2,z,d, \varepsilon\right)\right) \\& = 
\exp\left(\frac{k}{4}\spc\left(\frac{2n}{k}, \tilde{\Delta}, 2,z,d, \varepsilon\right)+\frac{k \ln 2}{2} + \ln \left(\frac{k^2}{8e}\right)\right).
\end{align*}
The size of the final dataset is at least the total number divided by the number of neighbors and thus greater than 
\begin{align*}
&\frac{\exp\left(\frac{k}{2}\spc\left(\frac{2n}{k}, \tilde{\Delta}, 2,z,d, \varepsilon\right)\right)}{\exp\left(\frac{k}{4}\spc\left(\frac{2n}{k}, \tilde{\Delta}, 2,z,d, \varepsilon\right)+\frac{k \ln 2}{2} + \ln \left(\frac{k^2}{8e}\right)\right)} \\&= \exp\left(\frac{k}{4}\spc\left(\frac{2n}{k}, \tilde{\Delta}, 2,z,d, \varepsilon\right)-\frac{k \ln 2}{2} -\ln \left(\frac{k^2}{8e}\right)\right) \geq \exp\left(\frac{k}{8}\spc\left(\frac{2n}{k}, \tilde{\Delta}, 2,z,d, \varepsilon\right)\right),
\end{align*}
where the last equation holds for $\spc\left(\frac{2n}{k}, \tilde{\Delta}, 2,z,d, \varepsilon\right) \geq 4\ln 2 + \frac{8}{k}\ln \left(\frac{k^2}{8e}\right)$ larger than a large enough constant.

Therefore, we can find a family $\mathcal{P}^{\frac{k}{2}}$ with size being at least $\exp\left(\frac{k}{8}\spc\left(\frac{2n}{k}, \tilde{\Delta}, 2,z,d, \varepsilon\right)\right)$ satisfying the requirement of Lemma~\ref{lmm:setsize_to_bit}. We thus have,
\begin{align*}
\spc\left(n, \Delta, k,z,d, \frac{1}{2}\varepsilon\right) &=\Omega\left(\log \left|\mathcal{P}^{\frac{k}{2}}\right|\right)= \Omega\left(k\cdot \spc\left(\frac{2n}{k}, \tilde{\Delta}, 2,z,d, \varepsilon\right)\right) \\&= \Omega\left(kd\min\left\{\frac{1}{\varepsilon^2},\frac{d}{\log d},\frac{n}{k}\right\}\right). 
\end{align*}

\section{Application to Space Lower Bound for Terminal Embedding}
\label{sec:application}

Recall that the definition of Terminal Embedding is given as

\begin{definition}[Restatement of Definition~\ref{def:embedding}]
Let $\eps\in (0,1)$ and $\bsP$ be a dataset of $n$ points. 
	A mapping $\tau: \R^d\rightarrow \R^m$ is called an $\eps$-terminal embedding of $\bsP$ 
 if for any $\bsp\in \bsP$ and $\bsq\in \R^d$,
	$
	\dist(\bsp,\bsq)\leq \dist(\tau(\bsp),\tau(\bsq))\leq (1+\eps)\cdot \dist(\bsp,\bsq).
	$
\end{definition}
In the case that $\mathcal{X}$ and $\mathcal{Y}$ are both Euclidean metrics with $\mathcal{Y}$ being lower-dimensional, work of\ \cite{narayanan2019optimal} prove that the dimension of latter space can be as small as $O\left(\eps^{-2} \log n\right)$, which is optimal proven in \cite{larsen2017optimality}. 

In the following part, we will show that our lower bound on the minimum number of bits of computing the cost function actually sheds light on the number of bits of terminal embedding. 

\begin{definition}[Space complexity for terminal embedding] \label{def:space_terminal} Let $\bsP\subset \mathbb{R}^d$ be a dataset, $n\geq 1$ and $\eps>0$ be an error parameter.
We define $\embed(\bsP, d, \eps)$ to be the minimum possible number of bits of a $\eps$-terminal embedding from $\bsP$ into $\mathbb{R}^m$ with $m=O\left(\eps^{-2} \log n\right)$. 
Moreover, we define $\embed(n, d, \eps):= \sup_{\bsP\subset\mathbb{R}^d: |\bsP| = n} \embed(\bsP, d, \eps)$ to be the space complexity function, i.e., the maximum cardinality $\embed(\bsP, d, \eps)$ over all possible datasets $\bsP\subset\mathbb{R}^d$ of size $n$.
\end{definition}

{
Given the definition, we have the space complexity lower bound for terminal embedding. Our proof idea is to use terminal embedding to construct an algorithm for computing $\kzC$. Since the space complexity of $\kzC$ is very large, the terminal embedding will also have a significant space complexity.}

\begin{theorem}[Space lower bound for terminal embedding] 
\label{thm:embedding} Let $\eps\in (0,1)$ and assume $d = \Omega\left(\frac{\log n \log(n/\eps)}{\eps^2 }\right)$. 
The space complexity of terminal embedding $\embed(n,d,\eps) = \Omega(nd)$.
\end{theorem}

\begin{proof}

We will show that an $\eps$-sketch can be constructed for the case of our lower bound by using terminal embedding. 
We would consider the case where $z=2$ and $k$ is large enough such that $\frac{n}{k} \leq \min\left\{\frac{1}{\varepsilon^2}, \frac{d}{\log d}\right\}, k =O\left(n\right)$. 
Note that $\Delta = \Theta\left(\frac{k^{\frac{1}{d}}\sqrt{d}}{\varepsilon}\right) = O\left(\frac{n^{\frac{1}{d}}\sqrt{d}}{\varepsilon}\right)$.
Moreover, in our proof of lower bound, since we are only considering the scaled orthonormal bases, we would have that for any considered datasets $\bsP$ and center $\{\bsc_{j}\}_{j \in [2]}$. 
\[\operatorname{cost}_{2}\left(\bsP,\{\bsc_{j}\}_{j \in [2]}\right) \geq \frac{\Delta^2}{4} \left(2n-\sqrt{n}\right) \geq \Omega\left(\Delta^2 n\right). \]

To construct the $\eps$-sketch, we first use the terminal embedding to lower the dimension to $\frac{\log n}{\eps^2}$. 
For a dataset $\bsP = \left(\bsp_{1},\cdots,\bsp_{n}\right)$, Let $\hat{\bsP} = \left(\tau\left(\bsp_1\right),\cdots,\tau\left(\bsp_n\right)\right)$ be the embedded dataset. 
By the property of terminal embedding, we have for any center point $\bsc \in \mathbb{R}^d,i\in [n]$, 
\begin{align*}
 \dist(\bsp_i,\bsc)\leq \dist(\tau(\bsp_i),\tau(\bsc))&\leq (1+\eps)\cdot \dist(\bsp_i,\bsc), \dist(\bsp_i,\bsc)^2\leq \dist(\tau(\bsp_i),\tau(\bsc))^2\\&\leq (1+3\eps)\cdot \dist(\bsp_i,\bsc)^2.
 \end{align*}
 Therefore, for the cost function we have 
 \begin{align*}
\operatorname{cost}_{2}\left(\bsP,\{\bsc_{j}\}_{j \in [2]}\right) \leq \operatorname{cost}_{2}\left(\hat{\bsP},\{\tau(\bsc_{j})\}_{j \in [2]}\right) \leq (1+3\eps)\operatorname{cost}_{2}\left(\bsP,\{\bsc_{j}\}_{j \in [2]}\right).
\end{align*}
We then round the dataset to the grid points 
$\left[\Delta\right]^{\frac{\log n}{\varepsilon^2}}$. Let the rounded dataset be $$\tilde{\bsP} = \left(\lceil \tau\left(\bsp_1\right)\rceil,\cdots,\lceil \tau\left(\bsp_n\right)\rceil\right).$$ 
We would have that 
\begin{align*}
\left|\left\|\hat{\bsp}_{i}-\tau(\bsc)\right\|_2^2 - \left\|\tilde{\bsp}_{i}-\tau(\bsc)\right\|_2^2 \right|& \leq 2\left\|\hat{\bsp}_{i}-\tilde{\bsp}_{i}\right\|_2\left\|\tilde{\bsp}_{i}-\tau(\bsc)\right\|_2 + \left\|\hat{\bsp}_{i}-\tilde{\bsp}_{i}\right\|_2^2 \\& \leq \Delta\sqrt{\frac{\log n}{\eps^2}} + \frac{\log n}{\eps^2} \leq O\left(\Delta^2 \varepsilon\right), 
\end{align*}
where the last inequation holds due to our choice of $\Delta$. 
Therefore, we still have that 
\begin{align*}
&\left|\operatorname{cost}_{2}\left(\hat{\bsP},\{\tau(\bsc_{j})\}_{j\in [2]}\right)-\operatorname{cost}_{2}\left(\tilde{\bsP},\{\tau(\bsc_{j})\}_{j \in [2]}\right)\right|\leq O\left(\Delta^2 \varepsilon n \right) \leq \varepsilon \operatorname{cost}_{2}\left(\bsP,\{\bsc_{j}\}_{j \in [2]}\right),
\\& (1-\eps)\operatorname{cost}_{2}\left(\bsP,\{\bsc_{j}\}_{j \in [2]}\right) \leq \operatorname{cost}_{2}\left(\tilde{\bsP},\{\tau(\bsc_{j})\}_{j \in [2]}\right) \leq (1+4\eps)\operatorname{cost}_{2}\left(\bsP,\{\bsc_{j}\}_{j \in [2]}\right).
\end{align*}
We then only need to construct an $\eps$-sketch for $\operatorname{cost}_{2}\left(\tilde{\bsP},\{\tau(\bsc_{j})\}_{j \in [2]}\right)$ in the lower dimension of $ \frac{\log n}{\eps^2}$. 
With the result of Corollary \ref{thm:main_upper}, the sketch only needs to use 
\begin{align*}
\spc\left(n, \Delta, k, 2, \frac{\log n}{\eps^2}, \eps\right) \leq &
\frac{k\log n\log\Delta}{\eps^2} + \Psi(n)\left(\frac{\log n\left(\log \left(\log \Delta/\varepsilon\right)\right)}{\eps^2} + \log\log n \right)
\\\leq&  2n \frac{\log n\log \left( n/\varepsilon  \right)}{\varepsilon^2},
\end{align*}
where we bring in $\Delta=\Theta\left(\frac{k^{\frac{1}{d}}\sqrt{d}}{\eps}\right), k,\Psi(n) \leq n$. 
Combining together, our sketch exploits bits of only 
\[\spc\left(n, \Delta, k, 2, d, \eps\right) \leq \embed(n, d, \eps)+2n \frac{\log n\log \left( n/\varepsilon  \right)}{\varepsilon^2}.\]
On the other hand, with the result of the second part of Theorem~\ref{thm:main_lower}, we have that for $d = {\Omega}\left(\frac{\log n\log \left( n/\varepsilon  \right)}{\varepsilon^2}\right)$ and our choice of $n$,
\begin{align*}
\spc\left(n, \Delta, k, 2, d, \eps\right)&\geq \Omega\left(kd \min\left\{\frac{1}{\varepsilon^2},\frac{d}{\log d}, \frac{n}{k}\right\}\right) = \Omega\left(nd \right) \geq \Omega\left(n \frac{\log n\log \left( n/\varepsilon  \right)}{\varepsilon^2}\right).
\end{align*}
Therefore, we must have that 
\begin{align*}
&\embed(n, d, \eps) \geq \Omega\left(nd\right) - 2n\frac{\log n\log \left( n/\varepsilon  \right)}{\varepsilon^2}\geq \Omega\left(nd\right).
\end{align*}
\end{proof}

The same technique can be applied to other dimensionality reduction methods. 
For example, \cite{feldman2020turning} initially applies dimensionality reduction to project the given set of points into an $m$-dimensional subspace $L$, with $m = O(k/\varepsilon^2)$.
Subsequently, they construct an approximate coreset $S$ within this subspace $L$.
Using a similar proof procedure, we can show that the storage of the projection to the reduced-dimensional space is $\Omega\left(md\right) = \Omega\left(kd/{\varepsilon^2}\right)$.
{

\section{Application of Coreset Construction in Distributed and Streaming Settings}

In this section, we expand our compression scheme for coreset construction, as outlined in Algorithm~\ref{alg:upper}, to other well-studied contexts, including distributed and streaming settings (refer to \ref{subsec:distributed} and \ref{subsec:streaming} respectively). Within these settings, we provide the exact bit space complexity using our quantization scheme, demonstrating the versatility of our method.

\subsection{Communication Cost for Distributed \kzC}
\label{subsec:distributed}

In many practical applications, massive data is collected and stored
on a large number of nodes possibly deployed at different locations, while we want to learn properties of the union of the data.
For example, application data
from location based services\cite{schiller2004location}, images and videos over
networks\cite{mitra2011characterizing}. 
It has become increasingly important to develop effective clustering algorithms in distributed scenarios.
In such distributed systems and applications, communication cost is our major concern, since
communication is much slower than local computation.

Here we consider the distributed \kzC\ model introduced in \cite{balcan2013distributed}.
In this model, there is a set of $l$ sites  
$\calV=\left\{\bsv_i, 1 \leq i \leq l\right\}$, 
each holding a local data set $\bsP_i,$ $i=1,\ldots,l$.
%
These sites communicate through an undirected connected graph $\mathcal{G} = (\mathcal{V},\mathcal{E})$, where an edge $(\bsv_i,\bsv_j)\in \mathcal{E}$ indicates that sites $\bsv_i$ and $\bsv_j$ can communicate with each other.
Our goal is to construct an $\eps$-sketch for $\cup_{i=1}^l \bsP_i$ on a specified site, while keeping the communication efficient.
Previous research has primarily measured the communication cost in
number of points transmitted\cite{balcan2013distributed}. 
Our approach, however, focuses on minimizing the worst-case communication costs, i.e., the total number of bits exchanged.

As done in \cite{balcan2013distributed}, we consider the coordinator model introduced in~\cite{dolev1989multiparty}. The formal definition of our problem is provided in Definition~\ref{def:distributed_coordinator}. Similar results can be obtained for the general communication graphs using the Message-Passing algorithm proposed in \cite{balcan2013distributed} (See Algorithm 2 and Theorem 2 in \cite{balcan2013distributed}). The idea is to propagate messages on the graph in a breadth-first-search style so that all sites have a copy of the coreset at the end. 

\begin{definition}[Coreset for \kzC\ in the coordinator model]\label{def:distributed_coordinator}
Given integers $n,k\geq 1$, constant $z\geq 1$ and an error parameter $\eps\in (0,1)$.
Let there be $l$ sites each holding a private input data set $\bsP_i \subseteq [\Delta]^d$, and an additional site called coordinator.
Sites can only communicate with the coordinator.
The task of the coordinator is to collaborate with all sites to correctly output an $\eps$-sketch for $\cup_{i=1}^l \bsP_i$.
\end{definition}

Our objective is to minimize the communication cost defined in Definition~\ref{def:distributed}.

\begin{definition}[Coreset for \kzC\ in the distributed model and communication cost]\label{def:distributed}
We define the communication cost $\CC(\cup_{i=1}^l\bsP_i, \Delta, k, z, d, \eps)$ to be the minimum possible bits communicated by the sites to construct an $\eps$-sketch.
Moreover, we define $\CC(l, n, \Delta, k, z, d, \eps):= \sup_{\bsP_i\subseteq [\Delta]^d: \sum_{i=1}^l|\bsP_i| = n} \CC\left(\cup_{i=1}^l\bsP_i,\Delta, k, z, d, \eps\right)$ to be the communication complexity function, i.e., the maximum cardinality $\CC\left(\cup_{i=1}^l\bsP_i,\Delta, k, z, d, \eps\right)$ over all possible datasets $\cup_{i=1}^l\bsP_i\subseteq [\Delta]^d$ of size at most $n$.
\end{definition}

The idea for our algorithm is based on the mergeability of coresets, meaning that the union of coresets from multiple datasets forms a coreset for the combined datasets~\cite{balcan2013distributed}. Consequently, we begin by constructing a sketch for each site and then transmit these to the coordinator. The final sketch is then assembled by merging the results from each local dataset.


\begin{corollary}[Communication upper bounds for distributed Euclidean \kzC]
\label{cor:distributed}
    In the distributed \kzC\ problem, suppose for any dataset $\cup_{i=1}^l\bsP_i\subseteq [\Delta]^d$ such that $\sum_{i=1}^l|\bsP_i|=n$ and $|\bsP_i|>k,$ $i=1,\ldots,l$, there exists an $\eps$-coreset of $\bsP_i$ for \kzC\ of size at most $\Psi\left(|\bsP_i|\right) \geq 1$ on each site.
    Then the communication complexity to construct an $\eps$-sketch for $\cup_{i=1}^l \bsP_i$ is bounded by:
    $$\CC(l, n, \Delta, k, z, d, \eps) \leq O\left(l kd\log\Delta + \sum_{i=1}^l \Psi\left(|\bsP_i|\right)(d\log 1/\eps + d\log\log \Delta + \log\log n )\right).$$
\end{corollary}

\begin{proof}
We first let each site runs Algorithm~\ref{alg:upper}. Apply Theorem~\ref{thm:main_upper}, we have the bit complexity of the $\eps$-sketch for each local dataset is 
$$\spc(\bsP_i, \Delta, k, z, d, \eps) \leq O\left(kd\log\Delta +  \Psi\left(|\bsP_i|\right)(d\log 1/\eps + d\log\log \Delta + \log\log n )\right).$$
Each site will transmit its sketch to the coordinator, who will then combine these sketches to obtain the final result. The communication cost for transmitting these sketches is 
\begin{align*}
\CC(l, n, \Delta, k, z, d, \eps) &\leq \sum_{i=1}^l \spc(\bsP_i, \Delta, k, z, d, \eps) \\&\leq  O\left(l kd\log\Delta + \sum_{i=1}^l \Psi\left(|\bsP_i|\right)(d\log 1/\eps + d\log\log \Delta + \log\log n )\right). 
\end{align*}

\end{proof}

Combining with the recent breakthroughs that shows that for any $|\bsP_i|>k$, $\Psi(|\bsP_i|) = \tilde{O}\left(\min\left\{k^{\frac{2z+2}{z+2}} \eps^{-2}, k \eps^{-z-2}\right\}\right)$ \cite{cohen2021new,cohen2022towards,cohenaddad2022improved,huang2023optimal}, we conclude that
\begin{align*}
\label{eq:upper_bound}
\CC(l, n, \Delta, k, z, d, \eps) \leq \tilde{O}\left(ld\cdot \min\left\{\frac{k^{\frac{2z+2}{z+2}}}{\eps^2}, \frac{k}{\eps^{z+2}} \right\}\right).
\end{align*}

\subsection{Space Complexity for Streaming \kzC}
\label{subsec:streaming}

Modern datasets have significantly increased in size, often consisting of hundreds of millions of points, which poses great challenges for analyzing
them. In typical applications, the total volume of data is very large and can not be stored in its entirety. Over the last decade, the streaming model has proven to be successful in dealing with big data \cite{muthukrishnan2005data}. In this model, the input data arrive sequentially and we usually require a data structure using limited working space compared with the huge volume of the data. Our major concern is the storage cost, which is the bit complexity of storing such a data structure. We consider the insertion-only streaming model formally defined in Definition~\ref{def:streaming}.

\begin{definition}[Insertion-Only Streaming \kzC]\label{def:streaming}
Given integers $n,k\geq 1$, constant $z\geq 1$ and an error parameter $\eps\in (0,1)$.
Suppose a stream consists of $n$ point $\bsp_1,\ldots,\bsp_n \in [\Delta]^d$ that arrive sequentially.
The goal is to maintain an $\eps$-sketch for the stream at every point while using limited bits.
\end{definition}

There is a growing body of work studying Euclidean \kzC\  problems over data streams. However, existing studies mainly focus on the size of the coreset~\cite{har2004coresets,braverman2019streaming} or assume a single word for storing the coordinate and weight~\cite{cohen2023streaming}. In contrast, our approach focuses on minimizing the worst-case bit complexity. By using Algorithm~\ref{alg:upper}, we get Corollary~\ref{cor:stream}.


\begin{corollary}[Space upper bounds for streaming Euclidean \kzC]
\label{cor:stream}
In the streaming \kzC\ problem, suppose for any stream consists of $n$ point $\bsp_1,\ldots,\bsp_n \in [\Delta]^d$, there exists an $\eps$-coreset of the stream for \kzC\ using at most $\Phi(n)$ words.
When $n > k$, the bits needed for storage is upper bounded by: $$O\left(kd\log\Delta + \Phi(n) \left(\log 1/\eps + \log\log \Delta + \frac{1}{d}\log\log n \right)\right).$$
\end{corollary}



Combining with the result from \cite{braverman2019streaming} which constructed a streaming coreset using $\Phi(n) = O(\eps^{-2} kd (\log k\log n+\log 1/\delta))$ words with probability at least $1-\delta$, we obtain that the required number of bits is bounded by $$O\left(kd\left(\log\Delta +  \eps^{-2}  \left(\log k\log n+\log \frac{1}{\delta}\right) \left(\log 1/\eps + \log\log \Delta + \frac{1}{d}\log\log n \right) \right)\right).$$

}\color{black}

\section{Conclusions and Future Work}

In this study, we initiate the exploration of space complexity for the Euclidean \kzC\ problem, presenting both upper and lower bounds. 
Our findings suggest that a coreset serves as the optimal compression scheme when $k$ is constant. 
Furthermore, the space lower bounds for \kzC\ directly imply a tight space lower bound for terminal embedding when $d\geq \Omega(\frac{\log n \log(n/\eps)}{\eps^2})$.
The techniques we employ for establishing these lower bounds contribute to a deeper geometric understanding of principal angles, which may be of independent research interest.

Our work opens up several interesting research directions. 
One immediate challenge is to further narrow the gap between the upper and lower bounds of the space complexity for Euclidean \kzC. 
Additionally, it would be valuable to investigate whether a coreset remains optimal for compression when $k$ is large. 

\bibliography{ref}
\bibliographystyle{plainnat}

\appendix

\section{Missing Proofs of Lemmas \ref{lmm:less_than_identity} and \ref{lmm:inequation_for_arbitrary_power}}\label{appendix}

{
\begin{lemma}[Restatement of Lemma~\ref{lmm:less_than_identity}]Given real numbers $e,f,g,q$ and $0<h<1$, the ${ }_2 F_1$ function satisfies that 
\[{ }_2 F_1\left(e, f ; g ; (1-h)\bsI_q\right) \leq { }_2 F_1\left(e, f ; g ; \bsI_q\right). \]
\end{lemma}}

\begin{proof}

{
We have the definition of the Gaussian hypergeometric function
of matrix argument.}

{
\begin{definition}[Definition 7.3.1 in \cite{muirhead2009aspects}]\label{def:original_gaussian_hyper} The Gaussian hypergeometric function of matrix argument is given by
\[
{ }_2 F_1\left(e,f ; g ; \bsX\right)=\sum_{k=0}^{\infty} \sum_\kappa \frac{\left(e\right)_\kappa \left(f\right)_\kappa}{\left(g\right)_\kappa} \frac{C_\kappa(\bsX)}{k !},
\]
where $\sum_\kappa$ denotes summation over all partitions $\kappa=\left(k_1, \ldots, k_m\right), k_1 \geq \cdots \geq$ $k_m \geq 0$, of $k, C_\kappa(X)$ is the zonal polynomial of $X$ corresponding to $\kappa$ and the generalized hypergeometric coefficient $(a)_\kappa$ is given by
\[
(a)_\kappa=\prod_{i=1}^m\left(a-\frac{1}{2}(i-1)\right)_{k_i},
\]
where $(a)_k=a(a+1) \ldots(a+k-1),(a)_0=1$. Here $X$, the argument of the function, is a complex symmetric $q \times q$ matrix, and the parameters $e,f,g$ are arbitrary real numbers. Denominator parameter $g$ is not allowed to be zero or an integer or half-integer $\leqslant \frac{1}{2}(m-1)$.
\end{definition}}

{
From the definition, we can see that the only term involving the matrix is the zonal polynomial $C_\kappa(X)$. The value of it is defined in Definition~\ref{def:zonal}.}

{
\begin{definition}[Equation 13 in \cite{muirhead2009aspects}]\label{def:zonal} Let $x_1,\cdots,x_q$ be the eigenvalues of $\bsX$. If the partition $\lambda=\left(l_1, \ldots, l_m\right), l_1 \geq \cdots \geq l_m \geq 0$, the monomial symmetric functions is defined as 
$$M_{\lambda}(\bsX) = \sum \cdots \sum x_{i_1}^{l_1}x_{i_2}^{l_2}\cdots x_{i_p}^{l_p},$$ 
where $p$ is the number of nonzero parts in the partition $\lambda$ and the summation is over the distinct permutations $(i_1,i_2,\cdots,i_p)$ of $p$ different integers fromthe integers $1,\cdots,q$. Then for some constants $c_{\kappa,\lambda} \geq 0$, the value of zonal polynomial is 
$$C_\kappa(X) = \sum_{\lambda \leq \kappa}c_{\kappa,\lambda} M_{\lambda}(\bsX)$$
\end{definition}}

{
Now come back to our setting. We have all the eigenvalues of $(1-h)\bsI_q$ are $(1-h)$,  which are less than the eigenvalues of $\bsI$, whose eigenvalues are 1. Therefore, we must have that for any partition $\lambda,\kappa$, 
$$M_{\lambda}((1-h)\bsI_q) \leq M_{\lambda}(\bsI_q), C_{\kappa}((1-h)\bsI_q) \leq  C_{\kappa}(\bsI_q).$$
Therefore, we would have our desired bound, 
\begin{align*}
{ }_2 F_1\left(e,f ; g ; (1-h)\bsI_q\right)&=\sum_{k=0}^{\infty} \sum_\kappa \frac{\left(e\right)_\kappa \left(f\right)_\kappa}{\left(g\right)_\kappa} \frac{C_\kappa((1-h)\bsI_q)}{k !} \\&\leq \sum_{k=0}^{\infty} \sum_\kappa \frac{\left(e\right)_\kappa \left(f\right)_\kappa}{\left(g\right)_\kappa} \frac{C_\kappa(\bsI_q)}{k !} \\&= { }_2 F_1\left(e,f ; g ; \bsI_q\right).
\end{align*}

}

\end{proof}

{
\begin{lemma}[Restatement of Lemma~\ref{lmm:inequation_for_arbitrary_power}]
For any $0 < z \leq 2$ and any $x \in  \left[0, \frac{1}{2}\right]$, we have
\[1-{\frac{z}{2}}x - z\left(1-{\frac{z}{2}}\right)x^2 \leq (1-x)^{\frac{z}{2}} \leq  1-{\frac{z}{2}}x  . \]
For any $ z\geq 2$ and any $x \in  \left[0, \frac{1}{2}\right]$, we have
\[1-{\frac{z}{2}}x \leq (1-x)^{\frac{z}{2}} \leq   1-{\frac{z}{2}}x + {\frac{z}{2}}\left({\frac{z}{2}}-1\right)x^2. \]
\end{lemma}}

\begin{proof}
We first deal with the case when $0 < z \leq 2$. 
Note that the right hand of the inequalities can actually be found in \cite{cohen2022towards}. 
For any $0<z \leq 2$ and any $x \in  \left[0, \frac{1}{2}\right]$, we have $(1-x)^{\frac{z}{2}}= \exp \left(-{\frac{z}{2}} \sum_{n=1}^{\infty}(x)^n / n\right)$. 
Since $z \leq 2$, this is at most $\exp \left(-\sum_{n=1}^{\infty}\left({\frac{z}{2}}x\right)^n / n\right)=1-{\frac{z}{2}}x$. 
For the other side, it is equal for us to prove 
\[\h(x) = z\left(1-{\frac{z}{2}}\right)x^2 + {\frac{z}{2}}x - 1 + (1-x)^{\frac{z}{2}} \geq 0, \forall x \in \left[0,\frac{1}{2}\right].\]
Note that the first and second derivative of $\h(x)$ is 
\begin{align*}
&  \h^{\prime}(x) = 2z\left(1-{\frac{z}{2}}\right) x + {\frac{z}{2}}  - {\frac{z}{2}} (1-x)^{{\frac{z}{2}}-1}, \\ 
& \h^{\prime\prime}(x) = 2z\left(1-{\frac{z}{2}}\right) - {\frac{z}{2}}\left(1-{\frac{z}{2}}\right) (1-x)^{{\frac{z}{2}}-2}.
\end{align*}
since $x \in \left[0,\frac{1}{2}\right]$, we must have $(1-x)^{{\frac{z}{2}}-2} \leq 2^{2-{\frac{z}{2}}} \leq 4$. We have $\h^{\prime\prime}(x) \geq 0,\forall x \in \left[0,\frac{1}{2}\right] $ and consequently 
\begin{align*}
&\h^{\prime}(x)\geq \h^{\prime}(0)=0, \forall x \in \left[0,\frac{1}{2}\right],\\&\h(x)\geq \h(0)=0, \forall x \in \left[0,\frac{1}{2}\right].
\end{align*}

The case when $z\geq 2$ is rather similar. 
The left hand of the inequality can again be found in \cite{cohen2022towards}. 
For any $z\geq 2$ and any $x \in  \left[0, \frac{1}{2}\right]$, we have $(1-x)^{\frac{z}{2}}= \exp \left(-{\frac{z}{2}} \sum_{n=1}^{\infty}(x )^n / n\right)$. 
Since $z \geq 2$, this is at least $\exp \left(-\sum_{n=1}^{\infty}\left({\frac{z}{2}}x\right)^n / n\right)=1-{\frac{z}{2}}x$.

For the other side, it is equal for us to prove 
\[\h(x) = -{\frac{z}{2}}\left({\frac{z}{2}}-1\right)x^2 + {\frac{z}{2}}x - 1 + (1-x)^{\frac{z}{2}} \leq 0, \forall x \in \left[0,\frac{1}{2}\right].\]
Note that the first and second derivative of $\h(x)$ is 
\begin{align*}
& \h^{\prime}(x) = -{\frac{z}{2}}\left({\frac{z}{2}}-1\right) x + {\frac{z}{2}}  - {\frac{z}{2}} (1-x)^{{\frac{z}{2}}-1},\\
& \h^{\prime\prime}(x) = -{\frac{z}{2}}\left({\frac{z}{2}}-1\right) + {\frac{z}{2}}\left({\frac{z}{2}}-1\right) (1-x)^{{\frac{z}{2}}-2}.
\end{align*}
since $x \in \left[0,\frac{1}{2}\right]$, we must have $(1-x)^{{\frac{z}{2}}-2} \leq \max\{1,2^{2-{\frac{z}{2}}}  \}\leq 2$. We have $\h^{\prime\prime}(x) \leq 0,\forall x \in \left[0,\frac{1}{2}\right] $ and consequently 
\begin{align*}
&\h^{\prime}(x)\leq \h^{\prime}(0)=0, \forall x \in \left[0,\frac{1}{2}\right],\\&\h(x)\leq \h(0)=0, \forall x \in \left[0,\frac{1}{2}\right].
\end{align*}
\end{proof}

\end{document}